\documentclass{article}
\usepackage{graphicx} 

\usepackage{amssymb,amsmath,amsthm}

\theoremstyle{plain}
\newtheorem{theorem}{Theorem}[section]
\newtheorem{corollary}[theorem]{Corollary}
\newtheorem{lemma}[theorem]{Lemma}

\newtheorem{claim}[theorem]{Claim}
\newtheorem{inductive-invariant}[theorem]{Inductive Invariant}
\newtheorem{SP}[theorem]{Step-property}

\theoremstyle{definition}
\newtheorem{definition}[theorem]{Definition}
\newtheorem{example}[theorem]{Example}

\newtheorem{remark}[theorem]{Remark}
\newtheorem{observation}[theorem]{Observation}

\usepackage[inline]{enumitem}
\usepackage{hyperref}
\usepackage{tikz}
\usepackage{caption}
\usepackage{subcaption}
\usepackage{alltt}
\usepackage{multirow}

\usepackage[margin=1.4in]{geometry}

\usepackage{authblk}

\usetikzlibrary{shapes.geometric, arrows, automata, positioning, calc, decorations.pathmorphing}
\tikzstyle{startstop} = [rectangle, rounded corners, minimum height=1cm, minimum width=3cm, text centered, draw=black, fill=red!30]
\tikzstyle{process} = [rectangle, text centered, minimum height=1cm, text width=3cm, draw=black, fill=orange!30]
\tikzstyle{decision} = [circle, text centered, text width=1.5cm, draw=black, fill=blue!20]
\tikzstyle{external} = [ellipse, text centered, text width=3cm, draw=black, minimum height=1cm, fill=green!30]
\tikzstyle{arrow} = [thick,->,>=stealth]

\def\ImportFromMnSymbol#1{%
  \DeclareFontFamily{U} {MnSymbol#1}{}
  \DeclareFontShape{U}{MnSymbol#1}{m}{n}{
   <-6> MnSymbol#15
   <6-7> MnSymbol#16
   <7-8> MnSymbol#17
   <8-9> MnSymbol#18
   <9-10> MnSymbol#19
   <10-12> MnSymbol#110
   <12-> MnSymbol#112}{}
  \DeclareFontShape{U}{MnSymbol#1}{b}{n}{
   <-6> MnSymbol#1-Bold5
   <6-7> MnSymbol#1-Bold6
   <7-8> MnSymbol#1-Bold7
   <8-9> MnSymbol#1-Bold8
   <9-10> MnSymbol#1-Bold9
   <10-12> MnSymbol#1-Bold10
   <12-> MnSymbol#1-Bold12}{}
  \DeclareSymbolFont{MnSy#1} {U} {MnSymbol#1}{m}{n}
}
\newcommand\DeclareMnSymbol[4]{\DeclareMathSymbol{#1}{#2}{MnSy#3}{#4}}
\ImportFromMnSymbol{C}
\DeclareMnSymbol{\DDiamond}{\mathrel}{C}{120}

\newcommand{\rom}[1]{\mbox{\rm #1}}

\newcommand{\arrow}{\mbox{$\mapsto$}}
\newcommand{\notarrow}{\mbox{$\not\hspace{-1mm}\arrow$}}

\newcommand{\LeftDes}{\rm LeftDes}
\newcommand{\RightDes}{\rm RightDes}
\newcommand{\Des}{\rm Des}

\newcommand{\booleanT}{\mbox{\tt boolean}}
\newcommand{\TRUE}{\mbox{$\tt true$}}
\newcommand{\FALSE}{\mbox{$\tt false$}}
\newcommand{\goto}{\mbox{\tt{goto}}}
\newcommand{\pre}{\rm pr}
\newcommand{\LR}{\mbox{\tt LR}}
\newcommand{\IF}{\mbox{\tt{if}}}
\newcommand{\THEN}{\mbox{\tt{then}}}
\newcommand{\ELSE}{\mbox{\tt{else}}}
\newcommand{\newnode}{\mbox{\sf new}}
\newcommand{\return}{\mbox{\tt return}}
\newcommand{\lock}{\mbox{\sf wait\_lock}}

\newcommand{\key}{\mbox{\tt{key}}}
\newcommand{\Left}{\mbox{\tt{left}}}
\newcommand{\Right}{\mbox{\tt{right}}}
\newcommand{\del}{\mbox{\tt{del}}}
\newcommand{\rem}{\mbox{\tt{rem}}}
\newcommand{\Root}{\mbox{\tt{root}}}

\newcommand{\bfcontains}{\mbox{\sf contains}}
\newcommand{\bfdelete}{\mbox{\sf delete}}
\newcommand{\bfinsert}{\mbox{\sf insert}}
\newcommand{\rotateLeft}{\mbox{\sf rotateLeft}}
\newcommand{\rotateRight}{\mbox{\sf rotateRight}}
\newcommand{\bfremove}{\mbox{\sf remove}}

\newcommand{\AddressT}{\mbox{\tt Adrs}}
\newcommand{\KeyT}{\mbox{\tt Key}}
\newcommand{\InstructionT}{\mbox{\tt Instrc}}
\newcommand{\ProcT}{\mbox{\tt Proc}}

\newcommand{\Sys}{\mbox{\it Sys}}

\newcommand{\Set}{\mbox{$\mathit Set$}}
\newcommand{\Step}{\mbox{\it Step}}

\newcommand{\DeletedP}{\mbox{\it Del}}
\newcommand{\RemovedP}{\mbox{\it Rem}}
\newcommand{\KeyF}{\mbox{\it Key}}      
\newcommand{\LeftF}{\mbox{\it Left}}
\newcommand{\RightF}{\mbox{\it Right}}
\newcommand{\ControlF}{\mbox{\it Ctrl}}
\newcommand{\LockedP}{\mbox{\it Lock}} 
\newcommand{\preRemovedP}{\mbox{$\mathit preRemoved$}}

\newcommand{\prt}{\mbox{$\mathit{prt}_0$}}
\newcommand{\child}{\mbox{$\mathit{chd}_0$}}
\newcommand{\lft}{\mbox{$\mathit{lft}_0$}}
\newcommand{\new}{\mbox{$\mathit new$}}
\newcommand{\n}{\mbox{$\mathit{nd}_0$}}
\renewcommand{\r}{{r_0}}
\newcommand{\rl}{\mbox{$r\ell_0$}}
\newcommand{\lr}{\mbox{$\ell r_0$}}
\newcommand{\nd}{\mbox{$\mathit nd$}}
\newcommand{\nxt}{\mbox{$\mathit nxt$}}

\newcommand{\oM}{\mbox{$\overline{M}$}}
\newcommand{\returnVal}{\mbox{$\mathit returnVal$}}
\newcommand{\Begin}{\mbox{$\mathit begin$}}
\newcommand{\End}{\mbox{$\mathit end$}}
\newcommand{\search}{\mbox{\rm search}}

\newcommand{\possible} {\mbox{\it Possible changes only in:}}

\title{Linearizability Analysis of the Contention-Friendly Binary Search Tree}
\author{Uri Abraham}
\author{Avi Hayoun}
\affil{Ben-Gurion University of the Negev, Beer-Sheva, Israel}
\date{November 2022}

\begin{document}

\maketitle


\begin{abstract}
We present a formal framework for proving the correctness of set implementations backed by binary-search-tree (BST) and linked lists, which are often difficult to prove correct using automation. This is because many concurrent set implementations admit non-local linearization points for their `contains' procedure. We demonstrate this framework by applying it to the Contention-Friendly Binary-Search Tree algorithm of Crain et al~\cite{Crain13,Crain16}.

We took care to structure our framework in a way that can be easily translated into input for model-checking tools such as TLA+, with the aim of using a computer to verify bounded versions of claims that we later proved manually. Although this approach does not provide complete proof (i.e., does not constitute full verification), it allows checking the reasonableness of the claims before spending effort constructing a complete proof. This is similar to the test-driven development methodology, that has proven very beneficial in the software engineering community.

We used this approach and validated many of the invariants and properties of the Contention-Friendly algorithm using TLA+~\cite{engberg1992TLA}. It proved beneficial, as it helped us avoid spending time trying to prove incorrect claims. In one example, TLA+ flagged a fundamental error in one of our core definitions. We corrected the definition (and the dependant proofs), based on the problematic scenario TLA+ provided as a counter-example.

Finally, we provide a complete, manual, proof of the correctness of the Contention-Friendly algorithm, based on the definitions and proofs of our two-tiered framework.
\end{abstract}

\section{Introduction}
Highly concurrent algorithms are often considered difficult to design and implement correctly. The large number of possible ways to execute such algorithms, brought about by the high degree of inter-process interference, translates into complexity for the programmer.

Linearizability~\cite{herlihy1990linearizability} is the accepted correctness condition for the implementation of concurrent data structures. It implies that each operation of a data structure implementation can be regarded as executing instantly at some point in time, known as the {\em linearization point} of the operation that is located between the initialization of the operation and its response (ending). This causes the operation to behave atomically for other concurrent operations. Although the definition of linearizability is intuitively simple, its proofs is usually complex.

The concurrent set data structure is particularly interesting in proving linearizability: many concurrent set algorithms are common examples of implementations in which some data operations have non-local linearization points~\cite{ohearn2010hindsight,zhu2015poling}. That is, the linearization point of an operation by process $p$ may be an event generated by the action of another process.

Crain et al.~\cite{Crain13,Crain16} presented an elegant and efficient concurrent, lock-based, contention-friendly, binary-search tree (BST), that provides the standard set interface operations. \bfcontains\ queries for the presence of a value $k$ (a key value) in the set; \bfdelete\ performs a logical deletion of a value $k$ (by changing the status of an address with value $k$ from undeleted to deleted); 
and \bfinsert\ performs either a logical insertion (by changing the status of an address with
value $k$ from deleted to undeleted) or a physical insertion (by adding a new address with value $k$, when
no address with key $k$ exists).
The main features of the contention-friendly algorithm are a self-balancing mechanism (the \rotateLeft\ and \rotateRight\ operations), and a physical removal procedure (the \bfremove\ operation), which help approximate the big O guarantees of a sequential BST implementation. The authors of the contention-friendly
algorithm provided experimental evidence of the efficiency of their approach, which is of prime importance,
but our work deals with correctness rather than the algorithm's efficiency. We present the full details of the algorithm itself in Section~\ref{sec:the-algorithm}.

Proofs of the correctness of variants of this contention-friendly algorithm have been proposed~\cite{Feldman18,Feldman20}. 
However, these variants forgo certain core behavioral aspects of the algorithm, specifically, that the backtracking mechanism of the original algorithm is not fully realized. 
In this chapter, we provide complete proof of the algorithm. More precisely, we do not deal here with the original algorithm of \cite{Crain13,Crain16},
but rather, with a simpler version that retains the original backtracking behavior. 
We acknowledge that it may seem strange to devote more than 50 pages to the proof of
an algorithm that is implemented using fewer than 20 instructions, but an important aim of Section~\ref{Sec:exploring-an-example} is 
to explain why such is necessary, using illustrations.

Some formal correctness proofs are quite easy to follow despite their formidable length 
because they are guided by relatively simple and intuitive arguments that motivate each step. 
That does not seem to be the case for the contention-friendly algorithm, and it was unclear to us at first how to approach the proof. 
In fact, we were initially unable to identify and formally characterize the states that the algorithm executions generate. 
A mathematical definition of these states is necessary to support the definition of invariants, 
which are the basic ingredients of any correctness proof. 
It is customary to define states of a memory system as functions from memory locations to a set of possible values. Still, it was evident at an early stage that such simple states would not do, 
and a richer language and corresponding structures are needed to reflect the subtleties of the states of the algorithm. 
The need for richer structures will be illustrated by the scenario examples described in Section ~\ref{Sec:exploring-an-example}.

We begin our journey in Section~\ref{sec:perlims} by rigorously defining the formal framework that we use
here and that we believe can be generalized to many other BST and Set implementations.
This framework is based on a model-theoretic approach to defining program states, as proposed by Abraham~\cite{abraham1999models}, 
and it is used here to formulate and prove inductive invariants and properties of steps specific to the algorithm in question (see Section~\ref{sec:invariants-and-props}).

The proof approach developed and presented here has two parts. The first part consists of defining and proving
invariants and properties of states. In the second part, we focus on properties of the histories, which
are the structures that describe executions of the algorithm. The reason for this two-step approach is that a state
 is a different object from an execution of the algorithm; an execution consists of a sequence of states,
that is, a history, and the correctness of the algorithm is a property of histories, not of states.

It is in the second part of the proof that we formulate and prove the central theorem of our work, the Scanning Theorem (see Section~\ref{SecHistory}), with the help of the theory of states developed in the first part.
This theorem is largely disconnected from the technical, low-level details of the algorithm in question, and is abstracted away from the model-theoretic framework developed and used in the previous sections. This abstraction turned out to be very powerful, and it greatly simplifies the final section of this study, proving the correctness of the contention-friendly algorithm.

The rest of this study is structured as follows. In Section~\ref{sec:the-algorithm}, we present the technical details of the contention-friendly algorithm and observe some interesting aspects of its behavior. Then, we lay out the formal foundation of our proof system in Section~\ref{sec:address-structure}, by presenting the logical language we used to formally define program states. In Section~\ref{sec:invariants-and-props}, we use this language to formulate claims about the states of the algorithm and the relationships between them. This section culminates in the definition of the notion of {\em regularity}, which is crucial for the next steps in the proof. In Section~\ref{SecHistory}, we present and prove the Scanning Theorem, followed by the correctness proof of the contention-friendly algorithm in Section~\ref{Sec6}. Finally, we survey some related work and conclude.

\section{The Contention-Friendly algorithm}
\label{sec:the-algorithm}

The contention-friendly (CF) binary-search tree~\cite{Crain13,Crain16} is a lock-based concurrent binary-search tree (BST) that implements the classic set interface of set \bfinsert/\bfdelete/\bfcontains operations. 
Each of its nodes contains the following fields: a key \key; \Left\ and \Right\ pointers to the left and right child nodes, respectively; a boolean \del\ flag indicating
if the node has been logically deleted; and a boolean \rem\ flag indicating if the node has been physically removed.

In Figures \ref{LM} and \ref{RLred}, we present a slightly modified version of the CF algorithm. Nonetheless, this version retains the core principles of the original algorithm and, most importantly, the backtracking mechanism.

\begin{figure}
\centering
\fbox{
\begin{minipage}[t]{\textwidth}
\begin{tabbing}
*\=****\=**\=**\=**\=**\=\kill
Master program for a process $p>0$: \\ 
\\
\> m0.\> set $k_p\in\omega$,  $\nd_p=\Root$, $\nxt_p=\Root$, and \goto\ c1, d1, or i1.  
\\ 
\par\noindent\rule{0.95\textwidth}{0.4pt}
\\
Master program for the system process \Sys:\\ 
\\
\> m0. set $\prt\in\AddressT$, $\lft\in \booleanT$, establish the prerequisites, and \goto\ f6, r6, or v6.
\end{tabbing}
\end{minipage}
}
\caption{Master Program for the working and system processes.}
\label{FigMP}
\end{figure}

\begin{figure}
\centering
\begin{tabular}{|l|}\hline
\begin{minipage}[t]{0.9\textwidth}
\begin{tabbing}
*\=****\=***\=***\=***\=\kill
 \booleanT\  $\rotateLeft(\prt,\, \lft)$: \\[0.3cm]

\> \pre 1. \>\ $\neg \RemovedP(\prt)\wedge \prt\neq \bot$\\[0.1cm]

\> \pre 2. \>\  $\n =\LR(\prt,\lft)\wedge \n\notin\{\Root,\bot,\prt\}\wedge\neg\RemovedP(\n)$\\[0.1cm]

\> \pre 3. \>\ $\r = \RightF(\n) \wedge \r \neq \bot \wedge  \neg \RemovedP(\r)$\\[0.1cm] 

\> \pre 4. \>\ $\LockedP(\prt,\Sys);\ \LockedP(\n,\Sys);\ \LockedP(\r,\Sys)$\\[0.1cm]
 
\> \pre 5.\>\ $\rl = \LeftF(\r)$;\ $\ell_0 = \LeftF(\n)$\\[0.1cm] 

--------------------------------- \ \\
 \>f6. \>\ $\r.\Left := \newnode(\n.\key,\n.\del, \FALSE, \ell_0, \rl)$ \\[0.1cm]

\> f7.\>\ $\n.\Left := \r$\\[0.1cm]

\> f8.\>\ \IF\ \lft\ \hspace{1mm} \THEN\ \hspace{1mm} $\prt.\Left := \r$ \hspace{1mm} \ELSE\ \hspace{1mm} $\prt.\Right := \r$\\[0.1cm]

\> f9. \>\ $\n.\rem := \TRUE$\\
\end{tabbing}
\end{minipage}\\
\hline
\begin{minipage}[t]{0.9\textwidth} 
\begin{tabbing}
*\=****\=***\=***\=\kill
 \booleanT\  $\rotateRight(\prt,\, \lft)$: \\[0.3cm]

\> \pre 1. \>\ $\neg \RemovedP(\prt)\wedge \prt\neq \bot$\\[0.1cm]

\> \pre 2. \>\  $\n =\LR(\prt,\lft)\wedge \n\notin\{\Root,\bot,\prt\}\wedge\neg\RemovedP(\n)$\\[0.1cm]

\> \pre 3. \>\ $\ell_0 = \LeftF(\n)\wedge \ell_0\neq \bot \wedge \neg\RemovedP(\ell_0)$\\[0.1cm]
\> \pre 4. \>\ $\LockedP(\prt,\Sys);\ \LockedP(\n,\Sys);\ \LockedP(\ell_0,\Sys)$\\[0.1cm]

\> \pre 5. \>\ $\lr = \RightF(\ell_0)$;\ $\r= \RightF(\n)$\\[0.1cm]
--------------------------------- \ \\ 

\> r6. \>\ $\ell_0.\Right := \newnode(\n.\key,\n.\del, \FALSE, \lr, \r)$ \\[0.1cm]

\> r7. \>\  $\n.\Right := \ell_0$\\[0.1cm]

\> r8. \>\ \IF\ \lft\ \hspace{1mm} \THEN\ \hspace{1mm} $\prt.\Left := \ell_0$ \hspace{1mm} \ELSE\ \hspace{1mm} $\prt.\Right := \ell_0$\\[0.1cm]

\>r9.\> $\n.\rem := \TRUE$\\
\end{tabbing}
\end{minipage}\\
\hline
\begin{minipage}[t]{0.9\textwidth}
\begin{tabbing}
*\=****\=***\=***\=\kill
\booleanT\ $\bfremove(\prt, \lft)$:\\[0.3cm]

\> \pre 1. \>\ $\neg \RemovedP(\prt)\wedge \prt\neq \bot$\\[0.1cm]

\> \pre 2. \>\ $\n =\LR(\prt,\lft)\wedge \n\notin\{\Root,\bot,\prt\}\wedge\neg\RemovedP(\n)$\\[0.1cm]

\> \pre 3. \>\  $\LockedP(\prt,\Sys);\ \LockedP(\n,\Sys)$\\[0.1cm]
\> \pre 4. \>\  $ \DeletedP(\n) \wedge$ $(\LeftF(\n)= \bot \vee \RightF(\n) = \bot)$\\[0.1cm]
\> \pre 5. \>\ $\LeftF(\n) \neq \bot\rightarrow \child = \LeftF(\n)\ \wedge \LeftF(\n) = \bot\rightarrow \child = \RightF(\n)$\\[0.1cm]

--------------------------------- \ \\

\> v6.\>\ \IF\ $\lft$ \hspace{1mm} \THEN\ \hspace{1mm} $\prt.\Left := \child$ \hspace{1mm} \ELSE\ \hspace{1mm} $\prt.\Right := \child$\\[0.1cm]

\> v7.\>\ \IF\ $\n.\Left=\bot$ \hspace{1mm} \THEN\ \hspace{1mm} $\n.\Left := \prt$ \hspace{1mm} \ELSE\ \hspace{1mm} $\n.\Right :=\prt$\\[0.1cm]

\> v8. \>\ \IF\ $\n.\Left =\prt$ \hspace{1mm} \THEN\ \hspace{1mm} $\n.\Right :=\prt$ \hspace{1mm} \ELSE\ \hspace{1mm} $\n.\Left:=\prt$\\[0.1cm]

\> v9. \>\ $\n.\rem := \TRUE$\\

\end{tabbing}
\end{minipage}




\\
\hline
\end{tabular}

\caption{The rotation and removal operations of the contention-friendly algorithm.}
\label{LM}
\end{figure}

\begin{figure}
\centering
\begin{tabular}{|l|l|}\hline
\begin{minipage}[t]{65mm}
\begin{tabbing}
*\=***\=***\=**\=**\=**\=\kill
 \booleanT\  $\bfcontains(k)$: \\ \\

\>c1.\>\ \IF\ $\nxt=\bot$ \THEN\ \return\ \FALSE \\
\>\>\ $\nd := \nxt$ \\
\>\>\  \IF\ $k=\nd.key$ \THEN\ \goto\ c2\\
\>\>\  $\nxt := \LR(\nd, k<\nd.\key)$\\ 
\>\>\ \goto\ c1 \\

\\
\>c2.\>\  $\return\ \neg\nd.\del$ \\
\end{tabbing}
\end{minipage}
&
\multirow{25}{*}{
\begin{minipage}[t]{75mm}
\begin{tabbing}
**\=***\=***\=***\=***\=***\= \kill
 \booleanT\  $\bfinsert(k)$: \\ \\

\> i1.\>\  \IF\ $\nxt = \bot$ \THEN\ \goto\ i3\\
\>\>\ $\nd := \nxt$\\
\>\>\ \IF\ $k =\nd.\key$ \THEN\ \goto\ i2\\
\>\>\  $\nxt := \LR(\nd,  k< \nd.\key)$\\
\>\>\  \goto\ i1\\
\\

\> i2.\>\  $\lock(\nd)$ \\  
\>\>\ \IF\ $\neg\nd.\del$ \THEN\ \return\ \FALSE\\
\>\>\ \IF\ $\nd.\rem$ \THEN\\
\>\>\>\  $\nxt:= \nd.\Right$  \\ 
\>\>\>\  \goto\ i1\\
\>\>\ $\nd.\del := \FALSE$\\ 
\>\>\ \return\ \TRUE\\  
\\

\> i3.\>\ $\lock(\nd)$\\
\>\>\ \IF\ $\LR(\nd, k<\nd.\key )\neq \bot$ \THEN\\
\>\>\>\ $\nxt := \LR(\nd,k<\nd.\key)$ \\
\>\>\>\ \goto\ i1 \\

\>\>\ \IF\ $k <  \nd.\key$ \THEN\\
\>\>\>\ $\nd.\Left := \newnode(k, \FALSE, \FALSE, \bot, \bot)$\\
\>\>\ \ELSE\\
\>\>\>\ $\nd.\Right := \newnode(k, \FALSE, \FALSE, \bot, \bot)$\\

\>\>\ \return\ \TRUE \\
\end{tabbing}
\end{minipage}}\\
\cline{1-1}

\begin{minipage}[t]{65mm}
\begin{tabbing}
*\=***\=***\=**\=**\=**\= \kill
\booleanT\  $\bfdelete(k)$: \\ \\

\> d1.\>\  \IF\ $\nxt = \bot$ \THEN\ \return\ \FALSE\\
\>\>\ $\nd := \nxt$ \\
\>\>\ \IF\ $k= \nd.\key$ \THEN\ \goto\ d2  \\
\>\>\ $\nxt := \LR(\nd, k<\nd.\key)$   \\
\>\>\ \goto\ d1\\
 
\\
\> d2.\>\ $\lock(\nd)$   \\
\>\>\ \IF\ $\nd.\del$ \THEN\ \return \ \FALSE\\
\>\>\ \IF\ $\nd.\rem$ \THEN \\
\>\>\>\ $\nxt := \nd.\Right$\\
\>\>\>\ \goto\ d1 \\

\>\>\ $\nd.\del := \TRUE$\\
\>\>\ \return\, \TRUE  \\


\end{tabbing}
\end{minipage} 
& \\
\hline
\end{tabular}

\caption{The Data Operations of the contention-friendly algorithm. 
  }
\label{RLred}
\end{figure}

The main difference between the original version and this modified version of the algorithm is in the \rotateLeft\ and \rotateRight\ operations of the \Sys\ process. In the original version, the \new\ node allocated by the operations is constructed in such a way that no other node points to it. The \new\ node is then attached to the tree as the child of $\r$ (of $\ell_0$, respectively) in the following step. We merged these two separate steps, so the allocation and the connection to the tree occur at once in line f6 (in line r6, respectively). In addition to simplifying the rotation protocols, this allows us to simplify the type of the \rem\ field to the boolean type (instead of the tertiary type used in the original work). We argue that this merge of steps makes sense since the only change to shared memory is the change to the $\Left$ field of $\r$ (to the $\Right$ field of $\ell_0$, respectively). In contrast, the $\new$ node is unreachable by any process other than $\Sys$ until it is connected to $\r$ (to $\ell_0$, respectively). This allows $\Sys$ to treat $\new$ as a process-local address for initialization.

There are more instances in which we merged distinct steps into a single instruction. Note that each of the atomic program commands presented in Figure~\ref{RLred} includes multiple steps. We relied on the work of Elmas et al.~\cite{elmas2009calculus}, which formalized the notion of abstraction through command reduction. In our case, two consecutive commands that read to and/or write from a thread-local variable may be merged. Additionally, two consecutive commands that access the same shared memory object may be merged if one reads an immutable field of the shared object. This includes the case when both commands occur within the same critical section and one is a read command since a shared object is effectively immutable to any process that is not executing the critical section.

The system process $\Sys=0$ and each of the working processes $p>0$ act by executing their Master Programs (Figure \ref{FigMP}). The Master Program is in charge of three things: \begin{enumerate*}[label=(\arabic*)]
    \item initializing the local variables of the process,
    \item enforcing the preconditions of the operation being invoked by the process, and
    \item changing the instruction pointer of the process to the start of the operation being invoked.
\end{enumerate*}

At this point, we remark that we omitted some of the technical details of the Master Program, compared to the original presentation by Crain et al. While they settled on a specific mechanism for choosing which balancing rotations to perform, we do not commit to any such mechanism. The specific details of this decision process do not influence the correctness of the algorithm, and due to the concurrent nature of the algorithm, do not admit any hard complexity bounds.

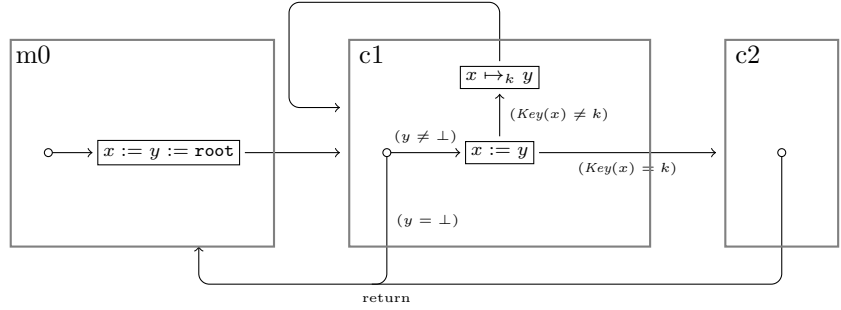
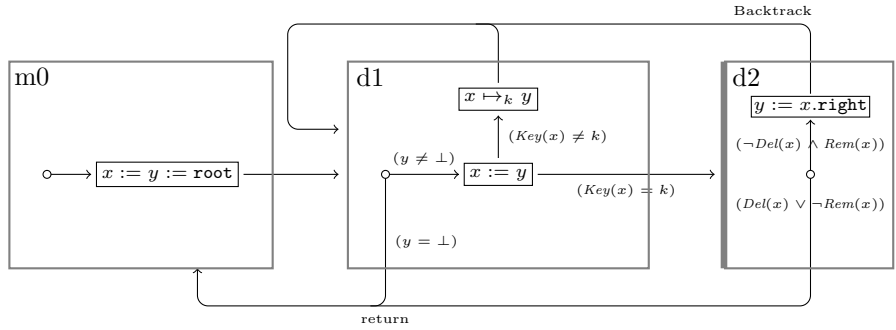
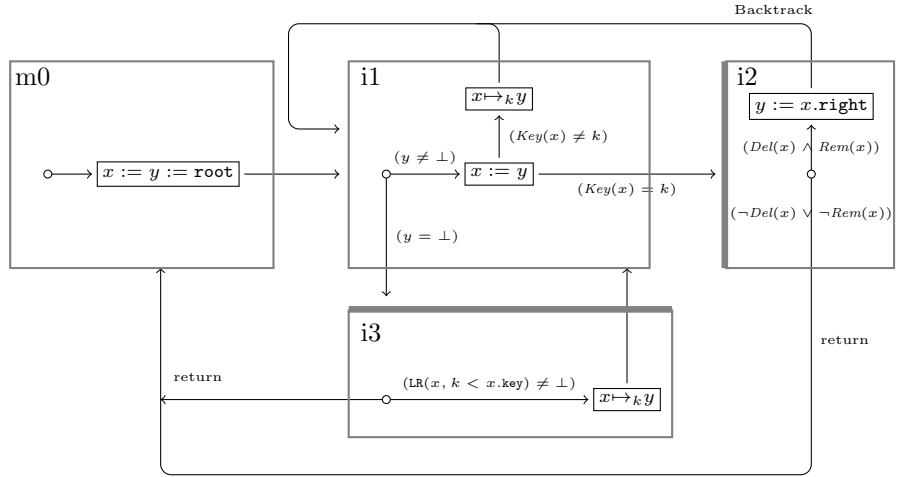
\begin{figure}
\begin{subfigure}{\textwidth}
    \centering
    \begin{tikzpicture}
        \draw
            node at (-2.5,0) [circle,draw,inner sep=0pt, outer sep=0pt, minimum size=3pt] (m0) {}
            node at (-0.9,0)[rectangle,draw,inner sep=2pt,outer sep=2pt,align=center] (init) {\scriptsize$x := y := \Root$}
            node at (1.5,0) (c1) {}
            node at (2,0) [circle,draw,inner sep=0pt, outer sep=0pt, minimum size=3pt] (in_c1) {}
            node at (3.5,0) [rectangle, draw, inner sep=2pt, outer sep=2pt] (y_to_x) {\scriptsize$x := y$}
            node [rectangle, draw, above of=y_to_x, inner sep=2pt, outer sep=2pt] (k-step) {\scriptsize$x\mapsto_k y$}
            node at (6.5,0) (c2) {}
            node at (7.25,0) [circle,draw,inner sep=0pt, outer sep=0pt, minimum size=3pt] (in_c2) {};
            
        \draw[->] (m0) -- (init);
        \draw[->] (init) -- (c1);
        \draw[->, rounded corners] (in_c1) -- node [right] {\tiny$(y=\bot)$} (2,-1.75) node [below] {\tiny{return}} -| (-0.5, -1.25);
        \draw[->] (in_c1) -- node [above] {\tiny$(y\neq\bot)$} (y_to_x);
        \draw[->] (y_to_x) -- node [below] {\tiny$(\KeyF(x)=k)$} (c2);
        \draw[->] (y_to_x) -- node [right] {\tiny$(\KeyF(x)\neq k)$} (k-step);
        \draw[->, rounded corners] (k-step) -- (3.5, 2) -- (0.7, 2) |- ($(c1.west) + (0, 0.6cm)$);
        \draw[-, rounded corners] (in_c2) |- (1.8, -1.75);
        
        \draw [color=gray,thick] (-3,-1.25) rectangle (0.5, 1.5);
        \node at (-2.7, 1.3) {m0};
        
        \draw [color=gray,thick](1.5,-1.25) rectangle (5.5,1.5);
	    \node at (1.8,1.3) {c1};
	    
	    \draw [color=gray,thick](6.5,-1.25) rectangle (8,1.5);
	    \node at (6.8,1.3) {c2};
    \end{tikzpicture} 
    \caption{Flowchart and transitions of Contains.}
    \label{FIGflowchart1}
\end{subfigure}
\\[1.5cm]
\begin{subfigure}{\textwidth}
    \centering
    \begin{tikzpicture}
        \draw
            node at (-2.5,0) [circle,draw,inner sep=0pt, outer sep=0pt, minimum size=3pt] (m0) {}
            node at (-0.9,0)[rectangle,draw,inner sep=2pt,outer sep=2pt,align=center] (init) {\scriptsize$x := y := \Root$}
            node at (1.5,0) (d1) {}
            node at (2,0) [circle,draw,inner sep=0pt, outer sep=0pt, minimum size=3pt] (in_d1) {}
            node at (3.5,0) [rectangle, draw, inner sep=2pt, outer sep=2pt] (y_to_x) {\scriptsize$x := y$}
            node [rectangle, draw, above of=y_to_x, inner sep=2pt, outer sep=2pt] (k-step) {\scriptsize$x\mapsto_k y$}
            node at (6.5,0) (d2) {}
            node at (7.65,0) [circle,draw,inner sep=0pt, outer sep=0pt, minimum size=3pt] (in_d2) {}
            node [rectangle, draw,inner sep=1pt, outer sep=1pt, above of=in_d2, yshift=-1mm] (backtrack) {\scriptsize$y := x.\Right$}
            ;
            
        \draw[->] (m0) -- (init);
        \draw[->] (init) -- (d1);
        \draw[->, rounded corners] (in_d1) -- node [right] {\tiny$(y=\bot)$} (2,-1.75) node [below] {\tiny{return}} -| (-0.5, -1.25);
        \draw[->] (in_d1) -- node [above] {\tiny$(y\neq\bot)$} (y_to_x);
        \draw[->] (y_to_x) -- node [below] {\tiny$(\KeyF(x)=k)$} (d2);
        \draw[->] (y_to_x) -- node [right] {\tiny$(\KeyF(x)\neq k)$} (k-step);
        \draw[->, rounded corners] (k-step) -- (3.5, 2) -- (0.7, 2) |- ($(d1.west) + (0, 0.6cm)$);
        \draw[->] (in_d2) -- node {\tiny$(\neg\DeletedP(x)\wedge\RemovedP(x))$} (backtrack);
        \draw[-, rounded corners] (backtrack) -- (7.65, 2) node [above,xshift=-5mm] {\tiny{Backtrack}} -- (3.2,2);
        \draw[-, rounded corners] (in_d2) -- node [yshift=5mm] {\tiny$(\DeletedP(x)\vee\neg\RemovedP(x))$} (7.65, -1.75) -- (1.8, -1.75);
        
        \draw [color=gray,thick] (-3,-1.25) rectangle (0.5, 1.5);
        \node at (-2.7, 1.3) {m0};
        
        \draw [color=gray,thick](1.5,-1.25) rectangle (5.5,1.5);
	    \node at (1.8,1.3) {d1};
	    
	    \draw [color=gray,thick](6.5,-1.25) -- (8.75,-1.25) -- (8.75,1.5) -- (6.5,1.5); \draw [color=gray,line width=2.5pt](6.5,1.5) -- (6.5,-1.25);
	    \node at (6.8,1.3) {d2};
    \end{tikzpicture} 
    \caption{Flowchart and transitions of Delete.}
    \label{FIGflowchart2}
\end{subfigure}
\\[1.5cm]
\begin{subfigure}{\textwidth}
    \centering
    \begin{tikzpicture}
        \draw
            node at (-2.5,0) [circle,draw,inner sep=0pt, outer sep=0pt, minimum size=3pt] (m0) {}
            node at (-0.9,0)[rectangle,draw,inner sep=2pt,outer sep=2pt,align=center] (init) {\scriptsize$x := y := \Root$}
            node at (1.5,0) (i1) {}
            node at (2,0) [circle,draw,inner sep=0pt, outer sep=0pt, minimum size=3pt] (in_i1) {}
            node at (3.5,0) [rectangle, draw, inner sep=2pt, outer sep=2pt] (y_to_x) {\scriptsize$x := y$}
            node [rectangle, draw, above of=y_to_x, inner sep=2pt, outer sep=2pt] (k-step) {\scriptsize$x\arrow_k y$}
            node at (6.5,0) (i2) {}
            node at (7.65,0) [circle,draw,inner sep=0pt, outer sep=0pt, minimum size=3pt] (in_i2) {}
            node [rectangle, draw,inner sep=2pt, outer sep=2pt, above of=in_i2, yshift=-1mm] (backtrack) {\scriptsize$y := x.\Right$}
            node at (2,-1.75) (i3) {}
            node at (2,-3) [circle,draw,inner sep=0pt, outer sep=0pt, minimum size=3pt] (in_i3) {}
            node at (5.2, -3) [rectangle, draw,inner sep=2pt, outer sep=2pt] (retry) {\scriptsize$x\arrow_k y$}
            ;
            
        \draw[->] (m0) -- (init);
        \draw[->] (init) -- (i1);
        \draw[->] (in_i1) -- node [right] {\tiny$(y=\bot)$} (i3);
        \draw[->] (in_i1) -- node [above] {\tiny$(y\neq\bot)$} (y_to_x);
        \draw[->] (y_to_x) -- node [below] {\tiny$(\KeyF(x)=k)$} (i2);
        \draw[->] (y_to_x) -- node [right] {\tiny$(\KeyF(x)\neq k)$} (k-step);
        \draw[->, rounded corners] (k-step) -- (3.5, 2) -- (0.7, 2) |- ($(i1.west) + (0, 0.6cm)$);
        \draw[->] (in_i2) -- node {\tiny$(\DeletedP(x)\wedge\RemovedP(x))$} (backtrack);
        \draw[-, rounded corners] (backtrack) -- (7.65, 2) node [above,xshift=-5mm] {\tiny{Backtrack}} -- (3.2,2);
        \draw[->, rounded corners] (in_i2) -- node [yshift=15mm] {\tiny$(\neg\DeletedP(x)\vee\neg\RemovedP(x))$} (7.65, -4) -- (-1, -4) -- (-1,-1.25);
        \draw[->] (in_i3) -- (-1, -3);
        \draw node at (-0.5, -2.7) {\tiny{return}};
        \draw[->] (in_i3) -- node [above] {\tiny$(\LR(x, k < x.\key)\neq\bot)$} (retry);
        \draw[->,rounded corners] (retry) -- (5.2,-1.25);
        \draw node at (8.1, -2.2) {\tiny{return}};
        
        \draw [color=gray,thick] (-3,-1.25) rectangle (0.5, 1.5);
        \node at (-2.7, 1.3) {m0};
        
        \draw [color=gray,thick](1.5,-1.25) rectangle (5.5,1.5);
	    \node at (1.8,1.3) {i1};
	    
	    \draw [color=gray,thick](6.5,-1.25) -- (8.75,-1.25) -- (8.75,1.5) -- (6.5,1.5); \draw [color=gray,line width=2.5pt](6.5,1.5) -- (6.5,-1.25);
	    \node at (6.8,1.3) {i2};
	    
	    \draw [color=gray,thick](1.5,-1.8) -- (1.5,-3.5) -- (5.8,-3.5) -- (5.8,-1.8); \draw [color=gray,line width=2.5pt] (5.8,-1.8) -- (1.5,-1.8);
	    \node at (1.8,-2.1) {i3};
    \end{tikzpicture} 
    \caption{Flowchart and transitions of Insert.}
    \label{FIGflowchart3}
\end{subfigure}
\caption{Flowcharts and transition diagrams of the data operations of the algorithm.}
\label{fig:flowcharts}
\end{figure}
The flowcharts in Figure~\ref{fig:flowcharts} are the graphical representations of the steps of the algorithm, which are detailed fully in Appendix~\ref{Sec4.1}. The flowcharts are presented to help understand the steps of the data operations, and the course of values that the program variables $\nd_p$ and $\nxt_p$ take as the operation is executed. Thus, the flowcharts do not specify the return values of the operations --- they only indicate when a return is executed and the operation terminates.

In these flowcharts, address $x$ follows the denotation of $\nd_p$ during the operation execution, and $y$ follows the denotation of $\nxt_p$.

We use Figure~\ref{FIGflowchart2}, which details the behavior of $\bfdelete(k)$, as an example to help explain the meaning of the different shapes and notations of the flowcharts.

The $\bfdelete(k)$  consists of instructions d1 and d2. Each instruction is represented by a large gray rectangle, labeled with the instruction name. Instruction m0 is also represented in the chart, as it is in charge of variable initialization, invocations, and returns. Transitions are marked with arrows, which may be labeled with parenthesized conditions that must hold for the transition to occur, e.g., the transition from d1 to d2 occurs when $(\KeyF(x)=k)$. Within each instruction, assignments to local variables appear as rectangular nodes. For example, the assignment of $y := x.\Right$ in instruction d2 if $(\neg\DeletedP(x)\wedge\RemovedP(x))$, or the assignment of the value \LR$(x, k<x.\key)$ to $y$ (denoted $x\mapsto_k y$) in instruction d1 if $\KeyF(x)\neq k$. The thick westerly-facing edge of the rectangle of d2 marks that this instruction constitutes a critical section, and thus, access to it requires that the process first acquires a lock on node $x$.


\begin{figure}
\centering
\begin{subfigure}{0.2\textwidth}
    \begin{tikzpicture}[node distance=0.75cm]
        \draw 
            node [circle, draw, fill, inner sep=0pt, outer sep=0pt, minimum size=3pt, label={180:\tiny$\prt$}] (prt) {}
            node [circle, draw, fill, inner sep=0pt, outer sep=0pt, minimum size=3pt, label={180:\tiny$\n$}, below of=prt] (nd0) {}
            node [circle, draw, fill, inner sep=0pt, outer sep=0pt, minimum size=3pt, label={0:\tiny$\r$}, below of=nd0, right of=nd0] (r) {}
            node [circle, draw, inner sep=0pt, outer sep=0pt, minimum size=3pt, label={0:\tiny$\rl$}, below of=r, xshift=-3mm] (rl) {}
            node [circle, draw, inner sep=0pt, outer sep=0pt, minimum size=3pt, label={180:\tiny$\ell_0$}, left of=rl, xshift=-2mm] (l) {};
        \draw[->] (prt) -- (nd0);
        \draw[->] (nd0) -- (l);
        \draw[->] (nd0) -- (r);
        \draw[->] (r) -- (rl);
    \end{tikzpicture}
    \caption{Initial state}
\end{subfigure}
\hfill
\begin{subfigure}{0.2\textwidth}
    \begin{tikzpicture}[node distance=0.75cm]
        \draw 
            node [circle, draw, fill, inner sep=0pt, outer sep=0pt, minimum size=3pt, label={180:\tiny$\prt$}] (prt) {}
            node [circle, draw, fill, inner sep=0pt, outer sep=0pt, minimum size=3pt, label={180:\tiny$\n$}, below of=prt] (nd0) {}
            node [circle, draw, fill, inner sep=0pt, outer sep=0pt, minimum size=3pt, label={0:\tiny$\r$}, below of=nd0, right of=nd0] (r) {}
            node [circle, draw, inner sep=0pt, outer sep=0pt, minimum size=3pt, label={90:\tiny$\new$}, below of=nd0, yshift=-2.5mm] (new) {}
            node [circle, draw, inner sep=0pt, outer sep=0pt, minimum size=3pt, label={0:\tiny$\rl$}, below of=r, xshift=-3mm] (rl) {}
            node [circle, draw, inner sep=0pt, outer sep=0pt, minimum size=3pt, label={180:\tiny$\ell_0$}, left of=rl, xshift=-2mm] (l) {};
        \draw[->] (prt) -- (nd0);
        \draw[->] (nd0) -- (l);
        \draw[->] (nd0) -- (r);
        \draw[->] (r) -- (new);
        \draw[->] (new) -- (rl);
        \draw[->] (new) -- (l);
    \end{tikzpicture}
    \caption{after f6}
\end{subfigure}
\hfill
\begin{subfigure}{0.2\textwidth}
    \begin{tikzpicture}[node distance=0.75cm]
        \draw 
            node [circle, draw, fill, inner sep=0pt, outer sep=0pt, minimum size=3pt, label={180:\tiny$\prt$}] (prt) {}
            node [circle, draw, fill, inner sep=0pt, outer sep=0pt, minimum size=3pt, label={180:\tiny$\n$}, below of=prt] (nd0) {}
            node [circle, draw, fill, inner sep=0pt, outer sep=0pt, minimum size=3pt, label={0:\tiny$\r$}, below of=nd0, right of=nd0] (r) {}
            node [circle, draw, inner sep=0pt, outer sep=0pt, minimum size=3pt, label={180:\tiny$\new$}, left of=r, yshift=-2.5mm] (new) {}
            node [circle, draw, inner sep=0pt, outer sep=0pt, minimum size=3pt, label={0:\tiny$\rl$}, below of=r, xshift=-3mm] (rl) {}
            node [circle, draw, inner sep=0pt, outer sep=0pt, minimum size=3pt, label={180:\tiny$\ell_0$}, left of=rl, xshift=-1mm] (l) {};
        \draw[->] (prt) -- (nd0);
        \draw[->] (nd0.east) -- (r.east);
        \draw[->] (nd0.west) -- (r.west);
        \draw[->] (r) -- (new);
        \draw[->] (new) -- (rl);
        \draw[->] (new) -- (l);
    \end{tikzpicture}
    \caption{after f7}
\end{subfigure}
\hfill
\begin{subfigure}{0.2\textwidth}
    \begin{tikzpicture}[node distance=0.75cm]
        \draw 
            node [circle, draw, fill, inner sep=0pt, outer sep=0pt, minimum size=3pt, label={180:\tiny$\prt$}] (prt) {}
            node [circle, draw, fill, inner sep=0pt, outer sep=0pt, minimum size=3pt, label={180:\tiny$\n$}, below of=prt] (nd0) {}
            node [circle, draw, fill, inner sep=0pt, outer sep=0pt, minimum size=3pt, label={0:\tiny$\r$}, below of=nd0, right of=nd0] (r) {}
            node [circle, draw, inner sep=0pt, outer sep=0pt, minimum size=3pt, label={180:\tiny$\new$}, left of=r, yshift=-2.5mm] (new) {}
            node [circle, draw, inner sep=0pt, outer sep=0pt, minimum size=3pt, label={0:\tiny$\rl$}, below of=r, xshift=-3mm] (rl) {}
            node [circle, draw, inner sep=0pt, outer sep=0pt, minimum size=3pt, label={180:\tiny$\ell_0$}, left of=rl, xshift=-1mm] (l) {};
        \draw[->] (prt) -- (r);
        \draw[->] (nd0.east) -- (r.east);
        \draw[->] (nd0.west) -- (r.west);
        \draw[->] (r) -- (new);
        \draw[->] (new) -- (rl);
        \draw[->] (new) -- (l);
    \end{tikzpicture}
    \caption{after f8}
\end{subfigure}
\caption{Illustration of the structural changes caused by \rotateLeft. $\bullet$ represents locked nodes, and $\circ$ represents unlocked nodes.}
\label{Fig:rotate-left}
\end{figure}
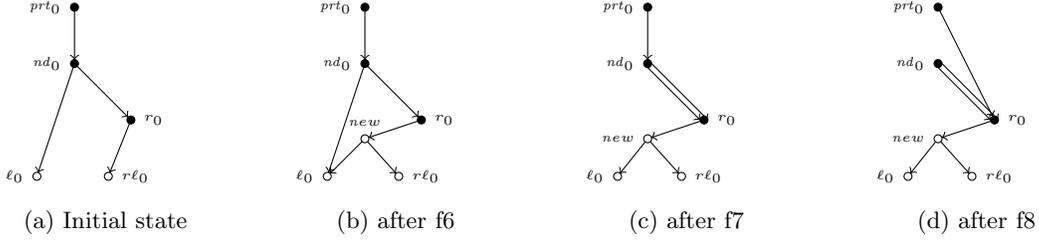

\begin{figure}
\centering
\begin{subfigure}{0.2\textwidth}
    \begin{tikzpicture}[node distance=0.75cm]
        \draw 
            node [circle, draw, fill, inner sep=0pt, outer sep=0pt, minimum size=3pt, label={180:\tiny$\prt$}] (prt) {}
            node [circle, draw, fill, inner sep=0pt, outer sep=0pt, minimum size=3pt, label={180:\tiny$\n$}, below of=prt, left of=prt] (nd0) {}
            node [circle, draw, inner sep=0pt, outer sep=0pt, minimum size=3pt, label={0:\tiny$\child$}, below of=nd0, right of=nd0] (chd) {}
            node [circle, draw, inner sep=0pt, outer sep=0pt, minimum size=3pt, label={180:\tiny$\bot$}, below of=nd0] (bot) {};
        \draw[->] (prt) to (nd0);
        \draw[->] (nd0) to (chd);
        \draw[->] (nd0) to (bot);
    \end{tikzpicture}
    \caption{Initial state}
\end{subfigure}
\hfill
\begin{subfigure}{0.2\textwidth}
    \begin{tikzpicture}[node distance=0.75cm]
        \draw 
            node [circle, draw, fill, inner sep=0pt, outer sep=0pt, minimum size=3pt, label={180:\tiny$\prt$}] (prt) {}
            node [circle, draw, fill, inner sep=0pt, outer sep=0pt, minimum size=3pt, label={180:\tiny$\n$}, below of=prt, left of=prt] (nd0) {}
            node [circle, draw, inner sep=0pt, outer sep=0pt, minimum size=3pt, label={0:\tiny$\child$}, below of=nd0, right of=nd0] (chd) {}
            node [circle, draw, inner sep=0pt, outer sep=0pt, minimum size=3pt, label={180:\tiny$\bot$}, below of=nd0] (bot) {};
        \draw[->] (prt) to (chd);
        \draw[->] (nd0) to (chd);
        \draw[->] (nd0) to (bot);
    \end{tikzpicture}
    \caption{After v6}
\end{subfigure}
\hfill
\begin{subfigure}{0.2\textwidth}
    \begin{tikzpicture}[node distance=0.75cm]
        \draw 
            node [circle, draw, fill, inner sep=0pt, outer sep=0pt, minimum size=3pt, label={0:\tiny$\prt$}] (prt) {}
            node [circle, draw, fill, inner sep=0pt, outer sep=0pt, minimum size=3pt, label={180:\tiny$\n$}, below of=prt, left of=prt] (nd0) {}
            node [circle, draw, inner sep=0pt, outer sep=0pt, minimum size=3pt, label={0:\tiny$\child$}, below of=nd0, right of=nd0] (chd) {}
            node [circle, draw, inner sep=0pt, outer sep=0pt, minimum size=3pt, label={180:\tiny$\bot$}, below of=nd0] (bot) {};
        \draw[->] (prt) to (chd);
        \draw[->] (nd0) to (chd);
        \draw[->] (nd0) to [bend left] (prt);
    \end{tikzpicture}
    \caption{After v7}
\end{subfigure}
\hfill
\begin{subfigure}{0.2\textwidth}
    \begin{tikzpicture}[node distance=0.75cm]
        \draw 
            node [circle, draw, fill, inner sep=0pt, outer sep=0pt, minimum size=3pt, label={0:\tiny$\prt$}] (prt) {}
            node [circle, draw, fill, inner sep=0pt, outer sep=0pt, minimum size=3pt, label={180:\tiny$\n$}, below of=prt, left of=prt] (nd0) {}
            node [circle, draw, inner sep=0pt, outer sep=0pt, minimum size=3pt, label={0:\tiny$\child$}, below of=nd0, right of=nd0] (chd) {}
            node [circle, draw, inner sep=0pt, outer sep=0pt, minimum size=3pt, label={180:\tiny$\bot$}, below of=nd0] (bot) {};
        \draw[->] (prt) to (chd);
        \draw[->] (nd0) to [bend right] (prt);
        \draw[->] (nd0) to [bend left] (prt);
    \end{tikzpicture}
    \caption{After v8}
\end{subfigure}
\caption{Illustration of the structural changes caused by \bfremove. $\bullet$ represents locked nodes, and $\circ$ represents unlocked nodes.}
\label{Fig:remove}
\end{figure}
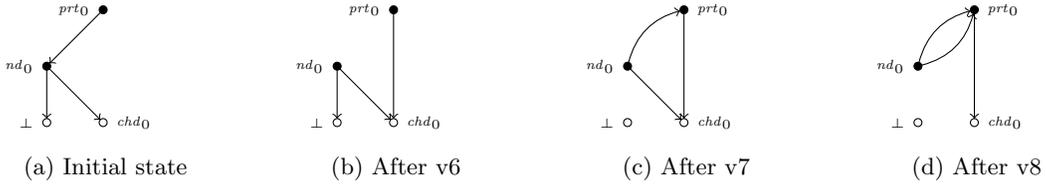

Now that we are acquainted with the algorithm and before delving into the details of the framework and the proof, we want to present the complexities inherent in the CF algorithm in an abstract manner and to demonstrate a few aspects of its behavior that we believe make it challenging to prove correct. In the next subsection, we will discuss some of these aspects with the help of Figure~\ref{fig:big-graph-example}.

\subsection{Exploring an Example}
\label{Sec:exploring-an-example}
Figure~\ref{fig:big-graph-example} illustrates an example of a series of non-contiguous memory structures that may appear in an arbitrary execution of the CF algorithm. All definitions given in this section will be repeated in due time in a more formal manner. 

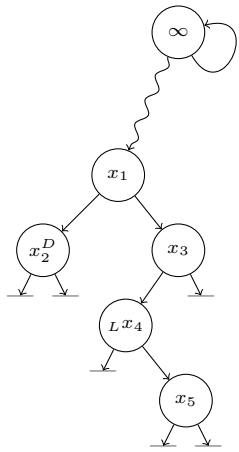
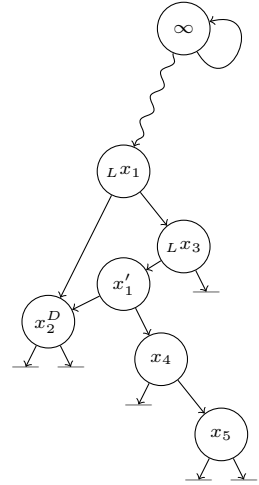
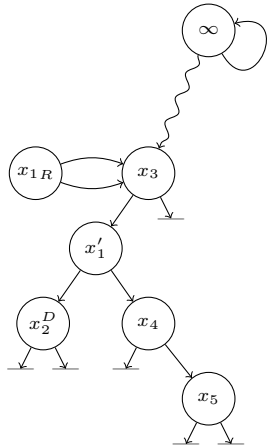
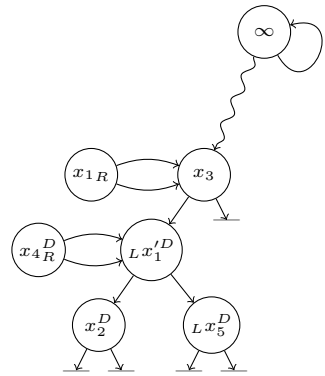
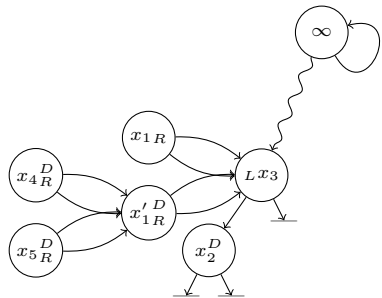
\begin{figure}
    \centering
   
    \begin{subfigure}{0.3\textwidth}
        \begin{tikzpicture}[decoration={snake, amplitude=0.6mm}]
            \tikzstyle{graphnode} = [circle, draw, font=\scriptsize, inner sep=1.2pt, minimum size=7mm]
            \draw
                node [graphnode] (root) {$\infty$}
                node [graphnode, below of=root,left of=root, xshift=2mm, yshift=-9mm] (nd0) {$x_1$}
                node [graphnode, below of=nd0, right of=nd0, xshift=-2mm] (r0) {$x_3$}
                node [graphnode, below of=nd0, left of=nd0] (l0) {$x_2^D$}
                node [graphnode, below of=r0, left of=r0, xshift=3mm] (rl0) {$_{L}x_4$}
                node [graphnode, below of=rl0, right of=rl0, xshift=-2mm] (rlr) {$x_5$}
                ;
                
                \draw[->, decorate] (root) -- (nd0);
                \draw[->] (root) to[out=-60,in=15,looseness=5] (root);
                \draw[->] (nd0) -- (l0);
                \draw[->] (nd0) -- (r0);
                \draw[->] (r0) -- (rl0);
                \draw[->] (r0) -- ++(.3cm,-.6cm) node[yshift=-0.25mm] {---};
                \draw[->] (rl0) -- ++(-.3cm,-.6cm) node[yshift=-0.25mm] {---};
                \draw[->] (rl0) -- (rlr);
                \draw[->] (rlr) -- ++(.3cm,-.6cm) node[yshift=-0.25mm] {---};
                \draw[->] (rlr) -- ++(-.3cm,-.6cm) node[yshift=-0.25mm] {---};
                \draw[->] (l0) -- ++(-.3cm,-.6cm) node[yshift=-0.25mm] {---};
                \draw[->] (l0) -- ++(.3cm,-.6cm) node[yshift=-0.25mm] {---};
        \end{tikzpicture}
        \caption{}\label{subfig:valid-state}
    \end{subfigure}
    \hfill
    \begin{subfigure}{0.3\textwidth}
        \begin{tikzpicture}[decoration={snake, amplitude=0.6mm}]
            \tikzstyle{graphnode} = [circle, draw, font=\scriptsize, inner sep=1.2pt, minimum size=7mm]
            \draw
                node [graphnode] (root) {$\infty$}
                node [graphnode, below of=root,left of=root, xshift=2mm, yshift=-9mm] (nd0) {$_{L}x_1$}
                node [graphnode, below of=nd0, right of=nd0, xshift=-2mm] (r0) {$_{L}x_3$}
                node [graphnode, below of=nd0, yshift=-5mm] (new) {$x_1'$}
                node [graphnode, below of=nd0, left of=nd0, yshift=-10mm] (l0) {$x_2^D$}
                node [graphnode, below of=new, xshift=5mm] (rl0) {$x_4$}
                node [graphnode, below of=rl0, right of=rl0, xshift=-2mm] (rlr) {$x_5$}
                ;
                
                \draw[->, decorate] (root) -- (nd0);
                \draw[->] (root) to[out=-60,in=15,looseness=5] (root);
                \draw[->] (nd0) -- (l0);
                \draw[->] (nd0) -- (r0);
                \draw[->] (r0) -- (new);
                \draw[->] (r0) -- ++(.3cm,-.6cm) node[yshift=-0.25mm] {---};
                \draw[->] (new) -- (rl0);
                \draw[->] (new) -- (l0);
                \draw[->] (rl0) -- ++(-.3cm,-.6cm) node[yshift=-0.25mm] {---};
                \draw[->] (rl0) -- (rlr);
                \draw[->] (rlr) -- ++(.3cm,-.6cm) node[yshift=-0.25mm] {---};
                \draw[->] (rlr) -- ++(-.3cm,-.6cm) node[yshift=-0.25mm] {---};
                \draw[->] (l0) -- ++(-.3cm,-.6cm) node[yshift=-0.25mm] {---};
                \draw[->] (l0) -- ++(.3cm,-.6cm) node[yshift=-0.25mm] {---};
        \end{tikzpicture}
        \caption{}\label{subfig:mid-rotate-l-f6}
    \end{subfigure}
    \hfill
    \begin{subfigure}{0.45\textwidth}
        \begin{tikzpicture}[decoration={snake, amplitude=0.6mm}]
            \tikzstyle{graphnode} = [circle, draw, font=\scriptsize, inner sep=1.2pt, minimum size=7mm]
            \draw
                node [graphnode] (root) {$\infty$}
                node [graphnode, below of=root, left of=root, xshift=2mm, yshift=-9mm] (r0) {$x_3$}
                node [graphnode, left of=r0, xshift=-5mm] (nd0) {${x_1}_R$}
                node [graphnode, below of=r0, left of=r0, xshift=3mm] (new) {$x_1'$}
                node [graphnode, below of=new, left of=new, xshift=3mm] (l0) {$x_2^D$}
                node [graphnode, below of=new, right of=new, xshift=-3mm] (rl0) {$x_4$}
                node [graphnode, below of=rl0, right of=rl0, xshift=-2mm] (rlr) {$x_5$}
                ;
                
                \draw[->, decorate] (root) -- (r0);
                \draw[->] (root) to[out=-60,in=15,looseness=5] (root);
                \draw[->] (nd0) to [bend left=20] (r0);
                \draw[->] (nd0) to [bend right=20] (r0);
                \draw[->] (r0) -- (new);
                \draw[->] (r0) -- ++(.3cm,-.6cm) node[yshift=-0.25mm] {---};
                \draw[->] (new) -- (rl0);
                \draw[->] (new) -- (l0);
                \draw[->] (rl0) -- ++(-.3cm,-.6cm) node[yshift=-0.25mm] {---};
                \draw[->] (rl0) -- (rlr);
                \draw[->] (rlr) -- ++(.3cm,-.6cm) node[yshift=-0.25mm] {---};
                \draw[->] (rlr) -- ++(-.3cm,-.6cm) node[yshift=-0.25mm] {---};
                \draw[->] (l0) -- ++(-.3cm,-.6cm) node[yshift=-0.25mm] {---};
                \draw[->] (l0) -- ++(.3cm,-.6cm) node[yshift=-0.25mm] {---};
        \end{tikzpicture}
        \caption{}\label{subfig:rotate-l-done}
    \end{subfigure}
    \hfill
    \begin{subfigure}{0.45\textwidth}
        \begin{tikzpicture}[decoration={snake, amplitude=0.6mm}]
            \tikzstyle{graphnode} = [circle, draw, font=\scriptsize, inner sep=1.2pt, minimum size=7mm]
            \draw
                node [graphnode] (root) {$\infty$}
                node [graphnode, below of=root, left of=root, xshift=2mm, yshift=-9mm] (r0) {$x_3$}
                node [graphnode, left of=r0, xshift=-5mm] (nd0) {${x_1}_R$}
                node [graphnode, below of=r0, left of=r0, xshift=3mm] (new) {$_{L}x_1'^D$}
                node [graphnode, below of=new, left of=new, xshift=3mm] (l0) {$x_2^D$}
                node [graphnode, left of=new, xshift=-5mm] (rl0) {${x_4}^D_R$}
                node [graphnode, below of=new, right of=new, xshift=-2mm] (rlr) {$_{L}x_5^D$}
                ;
                
                \draw[->, decorate] (root) -- (r0);
                \draw[->] (root) to[out=-60,in=15,looseness=5] (root);
                \draw[->] (nd0) to [bend left=20] (r0);
                \draw[->] (nd0) to [bend right=20] (r0);
                \draw[->] (r0) -- (new);
                \draw[->] (r0) -- ++(.3cm,-.6cm) node[yshift=-0.25mm] {---};
                \draw[->] (new) -- (rlr);
                \draw[->] (new) -- (l0);
                \draw[->] (rlr) -- ++(.3cm,-.6cm) node[yshift=-0.25mm] {---};
                \draw[->] (rlr) -- ++(-.3cm,-.6cm) node[yshift=-0.25mm] {---};
                \draw[->] (rl0) to [bend left=20] (new);
                \draw[->] (rl0) to [bend right=20] (new);
                \draw[->] (l0) -- ++(-.3cm,-.6cm) node[yshift=-0.25mm] {---};
                \draw[->] (l0) -- ++(.3cm,-.6cm) node[yshift=-0.25mm] {---};
        \end{tikzpicture}
        \caption{}\label{subfig:mid-remove-v6}
    \end{subfigure}
    \hfill
    \begin{subfigure}{0.5\textwidth}
        \begin{tikzpicture}[decoration={snake, amplitude=0.6mm}]
            \tikzstyle{graphnode} = [circle, draw, font=\scriptsize, inner sep=1.2pt, minimum size=7mm]
            \draw
                node [graphnode] (root) {$\infty$}
                node [graphnode, below of=root, left of=root, xshift=2mm, yshift=-9mm] (r0) {$_{L}x_3$}
                node [graphnode, left of=r0, xshift=-5mm, yshift=5mm] (nd0) {${x_1}_R$}
                node [graphnode, left of=r0, xshift=-5mm, yshift=-5mm] (new) {${x_1'}^D_R$}
                node [graphnode, below of=r0, left of=r0, xshift=3mm] (l0) {$x_2^D$}
                node [graphnode, left of=new, xshift=-5mm, yshift=5mm] (rl0) {${x_4}^D_R$}
                node [graphnode, left of=new, xshift=-5mm, yshift=-5mm] (rlr) {${x_5}^D_R$}
                ;
                
                \draw[->, decorate] (root) -- (r0);
                \draw[->] (root) to[out=-60,in=15,looseness=5] (root);
                \draw[->] (nd0) to [bend left=20] (r0);
                \draw[->] (nd0) to [bend right=20] (r0);
                \draw[->] (r0) -- (l0);
                \draw[->] (r0) -- ++(.3cm,-.6cm) node[yshift=-0.25mm] {---};
                \draw[->] (new) to [bend left=20] (r0);
                \draw[->] (new) to [bend right=20] (r0);
                \draw[->] (rl0) to [bend left=20] (new);
                \draw[->] (rl0) to [bend right=20] (new);
                \draw[->] (rlr) to [bend left=20] (new);
                \draw[->] (rlr) to [bend right=20] (new);
                \draw[->] (l0) -- ++(-.3cm,-.6cm) node[yshift=-0.25mm] {---};
                \draw[->] (l0) -- ++(.3cm,-.6cm) node[yshift=-0.25mm] {---};
        \end{tikzpicture}
        \caption{}\label{subfig:remove-done}
    \end{subfigure}
    \caption{Illustration of five states in some execution of the algorithm. Nodes marked with $_L$ are locked, nodes marked with $_R$ are removed, and nodes marked with $^D$ are deleted. The zigzag line refers to
		a path from the root to $x_1$.}
    \label{fig:big-graph-example}
\end{figure}

Figure \ref{subfig:valid-state} shows a valid state $M_a$ of the CF algorithm, with the focus on an unbalanced sub-tree consisting of nodes $x_2$, $x_1$, $x_4$, $x_5$, $x_3$ in their increasing key values. Node $x_2$ is logically deleted and thus, is not in $\Set(M_a)=\{ x_1, x_4, x_5, x_3\}$, represented by $M_a$. A process $p$ is inserting the value $x_5$ into the tree, but the parent of the new node, $x_4$, which is locked by $p$ for the duration of the insertion has yet to be unlocked. 

We suggest, as an exercise, to describe a full scenario, beginning with the initial state and ending
 with $M_a$. For example, insert the values 10, 22, 14, 18, and 13 one after the other. At this point, the
tree is not balanced; for example, node 10 has no left descendant but has 4 right descendants. Now
continue by deleting node 13 and adding node 15, and then node 17. The letter $L$ at node 15 indicates that
this node is still locked by process $p$ that added node 17. Every node of $M_a$ is path-connected, which
means that there is a parent--child path from the root to that node.

Figure \ref{subfig:mid-rotate-l-f6} shows a later state, $M_b$, in which the system process (called \Sys)
 is in the middle of  a \rotateLeft\ operation. Virtually, the rotation is performed by moving $x_3$ ``up'' and moving $x_1$ ``down and to the left''. However, the rotation is actually more complicated. Both nodes are locked by the system process for the duration of the operation (and so is the parent node of $x_1$, which
is not explicitly shown in the figure, but is the node with key 22 in our concrete example). Instead of shifting node $x_1$ down, the CF algorithm clones it (i.e., creates a new node with the same key and delete features). The new clone, denoted as $x_1'$, is the left child of $x_3$. It has the same left child as the original ($x_2$), and the previous left child of $x_3$ ($x_4$) is now the right child of the cloned node. Note that in $M_b$, the graph is no longer a BST: the original node $x_1$ has a right descendant (i.e., $x_1'$) that has the same key-value as $x_1$, and node $x_2$ has two parents: $x_1$ and $x_1'$.

With Figure~\ref{subfig:mid-rotate-l-f6} we can exemplify two important concepts that play a significant role
in the proof. A path-connected node $a$ is said to be {\em tree-like} if any right-descendant of $a$ has a key value
greater than the value of $a$ and every left-descendant has a smaller key value. In our example, node $x_1$
(the node with key 14 in our concrete example) is not tree-like, but all other nodes of the graph are tree-like,
and in particular, $x_1'$ is a tree-like node. The second important concept is that of a {\em confluent} node, that is, a node
that has two parents. In our example, node $x_2$ (13) is confluent, and its two parents are node $x_1$ and $x_1'$ (14). 

With these two concepts, several important questions arise, the answers to which are required in the proof. Is it possible
to have more than one node that is not tree-like? Could there be more than one confluent node in a state?
Can a confluent node have more than two parents? (We say that node $x$ is a parent of node $y$ if there is
a path from the root to node $y$ where $x$ is the immediate predecessor of $x$ on that path.) The answer
to all these questions is negative. The reader may find that these answers are intuitively evident, but we do not
think there is an easy proof for them. 

Figure~\ref{subfig:rotate-l-done} shows a later state, $M_c$, in which the system process has just completed the \rotateLeft\ operation from Figure~\ref{subfig:mid-rotate-l-f6}. The original node $x_1$ has been removed, $x_3$ has taken its place, and the cloned node $x_1'$ is the left child of $x_3$. Note that while the nodes that are {\em reachable} from the root node once again constitute a BST, the graph, as a whole, does not, 
since node $x_3$ is pointed to by two nodes (one of which is removed). There is a distinction between a removed node and
a deleted node. A deleted node that is not removed is not contributing its key value to the set of the state, and
it may (under some conditions) regain its status as non-deleted through an insert operation.
 A removed node is not necessarily deleted, and it remains removed forever. Removing does not mean it
is not part of the tree; it is possible for a process executing an operation to reach a non-removed node in its
searching phase, stay dormant in that node while it is being removed, and wake up in what is now a removed node. The process
is still required to continue its search, even if that node's key is the search key (see instructions d2 and i2). 
This mechanism is called {\em backtracking}, an essential feature of the CF algorithm.

As an exercise, the reader may want to complete the sequence of states that need to be added to
reach state $M_c$. Figure~\ref{Fig:rotate-left} can help in this exercise. 
Now, suppose that a process $p>0$ (a {\em working} process) deletes node $x_4$ and another process $q>0$
deletes $x_5$, that is, continuing our exercise, they execute operations $\bfdelete(15)$ and 
$\bfdelete(17)$. This requires that nodes $x_4$ and $x_5$ be locked by $p$ and by $q$, respectively.
Then, the system process \Sys\ is called to remove node $x_4$, that is,
to execute $\bfremove(x_1', \FALSE)$, which should be interpreted as the removal of the right child of
node $x_1'$, i.e., the removal of $x_4$. Note that the left child of $x_4$ is $\bot$, and hence, in the notation of the code, $x_5=\child$. The next state is the result of this removal.

Figure~\ref{subfig:mid-remove-v6} shows a later state, $M_d$, in which $x_4$ is already removed, and nodes
$x_1'$ and $x_5$ are deleted. The system process is in the middle of performing a \bfremove\ operation, physically removing the logically deleted node $x_5$. Although $x_1'$ is also logically deleted, it cannot be removed yet, since neither of its children is $\bot$, which is a precondition of the \bfremove\ operation.

Figure~\ref{subfig:remove-done} shows a much later state, $M_e$, in which the system process has just finished removing the logically deleted node $x_1'$. This follows the completion of the \bfremove\ operation being executed in Figure~\ref{subfig:mid-remove-v6}. Among the effects of that operation was setting the right child of $x_1'$ to $\bot$, enabling the physical removal of $x_1'$. Of particular interest in this figure is the fact that a complex structure of removed nodes has begun to develop, in which separate sub-graphs of removed nodes are chained together, forming complicated ``dendrite-like'' structures that, though external to the ``tree'' portion of the graph, may still participate
in active operations. Note that these dendrite-like structures are anchored to a single ``in-tree'' node (in this case, $x_3$).

As is evident from this small example, the states of the CF algorithm quickly evolve from having the structure of a BST to a much more complex graph structure. Nevertheless, these complex graphs still have some form of regularity, maintaining a multitude of invariants. We delve more deeply into invariants in Section~\ref{sec:invariants-and-props}, in which we try to formulate a notion of regularity that is both an inductive invariant and a useful statement that can be used in the correctness 
proof. (See Definition~\ref{DefReg} of regularity and Theorem~\ref{InvReg} for the proof that regularity
is an invariant.)

Next, we present some of the challenges that this algorithm poses. Elegant and simple as it may seem, it hides
quite complex behaviors. The simplest operation, that of $\bfcontains(x_1)$, is useful for this purpose. Consider again Figure~\ref{subfig:mid-rotate-l-f6}, and imagine that a process $p>0$ is traversing the graph, searching for a node with value $x_1$ (corresponding to 14 in our exercise). If $p$ works alone 
on structure $M_a$ with other processes being inactive, then it
would certainly return a correct answer: $p$ reaches node $x_1$ and reports that value $x_1$ is on the set. 
(Similarly, if $p$ traverses state $M_e$ in search of $x_1$, then it reaches the bottom node
$\bot$ and reports correctly that value $x_1$ is not to be found on the set).
More commonly, the execution of an operation is spread over many structures, since the operations of the different processes
(including of the system) are interleaved. Thus, the processes ``pretend'' that their world is not in permanent flux.
With values as in the exercise, take an execution of $\bfinsert(16)$ by process $p$ and suppose first
a simple case in which the execution occurs completely in state $M_a$. Then
$p$ reaches node $x_5$ (value 17) and finds that its left child, $\nxt_p$, is $\bot$. The code
directs $p$ to \goto i3, after which $p$ obtains a lock on $x_5$ and then adds a new node with key value 16. 
Suppose that instead of $p$ reaching node $x_5$ in $M_a$ after a long traversal, it is sent to execute i3. 
Process $p$ requests a lock on $\nd_p=x_5$, but the scheduler prefers to activate the system process \Sys
so that when $p$ gets the lock, it finds itself at state $M_e$, and when it 
checks $\LR(\nd_p, k_p<\nd_p.\key)$, instead of the previous $\bot$ node,
it finds $x_1'\neq \bot$. 
As a result, process $p$ performs a backtracking step by executing $\nxt := \LR(\nd, k<\nd.\key)$ at line i3. 
Thus, process $p$ at state $M_e$ would reach node $x_1'$ in one step and then,
in two additional steps, get to node $x_2$ and, if all goes well, add a new node of key
value 16 as a right child of $x_2$.

This example demonstrates the need to have a precise definition for a correct traversal process. Such a definition is crucial to proving the correctness of the algorithm. This is why the Scanning Theorem (see Section~\ref{SecHistory}) is such a core component of our proof system. This theorem, in turn, relies on the foundation of a whole body of invariants and behavioral properties that, at first glance, may seem simple, even trivial, but are in fact not so.

Consider, for example, one of the properties we prove in Section~\ref{sec:invariants-and-props}: In any state of the algorithm, for any node $x$, if $x$ is physically removed, then there is no path from the root node to $x$. This invariant sounds intuitively correct but is actually difficult to prove, and it relies on a step-property that depends on the notion
of regularity (see \ref{sp:rem-not-pc}).


\section{Preliminaries}
\label{sec:perlims}
\subsection{Address Structures}
\label{sec:address-structure}
An {\em address structure} is a structure in the model-theoretic sense, the aim of which is to model the state of the memory space at a specific moment during the execution of the CF algorithm. Thus, throughout the execution of the algorithm, a sequence of memory states is created, and each one is represented by a specific address structure. In some contexts we prefer the term ``state'' over ``address structure'', but these terms have the same meaning here. The term {\em structure} refers here to an interpreting structure
of a certain (mostly first-order) logical language, and specifically
an address structure interprets the logical language $\mathcal{L}_{AS}$ 
which we now define. 

\begin{enumerate}
    \item An address structure has four sorts (types of the members of the structure universe) which are \AddressT, \KeyT, \InstructionT, and \ProcT. There are additional sorts that are standard and not specific to the address structure language, such as the Boolean sort (with values \TRUE, \FALSE), and the natural numbers $\omega$. Members of the \AddressT\ sort are called {\em addresses} or {\em nodes}.

    Address structures interpret these sorts as follows:
    \begin{enumerate*}
        \item \AddressT\ is a finite set of addresses which includes the two special distinct values $\Root$ and $\bot$.
        \item \KeyT\ is the set of natural numbers, with the addition of the two special distinct values $\infty$ and $-\infty$.
        \item \InstructionT\ is the set of the command identifiers of the algorithm (e.g, c1 is the identifier of the first command of the \bfcontains\ operation).
        \item \ProcT\ is the set of processes $\{0,\ldots,N\}$. Processes $p$ where $p>0$ are said to be `working' processes, and process $p=0$ is the system process \Sys.
    \end{enumerate*}
    
    The \AddressT\ sort of one structure may be different from that of another structure, but the other sorts
    have the same interpretation in our structures, which model
    states of the CF algorithm.

    \item There are two unary predicates defined over the \AddressT\ sort: \DeletedP, and \RemovedP.
  
    A binary predicate $\LockedP(a,p)$ is defined over $\AddressT\times \ProcT$.

    \item There are four function symbols in $\mathcal{L}_{AS}$:
    \begin{enumerate}
        \item $\KeyF\colon\AddressT\to \KeyT$ maps addresses to key values. We require that $\KeyF(\Root)=\infty$, and $\KeyF(\bot)=-\infty$ in every structure.
        \item $\LeftF,\RightF\colon\AddressT\to \AddressT$ map addresses to addresses. We require that 
        $\LeftF(\bot)=\RightF(\bot)=\bot$ and $\RightF(\Root)= \Root$ in every structure. 
        If $b=\LeftF(a)$ ($b=\RightF(a)$) we say that $b$ is the left (right) child of $a$. 
        \item $\ControlF\colon\ProcT \to \InstructionT$ maps process id's to instructions, and represents the program counters of the various processes.
    \end{enumerate}
        
    \item As any logical language, $\mathcal{L}_{AS}$ includes logical variables which range over the different sorts. For example, in
    the sentence
    $\forall x (x\neq \Root\wedge x\neq\bot\rightarrow \KeyF(x)\in\omega)$, $x$ is a
    quantified logical variable of sort \AddressT. 
    We have the following conventions.
    \begin{enumerate*}
        \item $x$, $y$, $z$, $w$, $a$ range over the \AddressT\ sort;
        \item $k$ ranges over the \KeyT\ sort; and
        \item $p$ and $q$ range over the \ProcT\ sort.
    \end{enumerate*}
    
    Additionally, $\mathcal{L}_{AS}$ has {\em names} which denote addresses,
    but unlike the logical variables cannot be quantified.
    For example, $p_3$ is the name of the third process, and
    it does not make sense to begin a formula with $\exists p_3 (...)$.
    The program variables($\nd_p$, $\nxt_p$, $k_p$ etc.) which appear in the code of the algorithm are names in $\mathcal{L}_{AS}$. 
      
    
    
    \item 
    The term $\new$ is a shorthand for an address definition:  
    \begin{equation}
    \label{Def-Of-new}
        \new = \begin{cases}
                    \Left(\r) & \ControlF(\Sys) \in \{\rom{f7, f8, f9}\} \\
                    \Right(\ell_0) & \ControlF(\Sys) \in \{\rom{r7, r8, r9}\} \\
                    \Root & \text{otherwise}
                \end{cases}
    \end{equation}
    
\end{enumerate}

Let $M$ be any structure that interprets the $\mathcal{L}_{AS}$ language.
For any term or formula $X$ of $\mathcal{L}_{AS}$, $X^M$ denotes the interpretation of $X$ in $M$. For example, $\AddressT^M$ is the
set of addresses of $M$, $\LeftF^M\colon\AddressT^M\to \AddressT^M$ is 
the interpretation in $M$ of the function symbol
$\LeftF$, $\nd_p^M$ is the address to which program-variable
$\nd_p$ refers to in $M$, and $(a=\RightF(\nd_p)^M)$
is the statement that the right child in $M$
of the address of program-variable $\nd_p$ is $a$. Sometimes, when the relevant address structure is 
obvious, the state-identifier superscript is omitted.

Throughout this work, we use $\varphi(p)$ to denote the instantiation of $\varphi$ to process $p>0$.

\begin{definition}
The initial address structure is defined as follows:
\begin{enumerate}
\item The \AddressT\ sort of the initial structure contains only \Root\ and $\bot$.
 
\item  Predicates \DeletedP,  \RemovedP, and \LockedP\ have the empty extension in the initial structure.

\item $\KeyF(\Root)=\infty$, $\KeyF(\bot)=-\infty$.
$\RightF(\Root)=\Root$,
$\LeftF(\Root) =\bot$,
and $\RightF(\bot)=\LeftF(\bot)=\bot$.

\item The program-counter of any process is at line m0 of the master-program, i.e. $\ControlF(p)=m0$ for every process $p\in\ProcT$.

\item For any process $p\in\ProcT$, $\nd_p=\Root$. For $p>0$, $\nxt_p=\Root$, and $\prt=\r=\ell_0=\lr=\rl=\Root$ and $\lft=\TRUE$.
\end{enumerate}
\end{definition}

\begin{definition}[Paths]
\label{DefPath}
Let $M$ be an address structure.
\begin{enumerate}
    \item We say that address $x$ {\em points to} address $y$ in $M$, denoted $x\arrow y$, if
    $\LeftF(x)=y \vee \RightF(x)= y$ holds in $M$. 
    
    \item A {\em path} in $M$ is a sequence $P$ of addresses $(x_0,\ldots, x_{n})$ (where $n\geq 0$)
    such that for every index $0\leq i <n$, $x_i\arrow x_{i+1}$. We say that
    $(x_i,x_{i+1})$ is an {\em arc} on $P$, and that addresses $x_i$ and $x_{i+1}$ are on the path.  
    Path $P$ is said to lead from $x_0$ to $x_{n}$ in $M$. 
    
    \item The transitive and reflexive closure of the $x\arrow y$ relation is denoted $\arrow^*$.
    If $\Root\arrow^* x$, then we say that address $x$ is {\em path-connected}.

    \item If $x$ and $y$ are nodes of $M$ and $k\in \omega$ (a value that is different from $\infty$ and $-\infty$)
    then $x \arrow_k y$ is the conjunction of the following statements.
    \begin{enumerate}
        \item $\KeyF(x)\neq k$, and 
        \item $k<\KeyF(x) \Rightarrow y=\LeftF(x)$, and 
        \item $k>\KeyF(x) \Rightarrow y=\RightF(x)$.
    \end{enumerate}
    
    \item Given a key value $k$, a $k$-search path in $M$ (or a ``$k$ path'') is a sequence of addresses
    $x_0,\ldots,x_n$ such that $(x_i\arrow_k x_{i+1})^M$ for every $i<n$. 
    $(x\arrow_k^* y)^M$ denotes the transitive and reflexive closure of the $\arrow_k$ relation. If $(x\arrow_k^*y)^M$ we say that $y$ is $k$-connected to $x$ in $M$.
    If $\Root\arrow_k^* x$, then we say that node $x$ is $k$-{\em connected}.
\end{enumerate}
\end{definition}

For any address structure $M$ we define a set of key values $\Set(M)$. 
\begin{definition}[The Set of a structure]
\label{DefSet}
\begin{equation*}
\Set(M)=\{k\in\omega \mid \exists a\, (\Root\arrow_k^* a\, \wedge k=\KeyF(a) \wedge \neg\DeletedP(a) )\}.
\end{equation*}
\end{definition}

Let $\mathit{Stp}$ be the set of steps of the Contention Free algorithm as described in Section~\ref{sec:the-algorithm}.
A {\em history} is a sequence of states (i.e. memory structures) $(M_i\mid {i\in I})$,
where $I$ is either the set $\omega$ of finite ordinal numbers or a finite interval of $\omega$, 
and for every index $i$ and its successor $i+1$ in $I$, $(M_i,M_{i+1})$ is a step in $\mathit{Stp}$.
We are mainly interested in infinite histories $(M_i\mid i\in \omega)$ such that $M_{0}$ is an initial structure state. 

If $l_1,l_2\in\InstructionT$ are instructions, then $\Step(p, l_1, l_2)$ denotes the set of all steps by process $p$ of atomic instruction $l_1$ that have the effect (among other things) of setting $\ControlF(p)=l_2$. For example, $s\in\Step(p,\rom{i3}, \rom{i1})$ says that step $s=(M,N)$ is an execution of an instruction i3 by process $p$ such that takes the \goto\ i1 branch, resulting in $\ControlF(p)^N=l_2$.

If $l\in\InstructionT$ is an instruction, then $\Step(p,l)$ denotes the set of all steps by process $p$ of atomic instruction $l$.
For example, executions of instruction i3 split into those that take the \goto\ i1 branch and those that return to m0: 
$\Step(p, \rom{i3})= \Step(p, \rom{i3, i1})\cup \Step(p, \rom{i3, m0})$.

\subsection{Invariants and step-properties}
In this chapter, we make extensive use of {\em invariants} and {\em step-properties} to prove various claims regarding aspects of the behavior of the CF algorithm:

\begin{definition}
\label{DefInv}
A {\em step-invariant} is a sentence $\sigma$ in $\mathcal{L}_{AS}$ such that for every step 
$(M,N)\in\mathit{Stp}$, $M\models \sigma \Rightarrow N\models \sigma$.

A step-invariant $\sigma$ is said to be an {\em inductive invariant} if it holds in every initial structure.

A sentence $\sigma$ in $\mathcal{L}_{AS}$ that holds in every state of every history sequence (of the CF algorithm) is said to be a {\em valid state property}.
Inductive invariants and their consequences are valid sentences.

A statement about pairs of states $(S,T)$ is said to be a {\em valid step-property} if it is true about every
step of the algorithm.
\end{definition}

\begin{remark}
A step-property is not a step-invariant simply
because a step-property is a property of {\em pairs}
of steps (shared by all steps) whereas a step-invariant
is a property of states (which no step can violate).
\end{remark} 

These definitions of ``step-invariant'', ``inductive invariant'' and ``step-property'' are the usual ones~\cite{lamport1977proving,ohearn2010hindsight}.
We often use the shorthand {\em invariant} instead of step-invariant.

The following is an example of a step-property:
\begin{SP}
\label{sp:immutable-keys}
We assume that the key fields of addresses are immutable. Formally, for any step $(M,N)$ and for any address $x$ in $\AddressT^M$ and in $\AddressT^N$, $\KeyF(x)^M=\KeyF(x)^N$.
\end{SP}

\section{Properties of the Contention-Friendly Algorithm}
\label{sec:invariants-and-props}
In this section, we formulate and prove a myriad of invariants and properties of the CF algorithm. 
The culmination of this section is the presentation of the {\em Regularity} property, and the proof that this property is an invariant of the algorithm.
This is a core component of our work, enabling the proofs in later sections.

Many of the other properties and invariants we prove in this section are necessary for the proof that regularity is an invariant of the algorithm.

We supplement the theoretical work in this section with a bounded model of the algorithm, encoded in TLA+~\cite{engberg1992TLA}, which was used to model-check all of the invariants and step-properties presented in this section. This proved to be quite useful, as demonstrated in footnote~\ref{footnote:pt2} of definition~\ref{DefPot}.
While not a full verification of our proofs (due to the bounded nature of the model), this model-checking process does act to validate the correctness of our proofs.
The model and accompanying invariants and properties can be found at~\cite{thesis-code}.

We begin our journey with a simple inductive invariant, which says that
there is no address that points to itself, except for $\Root=\RightF(\Root)$ and $\bot=\LeftF(\bot)=\RightF(\bot)$.
\begin{inductive-invariant}
\label{inv:address-points-to-root}
\begin{enumerate}
    \item[]
    \item For every address $x$, if $x\not\in\{\Root, \bot\}$ then 
    $x\neq\LeftF(x) \wedge x\neq \RightF(x)$. $\Root\neq \LeftF(\Root)$.
    \item For every address $x\neq\Root$, if $x\arrow \Root$, then either $\RemovedP(x)$ or $\LockedP(\Root,\Sys)$,
    $x=\n$ and $\ControlF(\Sys)\in \{v8, v9\}$.
    \item $\Root=\RightF(\Root)\wedge \bot=\LeftF(\bot)=\RightF(\bot)$.
\end{enumerate}
\end{inductive-invariant}

\begin{remark}
\label{remark:rl-rr-symmetry}
Throughout the proofs in this section we rely on the symmetry of \rotateLeft\ and \rotateRight; we will only prove the claims for the case of \rotateLeft, and omit the nearly-identical proofs for the case of \rotateRight.
\end{remark}

\begin{proof}
The invariant statement is trivially true at the initial state which has only two addresses --- \Root\ and $\bot$.

We note that for any step $s=(M,N)$ such that the functions \LeftF/\RightF\ are the same in $N$ and $M$, the claim holds trivially, since the invariant statement holds trivially (as it is a claim about the functions \LeftF/\RightF).

So let $s=(M,N)$ be any step such that $M$ satisfies the invariant and the step changes the functions \LeftF/\RightF. We have to prove that it holds also in $N$.

If $s$ is a successful execution of instruction i3
by some process $p>0$ (a {\em working} process), and \new\ is the new address
introduced by this step, then arc $(\nd_p,\bot)$ in $M$ is
replaced by arc $(\nd_p,\new)$ in $N$, and since $\new\not\in \{\nd_p,\Root\}$
(as $\nd_p$ and $\Root$ are not new addresses), arc $(\nd_p,\new)$
is neither a self-pointing arc nor a \Root\ pointing arc. If $(x,y)$ is an old arc of $M$ that remains
in $N$, then it is obvious that $y\neq\Root$ by our assumption on $M$.

Suppose next that $s$ is a step by the \Sys\ process that introduces a new node. 
Then $s$ is an execution of instruction f6 or r6. Suppose that $s$ is an execution of instruction f6. 
Then a new node \new\ is created whose left and right children
are nodes $\rl$ and $\ell_0$ which are already in $M$.
Thus $(\new,\rl)$ and $(\new,\ell_0)$ are not
self pointing arcs. But neither are they \Root\ pointing arcs:
By the precondition pr2 of \rotateLeft, $\n\not\in \{\Root,\bot\}$. Node $\n$ points to $\rl$ and to $\ell_0$
in $M$, and so these two nodes are different from the root, since $\ControlF(\Sys)\notin\{v8,v9\}$, which means that only the root can point to itself in $M$.

Finally, suppose that $s$ is a step by the \Sys\ process
that does not introduce a new address, but changes the \LeftF\ 
or the \RightF\ function. Executions of instructions f7 and f8 (as well as
r7 and r8) and v6, v7, and v8 are such steps:

In a step $s$ that executes instruction f7, arc $(\n,\ell_0)$ (due to $\ell_0=\LeftF(\n)$) of $M$ is replaced by arc $(\n,\r)$ 
of $N$ (due to $\r=\LeftF^N(\n)$). But $(\n,\r)$ is an 
arc of $M$ (due to $\r=\RightF(\n)$), and since $\ControlF(\Sys)\notin\{v8,v9\}$, $\r\neq\Root$. The other clauses of the invariant hold trivially.

In a step $s$ that executes instruction f8, arc $(\prt,\n)$ of $M$ is replaced by arc $(\prt,\r)$ of $N$. 
Once again, since $\ControlF(\Sys)\notin\{v8,v9\}$, $\r\neq\Root$. The other clauses of the invariant hold trivially.

In a step $s= (M,N)$ that executes v6, arc $(\prt,\n)$ of $M$ is replaced by arc $(\prt,\child)$ of $N$. 
Since $(\n, \child)^M$, $\child\neq\Root$, and the invariants hold in $N$.

If $s$ executes v7, then arc $(\n,\child)$ is replaced by $(\n,\prt)$, and it is indeed possible that $\prt=\Root$, as the invariant states. 
The arguments in case step $s$ is an execution of v8 are similar.
\end{proof}

It is easy to check (syntactically) that no step reduces the extension of predicate \RemovedP. 
Also, a step that adds a new address (one of i3, f6 and r6) adds an address that is not removed. 
We formalize this as the following trivial step-property:

\begin{SP}
\label{InvRem}
For any step $s=(M,N)$, if address $x$ is removed in $M$,
then it is removed in $N$, and if $x$ is an address
of $N$ but not of $S$ then $\neg\RemovedP(x)^N$.
\end{SP}

\begin{remark}[Using invariants and step-properties as axioms.] 
\label{remark:invs-as-axioms}
In proving that an $\mathcal{L}_{AS}$ sentence $\alpha$ is an invariant we 
use the tools of mathematics, and once $\alpha$ is an established invariant
we may use it as an axiom in proving that other sentences $\beta$ are invariants. That is, when proving
that for any step $(M,N)$, $\beta^M\rightarrow \beta^M$ we may assume that $\alpha$ holds in both $M$
and $N$ and use this assumption in the proof. Surely, this is nothing more than proving that
$\alpha\wedge\beta$ is an invariant, but it brings about clearer proofs. When declaring that $\beta$
is an invariant we usually write in square brackets the
invariants and step-properties on which that proof relies. Likewise, we may
use proven step-properties as axioms when proving newer step-properties. 
\end{remark}

\begin{definition}
\label{Def:ControlDependent}
Many useful invariants have the form $\ControlF(p)=line \rightarrow \varphi$ 
where $line\in\InstructionT$; we say that such valid statements are {\em control-dependent} invariants. 
\end{definition}

Control-dependent invariants are often quite simple to prove. 
We present some such invariants in Figures
\ref{fig:protocol-elementary-invariants} and \ref{cdLM}.
As explained in Remark \ref{remark:invs-as-axioms}, we may use these control-dependent invariants
as axioms when proving other invariants.
The proof that the control-dependent invariants are indeed invariant statements
is simple but not completely trivial. As an example, 
we prove one of the invariants presented in 
Figure~\ref{fig:protocol-elementary-invariants}.

\begin{inductive-invariant}[uses: \ref{sp:immutable-keys}, \ref{InvRem}]
\label{inv:nd=k}
For every process $p>0$,
\[\ControlF(p)=\rom{i1} \rightarrow (\KeyF(\nd_p) = k_p \rightarrow \RemovedP(\nd_p)).\]
\end{inductive-invariant}

\begin{proof}
Let $p>0$ be some working process, and let $s=(M,N)$ by a step of the algorithm such that the invariant holds in $M$. If $\ControlF(p)^N\neq \rom{i1}$, then the claim holds trivially, and specifically in the initial state.

By Step-property \ref{sp:immutable-keys} and Invariant \ref{InvRem}, if $\nd_p^M=\nd_p^N$, then $\KeyF(\nd_p)^M=\KeyF(\nd_p)^N$ and $\RemovedP(\nd_p)^M \iff \RemovedP(\nd_p)^N$. Since $\nd_p$ is a local variable of process $p$, only $p$ is able to modify $\nd_p$. Using all of these facts together leads to the conclusion that if step $s$ is not a step by $p$ and $\ControlF(p)^N=\rom{i1}$, then the invariant holds in $N$.

If $s$ is a step by $p$ such that $\ControlF(p)^N=\rom{i1}$, then there are four possibilities:

\begin{enumerate}
\item $s\in \Step(p,\rom{m0, i1})$, and so $\nd_p^N=\Root$, and $\KeyF(\Root)=\infty$ whereas $k_p^N\in\omega$, so that $\KeyF(\nd_p)\neq k_p$ at $N$. Hence the invariant holds in $N$.
\item $s\in \Step(p,\rom{i1, i1})$. An inspection of this step shows that, since the \goto\ i2 branch was not taken in this execution of instruction i1, then $\KeyF(\nd_p)\neq k_p$ holds, and hence the invariant holds in $N$.
\item $s\in\Step(p, \rom{i2, i1})$. An inspection of this execution of instruction i2 shows that $M\models \RemovedP(\nd_p)$. Since step $s$ does not change the denotation of $\nd_p$ or the predicate \RemovedP, the invariant holds in $N$.
\item $s\in\Step(q, \rom{i3, i1})$. Here we have that $\nd_p^N=\nd_p^S$, $k_p^N=k_p^M$, and $\RemovedP(\nd_p)^N$ if and only if $\RemovedP(\nd_p)^M$. So the invariant holds in $N$.
\end{enumerate} 
\end{proof}

As this proof demonstrated, proving control-dependent invariants is a mostly mechanical process, and can be automated. Indeed, the remainder of these control-dependent invariants are  validated in the supplementary TLA+ specification~\cite{thesis-code}.

\begin{figure}
\centering
\begin{tabular}{|l|l|}\hline
\begin{minipage}[t]{65mm}
\begin{alltt}
\booleanT\ \bfcontains(\(k\))
 c1-2  \(\nd\sb{p}\neq\bot\)
 c1    \(\KeyF(\nd\sb{p})\neq k\sb{p}\)  
 c2    \(\KeyF(\nd\sb{p}) = k\sb{p}\)

\booleanT\ \bfdelete(\(k\))
 d1-2  \(\nd\sb{p}\neq\bot\)
 d1    \(\KeyF(\nd\sb{p})={k\sb{p}}\rightarrow\RemovedP(\nd\sb{p})\)
 d2    \(\LockedP(nd\sb{p},p)\wedge\KeyF(\nd\sb{p})=k\sb{p}\)
\end{alltt}
\end{minipage}
&
\begin{minipage}[t]{65mm}
\begin{alltt}
\booleanT\ \bfinsert(\(k\))
 i1-3  \(\nd\sb{p}\neq\bot\)
 i1    \(\KeyF(\nd\sb{p})={k\sb{p}}\rightarrow\RemovedP(\nd\sb{p})\)
 i2    \(\LockedP(nd\sb{p},p)\wedge\KeyF(\nd\sb{p})=k\sb{p}\)
 i3    \(\LockedP(nd\sb{p},p)\wedge\KeyF(\nd\sb{p})\neq{k\sb{p}} \wedge\)
       \(\nxt\sb{p} =\bot\)

       
       
 
\end{alltt}
\end{minipage}\\
\hline
\end{tabular}
\caption{Control dependent invariants of the working processes.}
\label{fig:protocol-elementary-invariants}
\end{figure}

\begin{figure}
\centering
\begin{tabular}{|l|}\hline
\begin{minipage}[t]{0.9\textwidth}
\begin{alltt}
\booleanT\ \rotateLeft(\(\prt,\lft\)):
 f6-9 \(\prt\neq\bot \wedge \neg\RemovedP(\prt) \wedge \n\neq\bot \wedge \r=\RightF(\n)\neq\bot \wedge\)
      \(\neg\RemovedP(\r) \wedge \neg\RemovedP(\n) \wedge \KeyF(\n)\neq\KeyF(\prt) \wedge \n\neq\prt \wedge\)
      \(\LockedP(\prt,\Sys) \wedge \LockedP(\n,\Sys) \wedge \LockedP(\r,\Sys)\)
 f6-8 \(\n=\LR(\prt,\lft)\)
 f6   \(\rl=\LeftF(\r) \wedge \ell\sb{0}=\LeftF(\n)\)
 f7-9 \(\KeyF(\new)=\KeyF(\n) \wedge \new\neq\n \wedge \neg\RemovedP(\new) \wedge\)
      \(\DeletedP(\new)\iff\DeletedP(\n) \wedge \LeftF(\new)=\ell\sb{0} \wedge \RightF(\new)=\rl \wedge\)
      \(\LeftF(\r)=\new \)
 f7   \(\LeftF(\n)=\ell\sb{0}\)
 f8-9 \(\LeftF(\n)=\r\)
 f9   \(\r=\LR(\prt,\lft)\)
 
\end{alltt}
\end{minipage}\\
\hline
\begin{minipage}[t]{0.9\textwidth}
      \begin{alltt}
\booleanT\ \rotateLeft(\(\prt,\lft\)):
 r6-9 \(\prt\neq\bot \wedge \neg\RemovedP(\prt) \wedge \n\neq\bot \wedge \ell\sb{0}=\LeftF(\n)\neq\bot \wedge\)
      \(\neg\RemovedP(\ell\sb{0}) \wedge \neg\RemovedP(\n) \wedge \KeyF(\n)\neq\KeyF(\prt) \wedge \n\neq\prt \wedge\)
      \(\LockedP(\prt,\Sys)\wedge\LockedP(\n,\Sys) \wedge \LockedP(\ell\sb{0},\Sys)\)
 r6-8 \(\n=\LR(\prt,\lft)\)
 r6   \(\lr=\RightF(\ell\sb{0})\wedge\r=\RightF(\n)\)
 r7-9 \(\KeyF(\new)=\KeyF(\n) \wedge \new\neq\n \wedge \neg\RemovedP(\new) \wedge\)
      \(\DeletedP(\new)\iff\DeletedP(\n) \wedge \RightF(\new)=\r \wedge \LeftF(\new)=\lr \wedge\)
      \(\RightF(\ell\sb{0})=\new\)
 r7   \(\RightF(\n)=\r\)
 r8-9 \(\RightF(\n)=\ell\sb{0}\)
 r9   \(\ell\sb{0}=\LR(\prt,\lft)\)
 
\end{alltt}
\end{minipage}\\
\hline
\begin{minipage}[t]{0.9\textwidth}
\begin{alltt}
\booleanT\ \bfremove(\(\prt,\lft\)):
 v6-9 \(\prt\neq\bot \wedge \neg\RemovedP(\prt) \wedge \n\neq\bot \wedge \neg\RemovedP(\n) \wedge\)
      \(\LockedP(\n,\Sys)\wedge\LockedP(\prt,\Sys) \wedge \DeletedP(\n) \wedge \n\neq\prt\)
 v6   \(\n=\LR(\prt,\lft)\wedge\n\arrow\child \wedge \n\arrow\bot\) 
 v6-7 \((\LeftF(\n)=\bot\rightarrow\RightF(\n)=\child)\wedge(\LeftF(\n)\neq\bot\rightarrow\LeftF(\n)=\child)\)      
 v7-8 \(\n\arrow\child \wedge \prt\arrow\child \wedge \prt\notarrow\n\)
 v8-9 \(\n\arrow\prt\)
 
\end{alltt}
\end{minipage}\\
\hline
\end{tabular}

\caption{Control-dependent invariants of the \Sys\ process.}
\label{cdLM}
\end{figure}

Two simple observations can be made about steps that change the \LeftF/\RightF\ functions:
\begin{enumerate}
    \item For any address $x$, a step can only change either $\LeftF(x)$ or $\RightF(x)$ of only a single address $x$, but not both\footnote{A step in which a new node
    is introduced, for example an execution of f6, expands the domain of the \LeftF\ or the \RightF\ function, but {\em changes} the value
    of just one argument.}. Such a step does not change the truth value of $\RemovedP(x)$ or $\DeletedP(x)$.
    
    \item A working process $p>0$ can change these functions only when executing instruction i3 of the \bfinsert\ operation. 
    In this atomic step, process $p$ creates a new node and assigns it to $\LeftF(\nd_p)$ or to $\RightF(\nd_p)$. 
    Observe that this change is from $\LeftF(\nd_p)=\bot$ to $\LeftF(\nd_p)\neq \bot$, 
    or else from $\RightF(\nd_p)=\bot$ to $\RightF(\nd_p)\neq \bot$. 
\end{enumerate}

We formalize this observation in the following step-property:

\begin{SP}
\label{Cl3.3}
Let $(M,N)$ be a step by process $p>0$.
Suppose that for some three distinct addresses $x$, $y$ and $z$: 
$\LeftF^M(x)=y \wedge \LeftF^N(x)=z$ or $\RightF^M(x)=y \wedge 
\RightF^N(x) =z$.
Then $(M,N)$ is an execution of instruction \rom{i3}, $y=\bot$, $M\models x =\nd_p \wedge \neg\RemovedP(x)$.
And in $N$, $z$ is a new node, and $N\models \LeftF(z)=\RightF(z)=\bot$.
\end{SP}

The proof of this step-property is similar to the proof of Invariant \ref{inv:nd=k} above, relying on syntactic reasoning, 
on Invariant \ref{InvRem}, and on the fact that local variables of a process can only be modified by that process.
Validation of this step-property is also included in the accompanying repository~\cite{thesis-code}.

\begin{inductive-invariant}[uses: \ref{InvRem}, \ref{Cl3.3}]
\label{inv:rem-chld-not-bot}
\[\forall x (\RemovedP(x)\rightarrow \LeftF(x)=\RightF(x)\neq\bot).\]
\end{inductive-invariant}

\begin{proof}
Since the initial state contains only two addresses, \Root\ and $\bot$ which are not removed, it is obvious that the initial state
satisfies invariant. We have to prove that the invariant is preserved by every step $s=(M,N)$.   

So assume that the invariant holds in $M$,
and let $x$ be an address in $N$ such that $\RemovedP^N(x)$. 
By Invariant \ref{InvRem}, $x$ is not a new address of $N$. Thus $x$ is an
address in $M$, and either: 
\begin{enumerate*}[label=(\arabic*)]
    \item $x$ is removed in $M$, or
    \item $x$ is not removed in $M$.
\end{enumerate*}
\begin{enumerate}
    \item Assume that $\RemovedP^M(x)$, and so 
    \[
        M\models \RemovedP(x)\wedge \LeftF(x)=\RightF(x)\neq\bot.
    \]
    We check that there is no step $s=(M,N)$ that changes $\LeftF(x)$ or $\RightF(x)$.
    The instructions that may change $\LeftF(x)$ or $\Right(x)$
    are i3 (affecting $\nd_p$), f6 (affecting $\LeftF(\r)$),
    f7 (affecting $\LeftF(\n)$, and f8 (affecting $\LeftF(\prt)$
    or $\RightF(\prt)$) (as well as the respective \rotateRight).
    \begin{enumerate}
        \item $s$ cannot be an execution of i3: by Step-property \ref{Cl3.3}, $\neg\RemovedP(x)$ when $x=\nd_p$ for some working process $p>0$.
        \item An execution of f6 changes $\LeftF^M(\r)$ from $\rl$ (in $M$) to $\new$ (in $N$). However, $\r$ is not removed
        in $S$ because of precondition pr3 (see control-dependent invariants f6-9 in Figure~\ref{cdLM}).
        \item An execution of f7 changes $\LeftF^M(\n)$, but $\n$ is not removed in $M$ (see control-dependent invariants f6-8 in Figure~\ref{cdLM}). 
        \item An execution of f8 changes $\LeftF^M(\prt)$ (or $\RightF^M(\prt)$) but again, $\prt$ is not removed in $M$ (see control-dependent invariants in Figure~\ref{cdLM}).
    \end{enumerate}
    
    \item Assume next that $\neg\RemovedP(x)$ in $M$, and hence step $s=(M,N)$ is a removal step. There are three removal steps
    and they are all by the \Sys\ process: f9, r9, and v9, and they all remove node $\n$. 
    In $N$, $\LeftF(\n)=\RightF(\n)=\r, \ell_0,\prt$ which are all not $\bot$ as can be gathered from the control-dependent invariants in Figure~\ref{cdLM}.
    \end{enumerate} 
\end{proof}

\begin{observation}[uses: \ref{inv:rem-chld-not-bot}]
\label{observation:rem-children-constant}
We already made the syntactic observation that there is no step that changes both \LeftF\ and \RightF\ at once. It follows immediately from Invariant~\ref{inv:rem-chld-not-bot} that for any step $(S, T)$ and $x\in\AddressT^S$, $\RemovedP^S(x)\rightarrow  \LeftF^S(x)=\LeftF^T(x) \wedge \RightF^S(x)=\RightF^T(x)$.
\end{observation}

\begin{corollary}[uses: \ref{Cl3.3}, \ref{observation:rem-children-constant}]
\label{cor:i3-not-rem}
The combination of Step-property \ref{Cl3.3} and Observation \ref{observation:rem-children-constant} implies that if a working process $p>0$ executes a step $(S,T)\in\Step(p, \rom{i3, m0})$, then $nd_p$ is not removed in $S$ or in $T$, since a removed node does not have $\bot$ as a child, while $\nd_p$ must have $\bot$ as a child in order for $(S,T)$ to execute.
\end{corollary}

\begin{definition}
\label{def:focused}
An address $a$ is {\em focused} if $\LeftF(a)=\RightF(a)\neq \bot$.
\end{definition}

Thus $\bot$ and \Root\ are not focused, but as Figure \ref{cdLM} shows, $\n$ is
focused when $\ControlF(\Sys)\in\{\rom{f8, f9, r8, r9}\}$.

\begin{corollary}[uses: \ref{InvRem}, \ref{inv:rem-chld-not-bot}, \ref{observation:rem-children-constant}]
\label{cor:constantly-focused}
In any history sequence, once a node is removed, it stays removed and
its left and right children are equal, do not change, and are not $\bot$, i.e., a removed address is constantly focused.
\end{corollary}

\begin{SP}[uses: \ref{Cl3.3}]
\label{Cl4.9}
For any step $s=(M,N)$ and address $x\neq \bot$, if $(\Root\arrow^* x)$ in $M$ and $(\Root\notarrow^* x)$ in $N$,
then $x=\n$ in both $M$ and $N$, and $s\in \Step(\Sys,\rom{f8}) \cup \Step(\Sys,\rom{r8}) \cup \Step(\Sys,\rom{v6})$.
\end{SP}

\begin{proof}
Let $s=(M,N)$ be a step and $x\neq \bot$ an address that is path-connected in $M$ but not in $N$, i.e., $s$ is a step that modifies \LeftF/\RightF.
We must prove that $s$ is an execution of one of the instructions f8, r8, v6, and that $x=\n$ in $M$ and $N$. 

First, note that $s$ is not a step by some working process $p>0$ or else
we would find that $b$ is $\bot$ (by Step-property \ref{Cl3.3}).

Thus $s$ is one of the steps by the \Sys\ process that modifies \LeftF/\RightF.
We check below all such steps that are not executions of the f8, r8, v6 instructions,
and find that none could make the disconnecting mutation.
(Illustrations \ref{Fig:rotate-left} and \ref{Fig:remove} can be consulted while following our proof.)

\begin{enumerate}
    \item $s\in \Step(\Sys, \rom{f6},\rom{f7})$. 
    Then any node that is path-connected in $M$ remains
    path-connected in $N$. The reason is that the only arc of $M$ that is lost in $N$
    is $(\r, \rl)$, which is replaced by to $(\r, \new)$. 
    However, since $s$ also adds arc $(\new,\rl)$, path $\r\arrow\new\arrow\rl$ connects $\r$ and $\rl$ in $N$.
    
    \item  $s\in \Step(\Sys,\rom{f7},\rom{f8})$. 
    Here the lost arc $(\n, \ell_0)$ of $M$ is compensated for by the path $\n\arrow \r\arrow \new\arrow\ell_0$ of $N$. 

    \item If $s\in \Step(\Sys,\rom{v7})$, then an arc $(\n,\bot)$ of $M$ is replaced by to $(\n, \prt)$ of $N$. 
    Since $\bot$ can only lead to $\bot$, and $x\neq\bot$, then $(\nd,\bot)$ is not on path $\Root\arrow^*x$ in $M$.
    This implies that in this case $\Root\arrow^*x$ would hold in $N$, making the step irrelevant. 

    \item If $s\in \Step(\Sys,\rom{v8})$, then the lost of arc $(\n,\child)$ of $M$ is replaced by arc $(\n,\prt)$ of $N$. 
    By the control-dependent invariants of Figure \ref{cdLM}, we have that $\prt\arrow\child$, and so the path $\n \arrow \prt \arrow \child$ is in $N$.

\end{enumerate}

Thus, $s$ may only be a step in $\Step(\Sys,\rom{f8})$, $\Step(\Sys,\rom{r8})$, and $\Step(\Sys,\rom{v6})$.

We first prove that for every node $y\neq\n$, $s$ does not disconnect $y$.
Let $P$ be the shortest path from $\Root$ to $y$ in $M$. We prove the claim for the possibilities for $s$:

\begin{enumerate}
\item Suppose that $s\in \Step(\Sys, \rom{f8})$, in which arc $(\prt, \n)$ is replaced by arc $(\prt, \r)$. 
If $(\prt, \n)$ is not an arc of $P$, then every arc of $P$ remains in $N$, and hence $y$ is path-connected in $N$. 
If $(\prt, \n)$ is an arc on $P$, then it is not the last arc (because $y\neq \n$ is the last node of $P$).
Since $\LeftF(\n)=\RightF(\n)=\r$ in $M$ and in $N$ by the control-dependent invariants of Figure \ref{cdLM}, 
$\r$ must be the successor of $\n$ in $P$. So the sub-path $\prt\arrow \n \arrow \r$ appears in $P$.
Since arc $(\prt, \r)$ in $N$ replaces the sub-path $\prt\arrow \n\arrow \r$ of $P$, and we see that $y$ remains
path connected in $N$. 

\item Suppose that $s\in \Step(\Sys,\rom{v6})$. 
In this step arc $(\prt, \n)$ of $M$ is replaced by arc $(\prt, \child)$ of $N$.
If $(\prt, \n)$ is not an arc of path $P$, then every arc of $P$ remains in $N$, and hence $y$ is path-connected in $N$.
If $(\prt, \n)$ is an arc of $P$, then it is not the last arc (because $y\neq \n$ is the last node of $P$).
The children of $\n$ are $\bot$ and $\child$, and $\n\neq\prt$ in $M$ and in $N$. 
Since $\bot$ is not on $P$, $(\n,\child)$ is in $P$. The sub-path $\prt\arrow \n\arrow \child$
of $P$ is replaced by the arc $(\prt, \child)$ in $N$, and we see that $y$ remains
path connected in $N$. 
\end{enumerate}
\end{proof}

\begin{remark}
We will prove (in Lemma \ref{Lem4.14}) that $\bot$ is always path-connected in
any state of any history. So, in applications of Step-invariant \ref{Cl4.9}, assumption 
$x\neq\bot$ is not really necessary, but we are not yet in a position to prove this.
\end{remark}

\begin{definition}
\label{Def:preRemoved}
Address $x$ is {\em pre-removed} if $x\neq \bot \wedge \neg \RemovedP(x) \wedge \neg(\Root\arrow^* x)$.
\end{definition}

\begin{inductive-invariant}[uses: \ref{InvRem}, \ref{Cl3.3}, \ref{Cl4.9}]
\label{inv:pre-rem}
\[
    \forall\, x (\preRemovedP(x) \rightarrow  x=\n \wedge \ControlF(\Sys)\in\{\rom{f9, r9, v7, v8, v9}\})
\]
\end{inductive-invariant}

\begin{proof}
Let $s=(M,N)$ be a step such that $M$ satisfies our invariant. Suppose that $x_0$ is
an address of $N$ that is pre-removed in $N$, i.e. 
\begin{equation}
\label{Eq10}
   N\models (x_0\neq \bot \wedge \neg\RemovedP(x_0)\wedge \Root\notarrow^* x_0).
\end{equation}
By Step-property \ref{InvRem}, since $x_0$ is not removed in $N$ it cannot be removed in $M$.
So either $x_0$ is not in $M$, or else it is in $M$ and not removed.
We have to prove that 
\begin{equation}
\label{Eq11}
    M\models(\ControlF(\Sys)\in\{\rom{f9, r9, v7, v8, v9}\} \wedge x_0=\n).
\end{equation}
There are four cases to check: 
\begin{enumerate*}[label=(\arabic*)]
    \item $s$ is a step by some working process $p>0$ and $x_0$ is an address in $M$,
    \item $s$ is a step by $p>0$ but $x_0$ is a new address in $N$,
    \item $s$ is a step by process \Sys, and $x_0$ is in $M$, and
    \item $s$ is a step by \Sys\ and $x_0$ is a new address in $N$. 
\end{enumerate*}

Assume first that $s$ is a step by process $p>0$ and $x_0$ is an address of $M$.
There are two possibilities:
\begin{enumerate}
    \item If $x_0$ is not path-connected in $M$
    then it is pre-removed (since $\neg\RemovedP^M(x_0)$). Since $M$ satisfies the
    invariant, $M \models(\ControlF(\Sys)\in\{\rom{f9, r9, v7, v8, v9}\} \wedge x_0=\n)$. 
    But a step by $p>0$ does not change the denotation of $\n$
    or the value of $\ControlF(\Sys)$, and hence (\ref{Eq11}) as required.
    
    \item If $x_0$ is path-connected in $M$, then since $x_0\neq\bot$, the path
    $P$ from $\Root$ to $x_0$ in $M$ does not contain $\bot$. Thus step $s$ changes no arc of
    $P$ (by Step-property \ref{Cl3.3}), and hence
    $x_0$ remains path-connected in $N$ which contradicts
    our assumption in (\ref{Eq10}).
\end{enumerate}
 
Assume secondly that $s$ is a step by process $p>0$ but that $x_0$ is a new address of $N$.
So $s=(M,N)$ is an execution of i3, $x_0=\new$ is the new address in $N$, $\nd_p\neq \bot$ points to \new\ in $N$. 
Since $x_0=\new$ is not path-connected in $N$, $\nd_p$ is not path-connected
in $N$ (for $\nd_p$ is the sole node that points to $\new$ in $N$). 
It follows that already in $M$, $\nd_p$ is not 
path-connected, or else the path of $M$ from \Root\ to $\nd_p$ 
remains a path in $N$ (again by Step-property \ref{Cl3.3} as above). 
But $\nd_p$ is not removed in $M$ by Corollary \ref{cor:i3-not-rem}.
Thus $\nd_p$ is pre-removed in $M$, and the
invariant $\ControlF^M(\Sys)\in \{\rom{f9, r9, v7--9}\}$
implies that $\LockedP^M(\nd_p,\Sys)$, which cannot be the case
as step $s$, an execution of i3, requires that $\LockedP(\nd_p,p)$.

Assume thirdly that $s$ is a step by \Sys, and that $x_0$ is an address in $M$.
There are two possibilities:
\begin{enumerate}
    \item $x_0$ is not pre-removed in $M$, and so $(\Root\arrow^* x_0)^M$
    (since $x_0$ is not removed in $M$).
    By the assumption at (\ref{Eq10}) $x_0$, is not path-connected in $N$. 
    So by Step-property \ref{Cl4.9}, $s$ is an execution of f8, r8, or v6. 
    Thus $\ControlF(\Sys)\in \{\rom{f9, r9, v7}\}$ in $N$, 
    and again by \ref{Cl4.9}, $x_0=\n$ in $N$.

    \item $x_0$ is pre-removed in $M$, and so $M\models \ControlF(\Sys)\in\{\rom{f9, r9, v7, v8, v9}\} \wedge x_0= \n$:
    \begin{enumerate}
        \item If $\ControlF^M(\Sys)=\rom{f9}$, then $\RemovedP^N(x_0)$ (where $x_0=\n$), and the claim holds trivially.
        
        \item If $\ControlF^M(\Sys)=\rom{v7}$, then the effects of the step are $\ControlF^N(\Sys)=\rom{v7}$
        and arc $(\n,\bot)$ of $M$ being replaced with arc $(\n,\prt)$ in $N$. $x_0=\n$ remains not
        path-connected in $N$, and the connectivity of any other node $y$ of $N$ is unaffected, because if $y$ is
        path-connected in $M$, the path from \Root\ to $y$ in $M$ remains a path in $N$ (because $x_0$ is not on that path).
        
        \item If $\ControlF^M(\Sys)=\rom{v8}$, then the effect of step $s$ is that arc $(\n,\child)$
        of $M$ is replaced with arc $(\n,\prt)$ in $N$, and the invariant holds in $N$, similar to the previous case.
        
        \item If $\ControlF^M(\Sys)=\rom{v9}$, then $\n$ is removed in $N$, and the claim holds trivially.
    \end{enumerate}
\end{enumerate}

Finally, suppose that $s$ is a step by \Sys\ and $x_0=\new$ is a new node in $N$. Thus, $s$ is an execution of f6 (or r6).
At state $M$ we have that $\r=\RightF(\n)$, and at state $N$ we have that $\r=\RightF(\n)\wedge \new=\LeftF(\r)$. 
Since $\ControlF(\Sys)=\rom{f6}$ in $M$, the invariant implies that no address of $M$ is pre-removed. 
In particular $\r$ is not pre-removed, and since it is not removed at $M$ 
(by the control-dependent invariants of Figure \ref{cdLM}), $\r$ is path-connected there. 
If $P$ is the shortest path in $M$ from \Root\ to $\r$, 
then $P$ remains a path in $N$ (since the only arc of $M$ that is removed by $s$ is $(\r,\rl)$). 
Since $\r$ is path-connected in $N$, then \new\ is also path-connected
in $N$ and hence is not pre-removed, which contradicts our assumption on $x_0$.
\end{proof}

\begin{corollary}
\label{lemma:pdc-rem-or-locked}
If $\ControlF(\Sys)\not\in \{\rom{f9, r9, v7, v8, v9}\}$ and $x\neq\bot$ is any address that is
not removed, then $x$ is path-connected. 
In particular, if $\Root\notarrow^*x$, then either $\RemovedP(x)$ or $\LockedP(x, \Sys)$.
\end{corollary}

\begin{definition}
\label{Def:confluence}
Node $x$ is {\em confluent} in a state if $x\neq \bot$ and there are two path-connected
nodes that both point to $x$. That is, for some nodes $y\neq z$, both $y$ and
$z$ are path-connected, and $y\arrow x\wedge z\arrow x$.
\end{definition}

\begin{example}
\label{exmp:f7-conf}
If $\ControlF(\Sys)=\rom{f7}$ then $\ell_0$ is confluent if $\ell_0\neq\bot$.
\end{example}
\begin{proof}
Assume that $\ControlF(\Sys)=\rom{f7}$. Then both $\n$ and \new\ point to $\ell_0$.
If $\ell_0\neq \bot$, then it suffices to prove that
$\n$ and $\new$ are path connected in order to deduce
that $\ell_0$ is confluent: $\n$ and $\new$
are not removed (this is a control-dependent invariant). Hence $\n$ and $\new$ are not pre-removed
(by \ref{inv:pre-rem}) and thus, are path-connected.
\end{proof}

\begin{definition}
\label{def:descendants}
Given a state $M$ we have the following definitions.
\begin{enumerate}
    \item The {\em descendants} of any path-connected node $x$ are the set of all
    nodes $y\not\in \{x, \bot\}$ that are reachable from $x$:
    $\Des(x)=\{ y\not\in\{x,\bot\}\mid x\arrow^* y\}$

    \item For any node $x$ we define the set of its left and right descendants:
    \begin{equation}
    \label{EqLeftDes}
    \LeftDes(x) = \{ y\neq\bot\mid \LeftF(x)\arrow^* y\}
    \end{equation}
    \begin{equation}
    \forall x\neq\Root,\ \RightDes(x) = \{ y\neq\bot\mid \RightF(x)\arrow^* y\}      
    \end{equation}

    Define $\RightDes(\Root)=\emptyset$.
    Note that $\LeftDes(\bot)=\RightDes(\bot)=\emptyset$. 

\end{enumerate}
\end{definition}

Observe that $x\not\in \LeftDes(x)$ (and likewise $x\not\in\RightDes(x)$) unless $\LeftF(x)\arrow^* x$, which indicates
a cycle\footnote{We will see in Lemma \ref{Lem4.14} that regular states have no cycles.} in the $\arrow^*$ relation.
Thus, assuming that there are no cycles, the set of descending nodes of $x$ is $\Des(x) = \LeftDes(x)\cup \RightDes(x)$.

\begin{definition}
\label{def:properly-located}
In a state $M$, a node $x$ is {\em properly-located} with respect to another node $y$ if the following conditions hold:
\begin{enumerate}
    \item $x\in\Des(y)$ 
    \item if $x\in\LeftDes(y)$ then $\KeyF(x)<\KeyF(y)$
    \item if $x\in\RightDes(y)$ then $\KeyF(x)>\KeyF(y)$
\end{enumerate}
\end{definition}

\begin{definition}
\label{Def:tree-like}
In a state $M$, a node $x$ is {\em tree-like} if $x\in\{\Root, \bot\}$, if $\Des(x)=\emptyset$ (i.e., both children of $x$ are $\bot$), or if for every $y\in\Des(x)$, $y$ is properly-located with respect to $x$.
\end{definition} 

A cycle in $\arrow$ is a sequence of path connected addresses,
$a_1,\ldots,a_n$ such that $a_1=a_n$,  for every $i<n$
$a_i\arrow a_{i+1}$, and $a_i\neq a_j$ for any indexes $i<j<n$.

\begin{observation}
\label{Lem4.10}
If $x\not\in\{\Root,\bot\}$ is a tree-like node, then $x$ is not a node in a cycle.
Thus if all path-connected nodes are tree-like, then there is no cycle of path-connected nodes except
for the trivial cycles $\{\Root, \Root\}$ and $\{\bot, \bot\}$.
\end{observation}

\begin{lemma}
\label{lemma:tree-like-path}
Let $P$ be a path from node $x$ to node $y$. If
all the nodes on $P$ that precede $y$ are tree-like, then
$P$ is a $\KeyF(y)$-search path.
\end{lemma}
\begin{proof}
Let $P$ be a path from node $x$ to $y$. By Definition \ref{Def:tree-like}, $y$ is properly located with respect to every node along the path $P$, and so $P$ is a $\KeyF(y)$-search path by definition.
\end{proof}

\begin{corollary}
\label{cor:tree-like-path}
If every path-connected node is tree-like, then:
\begin{enumerate}
    \item There are no confluent nodes.
    \item For every path-connected node $x\neq \bot$, there is a single path $P$ from the root to $x$.
    \item No two path-connected nodes have the same key.
\end{enumerate}
\end{corollary}
\begin{proof}
Assume that every path-connected node is tree-like.
\begin{enumerate}
    \item Assume for a contradiction that there exists a confluent node $x$, and let its two distinct path-connected parents be $y$ and $z$. Then $y$ and $z$ must have a common path-connected ancestor $a$ such that $x\in\LeftDes(a)$ and $x\in\RightDes(a)$, and so $x$ is not properly-located with respect to $a$, contradicting that $a$ is tree-like.
    \item This trivially follows from the fact that there are no confluent nodes if every path-connected node is tree-like.
    \item If two different path-connected nodes $x_1$ and $x_2$ had the same key values, then there would be some path-connected $y$ such that either $x_1$ or $x_2$ would not be properly-located with respect to $y$.
\end{enumerate}
\end{proof}

\begin{definition}[Potential Connectivity.]
\label{DefPot}
Given a state $M$, we say that node $x$ is {\em potentially} $k$-connected in $M$ (where $k\in \omega$)
if one of the following three conditions holds in $M$.
\end{definition}
\begin{description}
\item[$\rom{PT1}(x,k)$]
$\equiv\, x$ is $k$-connected.

\item[$\rom{PT2}(x,k)$] $\equiv$ 
\begin{enumerate*}[label=(\arabic*)]
    \item $x$ is pre-removed, 
    \item there is a node $y$ such that $x\arrow y$ and $y$ is $k$-connected\label{item:pt2}\footnote{\label{footnote:pt2}Item \ref{item:pt2} of PT2 was originally ``if $x\arrow_k y$ then $y$ is $k$-connected''. However, this makes Step-property \ref{sp:pdc-remains-pkc} incorrect and unprovable. We did not notice the issue on our own, but, fortunately, the model-checking process we carried out with TLA+ flagged the problem, allowing us to correct the definition of PT2, and maintain the correctness of our proof.}, and 
    \item $\prt$ is $k$-connected.
\end{enumerate*}

\item[$\rom{PT3}(x,k)$] $\equiv\,$
$x$ is removed, and for some $d\geq 0$ there is a sequence
of removed nodes $t_0,\ldots,t_d$ such that $t_0=x$, 
$t_i\arrow t_{i+1}$ for $i<d$, and 
if $y$ is  such that $t_d \arrow y$
then $y$ is potentially connected but is not removed, i.e.,
either PT1$(y,k)$ or PT2$(y,k)$ holds\footnote{Note that
$y=\LeftF(t_d)=\RightF(t_d)$ since $t_d$ is removed
and by Invariant \ref{inv:rem-chld-not-bot}.}.
\end{description}

Intuitively, the notion of potential connectivity captures the 
idea that traversals do not ``get lost'':
When a process $p$ executes one of the $\bfcontains(k_p)$, $\bfdelete(k_p)$,
or $\bfinsert(k_p)$ operations, then we may be tempted to expect that, while $p$ is in midst of its search, 
node $\nd_p$ is on the path from the root to the address with key value $k_p$, if
there is one, or on the $k$-path from the root to $\bot$ if there is none.
Yet this is certainly not the case: process $p$ may reach some $\nd_p$ that becomes a
non path-connected node while $p$ is still there. Process $p$
however is not lost and does not have to abort, it may continue and in a finite number of steps reach an address
that is $k_p$-connected.

\begin{observation}
\label{lemma:pt3}
\begin{enumerate}
    \item[] 
    \item If PT3$(x,k)$ and $x\arrow a$, then $a$ is potentially $k$-connected.
    \item If $x$ is potentially $k$-connected and $x\arrow_k a$, then $a$ is potentially $k$-connected.
\end{enumerate}
\end{observation}

Regularity of a state is the central definition of this section. It represents the notion of a ``valid'' state during the execution of the algorithm.
This is represented by three properties. The first two deal with the nodes $\nd_p$, $\nxt_p$ and $\prt$ may refer to, and the relationships between them.
Due to the concurrent nature of the algorithm, as we discussed in Section~\ref{sec:the-algorithm}, even when considering only the path-connected addresses, 
the nodes in a state of the algorithm often do not constitute a binary-tree.
The last property of regularity covers the specific manner in which the structure of the path-connected section of the graph may deviate from the binary-tree structure.
\begin{definition}[Regularity]
\label{DefReg}
A state is said to be {\em regular} if the following conditions hold.
\end{definition}

\begin{enumerate}
\item[R0.] $\prt$ is $\KeyF(\n)$-connected, and if $\n\neq\prt$ then $\lft\rightarrow\KeyF(\n)<\KeyF(\prt)$ and $\neg\lft\rightarrow\KeyF(\n)>\KeyF(\prt)$.
\item[R1.] For every process $p$ with $p>0$, node $\nd_p$ is potentially $k_p$-connected, and if $\nxt_p\neq\bot$ then $\nxt_p$ is also potentially $k_p$-connected. 
\item[R2.] If $x\not\in\{\Root,\bot\}$ is a path-connected node that is not tree-like, then $x=\n$ and \linebreak{}$\ControlF(\Sys)\in \{\rom{f7, f8,  r7, r8}\}$.
\end{enumerate}

In the claims and proofs that follow, we use R1$(p)$ to denote the instantiation of the universal statement R1 with some process $p>0$.

\begin{lemma}
\label{Lem4.14}
In any regular state, there are no cycles in the $x\arrow y$ relation on the path-connected nodes (except for
the cycles $\bot\arrow \bot$ and $\Root\arrow \Root$).
In particular, for every path-connected node $x$, $x\neq \bot \rightarrow x\neq \LeftF(x)$ and  $x\not\in \{\bot,\Root\}\rightarrow x\neq \RightF(x)$. 

As a consequence, $\bot$ is path-connected. In fact $\Root \arrow^*_{-\infty}\, \bot$.
\end{lemma}
\begin{proof}
By R2, if $x$ is any path-connected node that is not tree-like, then $x=\n$. So there is
at most one path-connected node that is not tree-like. Hence if there is a cycle
of more than one node, then the cycle contains a tree-like node and that is impossible.
In the case of a cycle of a single node, the cycle must be $(\n,\n)$, so either
$\n=\LeftF(\n)$ or $\n=\RightF(\n)$. 
However, this is not the case by Inductive invariant \ref{inv:address-points-to-root}.

The conclusion that $\bot$ is path-connected relies on the assumption that the set of nodes is finite.
Starting with the root and following an arbitrary path (or the $\arrow_{-\infty}$ path)
we must reach $\bot$ and stop, or else a cycle is formed.
\end{proof}

\begin{SP}[uses: \ref{inv:pre-rem}, \ref{cor:tree-like-path}, \ref{Lem4.14}]
\label{Cm2}
For any step $s=(M,N)$ such that $M$ is regular, and for any address $a\neq\bot$ in $M$, if 
$(\Root\arrow^*_k a)^M$ then $((\Root\arrow^* a)\rightarrow (\Root\arrow^*_k a))^N$.
\end{SP}

\begin{proof}
Let $M$ be an arbitrary regular state, and suppose that $k\in\KeyT$ and address $a$ is $k$-connected in $M$.
Let $P$ be the shortest $k$-path of $M$ that leads from $\Root$ to $a$. 
Thus $a$ appears on $P$ only as its last node, and since $a\neq \bot$, the bottom node $\bot$ is not on $P$.

If no arc of path $P$ is removed by our step $(M,N)$, then $P$ remains a $k$-path in $N$ from the root to $a$, and the claim holds trivially.
Thus we may assume that there exists a single arc $(x,y)$ on $P$ such that either $y=\LeftF^M(x)\neq\LeftF^N(x)$ or $y=\RightF^M(x)\neq\RightF^N(x)$.

Since $P$ is a $k$-path in $M$ and $(x,y)$ is an arc on $P$, 
\begin{equation}
\label{Eq8}
(y=\LeftF^M(x) \Rightarrow k<\KeyF(x))\ \text{and}\ (y=\RightF^M(x) \Rightarrow k> \KeyF(x)).  
\end{equation}

In what follows, we prove for each step $(M,N)$ that mutates the function \LeftF/\RightF, that 
the following disjunction holds in $N$:
\begin{equation}
\label{EQ9}
N\models (\Root\arrow_k^* a) \vee (\Root \notarrow^* a).
\end{equation}
That is, either $a$ remains $k$-connected or else $a$ is not even path connected in $N$. 
We denote with $(x,z)$ the arc that replaces $(x,y)$ in $N$.

The only kind of step by a working process $p>0$ that changes \LeftF/\RightF\ is step i3, and
in that case, $y$ is necessarily the $\bot$ node in $M$ (and $z$ is the newly inserted
node). But this cannot be the case since $y$ is on $P$ but $\bot$ is not.

All other steps that mutate \LeftF/\RightF\ are in operations by process \Sys.
Such steps are of kinds f6--8 (and the corresponding r6--8), or v6--8.
In the remainder of this proof, we assume without loss of generality that $\lft=\TRUE$, i.e. we shall deal with steps in $\rotateLeft(\prt,\TRUE)$ and 
$\bfremove(\prt,\TRUE)$.

\begin{description}
    \item[$s\in\Step(\Sys,\rom{f6})$:] As can be observed in Figure \ref{LM}, 
    in this case arc $(x,y)=(\r,\rl)$ of $M$ is replaced by arc $(x,z)=(\r,\new)$ of $N$. 
    
    Note that $\prt$ is path-connected in $M$ since it is neither removed nor pre-removed (by Inductive invariant \ref{inv:pre-rem}).
    By the control-dependent invariants of Figure \ref{cdLM}, we have that $\prt\arrow\n\arrow\r$, and so $\r$ is also path-connected in $M$.
    Let $Q$ be a path in $M$ from the root to $\r$, then $Q$ remains a path in $N$ 
    (because arc $(\r,\rl)$ is the only arc of $M$ that is removed by $s$, and it cannot be on $Q$ or else we would have a cycle in $M$).
    We claim that $\r$ is not confluent in $N$, meaning that $Q$ is the only path from \Root\ to $\r$ in $N$. 
    This follows from the fact that since $M$ is a regular state and $\ControlF^M(\Sys)=\rom{f6}$, 
    $M$ contains no confluent nodes (Corollary \ref{cor:tree-like-path}). As a result $Q$ is the sub-path of $P$ from \Root\ to $\r$.
    
    Any arc other than $(\r,\rl)$ is not removed by $s$ (since as we have said a step can remove at most one arc).
    Thus the interval of $P$ from \Root\ to $\r$ is a $k$-path, which means that $r_0$ is $k$-connected in $N$.
    Additionally, the interval of $P$ from $\rl$ to $a$ is also intact in $N$.
    
    Thus it remains to prove the following:
    \begin{claim}
    $(r_0,\new,\rl)$ is a $k$-path in $N$.
    \end{claim}
    It follows from this claim that $N\models \Root\arrow_k^* a$.
    
    Taking into account that $\new=\LeftF^N(r_0)$ and $\rl=\RightF^N(\new)$, we must prove that 
    $k<\KeyF(r_0)$ and that $\KeyF(\new)<k$ in order to conclude the proof of our claim. 

    $P$ is a $k$-path, $(\r,\rl)$ is an arc on $P$ and $\rl=\LeftF(r_0)$ in $S$. By the Definition \ref{DefPath} of a $k$-path we conclude that $k<\KeyF(\r)$.
    
    Next, since $\r=\RightF(\n)$ (in both $M$ and $N$), and since $\r$ is not confluent, $(\n, \r)$ is an arc on $P$, and so $k>\KeyF(\n)$. 
    Since $\KeyF(\new)=\KeyF(\n)$, we have that $k> \KeyF(\new)$, as required.

    \item[$s\in\Step(\Sys,\rom{f7})$:] In this case, arc $(x,y)= (\n,\ell_0)$ of $M$ 
    is replaced by arc $(x,z)=(\n,\r)$ of $N$. Taking into account that $\ell_0=\LeftF(\n)$ in $M$
    and that $(\n,\ell_0)$ is an arc of $P$ (which is a $k$-search path), we have that $k<\KeyF(\n)$.
    We have to prove that the path $(\n,\r,\new,\ell_0)$ is a $k$-path in $N$.
    
    Since $\r = \LeftF^N(\n)$ $k<\KeyF(\n)$, we have that $\n \arrow_k \r$.
    In $M$, $\n$ and $\r$ are path-connected (by Invariant \ref{inv:pre-rem}). 
    Since $M$ is regular, $\n$ is the sole path-connected node that is not tree-like. 
    As $\r\neq\n$ (by Invariant \ref{inv:address-points-to-root}), $\r$ is a tree-like node. 
    Since $\new=\LeftF(\r)$, $\KeyF(\n)=\KeyF(\new)< \KeyF(\r)$, and since $k<\KeyF(\n)$, we get that $k< \KeyF(\r)$.
    Hence $\r\arrow_k \new$.
    Clearly, $\new\arrow_k \ell_0$ since $\ell_0=\LeftF(\new)$ and $k<\KeyF(\n)$.

    \item[$s\in\Step(\Sys,\rom{f8})$:] In this case, arc $(x, y) =(\prt,\n)$ of $M$ is replaced by arc $(x, z) =(\prt,\r)$ of $N$.
    Nodes $\prt$ and $\n$ are path-connected in $M$.
    By the regularity of $M$, R2 implies that $\n\neq\bot$ is the sole path-connected node that is not tree-like, and hence there is a single path
    from the root to $\n$ is a $\KeyF(\n)$-path (by Lemma \ref{lemma:tree-like-path}). Thus $\n$ is not confluent in $M$.
    This implies that $\n$ is no longer path-connected in $N$ ($\n$ loses the only arc that connects with the root).
    So, in case $a=\n$, Equation (\ref{EQ9}) holds in $N$.
    If path $P$ continues past \n\ in $M$, then arc $(\n,\r)$ is on $P$ (since $\LeftF(\n)=\RightF(\n)$).
    Moreover, the final segment from $\r$ to $a$ of $P$ in $M$ remains a $k$-path in $N$. Hence arc $(\prt,\r)$ in $N$
    compensates for the missing arc $(\prt,\n)$ of $M$, and $a$ remains $k$-connected in $N$.

    \item[$s\in\Step(\Sys,\rom{v6})$:] In this case, $(x,y)=(\prt, \n)$ in $M$ is replaced by arc $(x,z)=(\prt,\child)$ in $N$.
    As $\ControlF^M(\Sys)=\rom{v6}$, Corollary \ref{lemma:pdc-rem-or-locked} implies that $\n$ is not pre-removed, 
    and since it is not removed, it is path-connected in $M$. 
    The regularity of $M$ implies by R2 that all path-connected nodes are tree-like, and hence
    there is no confluent node in $M$ (Corollary \ref{cor:tree-like-path}). 
    Thus at $N$, node $\n$ is no longer path connected. 
    In case $a=\n$, as $\n$ is no longer path-connected in $N$, (\ref{EQ9}) holds as required.
    It is not the case that $a=\bot$, and hence arc $(\prt,\bot)$ is not on $P$, and if $P$ continues past $\n$ in $M$, then $(\n,\child)$ is on $P$. 
    As above, arc $(\prt,\child)$ in $N$ compensates for the lost path $(\prt,\n,\child)$, and $a$ remains $k$-connected in $N$.
 
    \item[$s\in\Step(\Sys,\rom{v7})$:] In this case, arc $(x,y)=(\n, \bot)$ in $M$ is replaced by arc $(x,z)=(\n,\prt)$ in $N$. 
    However, as $a\neq \bot$, $(\n,\bot)$ is not an arc of $P$.
    
    \item[$s\in\Step(\Sys,\rom{v8})$:] In this case, arc $(x,y)=(\n,\child)$ of $M$ is replaced by arc $(x,z)=(\n,\prt)$ of $N$. 
    Since $\prt\arrow\child$ in $M$ (by the control-dependent invariants of Figure~\ref{cdLM}), $\child$ is both a left- and right-descendant of $\n$ in $M$, 
    and since $M$ is regular, R2 implies that $\n$ cannot be path-connected in $M$ (recall that $\child\neq\bot$ by the assumptions of the invariant). 
    Thus $(\n,\child)$ cannot be an oar of $P$ in $M$.
\end{description}
\end{proof}

\begin{SP}[uses: \ref{inv:pre-rem}, \ref{lemma:pdc-rem-or-locked}, \ref{Lem4.14}]
\label{Eq5a}
Let $s=(M,N)$ be a step such that $M$ is a regular state, and $x$
is an address of $M$ that is not path-connected in $M$. Then
 $x$ is not path-connected in $N$ as well.
\end{SP}

\begin{proof}
Let $x$ be an address of $M$ that is path-connected in $N$. Our aim
is to prove that $x$ is path-connected in $M$.
We may assume that $x\neq\bot$ since $M$ is a regular
state and $\bot$ is always path-connected in a regular state (Lemma \ref{Lem4.14}).

Let $P$ be the shortest path in $T$ from \Root\ to $x$. 
If all arcs of $P$ are in $M$ then surely $(\Root\arrow^* x)^M$, and
hence we may assume that $P$ contains a new arc of $N$ that is not in $M$. 

Suppose first that step $s$ introduces a new node \new. There are two possibilities for such a step.
\begin{description}
    \item[$s\in\Step(p,\rom{i3, m0})$] for some working process $p>0$, 
    and arc $(\nd_p,\bot)$ of $M$ is replaced by arc $(\nd_p,\new)$ of $N$. In this
    case the new arcs of $N$ are $(\nd_p,\new)$ and $(\new,\bot)$. Since \new\ is a node of all new arcs,
    \new\ is on $P$, or else all arcs of $P$ are in $M$. Since $x\neq\new$ (since $x$ is in $M$),
    \new\ is not the last node of $P$, and hence $(\new,\bot)$ is the last arc of $P$
    and hence $x=\bot$, and this contradicts our assumption about $x$. 
    
    \item[$s\in\Step(\Sys,\rom{f6})$] and the new arcs added in $N$ are $(r_0,\new)$, $(\new,\ell_0)$, and $(\new,\rl)$. 
    Since $P$ is not a path in $M$, \new\ must be on $P$, 
    and \new\ is not the last node of $P$ (because $x$ is an address of $M$). 
    Thus, either arc $(\new, \rl)$ is on $P$, or arc $(\new,\ell_0)$ is on $P$.
    
    If $(\new,\rl)$ is on $P$, then $\r$ (which is the sole node of $N$ that points to $\new$) is also on $P$. 
    Then arc $(\r,\rl)$ is in $M$, meaning that $\rl$, and thus, $x$ is path-connected in $M$.
    
    Otherwise, $(\new,\ell_0)$ is on $P$.
    The final segment of $P$ from $\ell_0$ to $x$ does not contain the node \new\ 
    (otherwise $\r$ would be on that segment, causing the cycle $\r$ to $\r$, contradicting regularity of $M$).
    Hence the segment of $P$ from $\ell_0$ to $x$ is in $M$.
    Additionally, $\n$ is not removed when $\ControlF(\Sys)=\rom{f6}$, and by \ref{lemma:pdc-rem-or-locked}, 
    we get that $\n$ is not pre-removed at $M$, so $\n$ is path-connected there. Since $\ell_0=\LeftF(\n)$ in $M$, $x$ is path-connected in $M$.
\end{description}

Assume next that step $s$ does not introduce a new address, i.e., $\AddressT^M=\AddressT^N$. 
So let $s$ be a step such that $N\models\Root\arrow^*x$ for some address $x$ of $M$. 
We must prove that $M\models\Root\arrow^*x$ as well.
  
It is obvious that we can ignore any steps that do not modify \LeftF/\RightF, which leaves the following cases: 
\begin{description}
    \item[$s\in\Step(\Sys, \rom{f7})$] in which arc $(\n,\ell_0)$ of $M$ is replaced with arc $(\n,\r)$ of $N$.
    However, arc $(\n,\r)$ is already in $M$, since $\r=\RightF(\n)$ in $M$ (by the control-dependent invariants of Figure \ref{cdLM}). 
    So every arc of $N$ is also an arc of $M$, and if $x$ is path-connected in $N$ then it is path-connected in $M$.

    \item[$s\in\Step(\Sys, \rom{f8})$] in which arc $(\prt,\n)$ of $M$ is replaced with arc $(\prt,\r)$ of $N$.
    $(\n,\r)$ is an arc in $M$ (by the control-dependent invariants of Figure \ref{cdLM}). 
    Thus, if $x$ is path-connected via $(\prt,\r)$ in $N$, then it is
    path-connected via $(\prt,\n)$ and $(\n,\r)$ in $M$.

    \item[$s\in\Step(\Sys, v6)$] in which arc $(\prt,\n)$ of $M$ is replaced with arc $(\prt,\child)$ of $N$.
    $(\nd,\child)$ is an arc in $M$ (by the control-dependent invariants of Figure \ref{cdLM}). 
    Thus, if $x$ is path-connected via $(\prt,\child)$ in $N$, then it is
    path-connected via $(\prt,\n),(\n,\child)$ in $M$.
    
    \item[$s\in\Step(\Sys, v7)$] in which arc $(\n,\bot)$ of $M$ is replaced with arc $(\n,\prt)$ of $N$.
    Since $M$ is regular, we have that $\prt$ is already path-connected in $M$, by R0. Thus we conclude that if $x$ is path-connected in $N$, then it must be
    path-connected in $M$.
    
    \item[$s\in\Step(\Sys, v8)$] in which arc $(\n,\child)$ of $M$ is replaced with arc $(\n,\prt)$ of $N$.
    However, arc $(\nd,\prt)$ is already arc in $M$ (by the control-dependent invariants of Figure \ref{cdLM}). 
    So any arc of $N$ is also an arc of $M$, and so if $x$ is path-connected in $N$ then it is path-connected in $M$.
\end{description}
\end{proof}

Proving that regularity is an inductive invariant is a major part of the correctness proof of the CF algorithm.

\begin{theorem}[uses: \ref{inv:pre-rem}, \ref{lemma:pdc-rem-or-locked}, \ref{lemma:tree-like-path}, \ref{cor:tree-like-path}, \ref{Lem4.14}, \ref{Cm2}]
\label{InvReg}
Regularity is an inductive invariant.
\end{theorem}

It may be helpful for to the reader to repeat the definition of regularity before commencing with this proof:

A state is said to be {\em regular} if the following conditions hold.
\begin{enumerate}
\item[R0.] $\prt$ is $\KeyF(\n)$-connected, and if $\n\neq\prt$ then $\lft\rightarrow\KeyF(\n)<\KeyF(\prt)$ and $\neg\lft\rightarrow\KeyF(\n)>\KeyF(\prt)$.
\item[R1.] For every process $p$ with $p>0$, node $\nd_p$ is potentially $k_p$-connected, and if $\nxt_p\neq\bot$ then $\nxt_p$ is also potentially $k_p$-connected. 
\item[R2.] If $x\not\in\{\Root,\bot\}$ is a path-connected node that is not tree-like, then $x=\n$ and \linebreak{}$\ControlF(\Sys)\in \{ \rom{f7, f8,  r7, r8}\}$.
\end{enumerate}

\begin{proof}
We shall consider all possible steps $s=(M,N)$, assume that $M$ is regular and deduce that $N$ is also regular. 
In addition to proving regularity of $N$, we check for each step whether 
$\Set(M)=\Set(N)$ (see Definition \ref{DefSet}), and in the case of inequality, we determine the relation between
the two sets (either the insertion or the deletion of some key value).

Before we commence, we note that the second part of R0, namely that if $\n\neq\prt$ then $\lft\rightarrow\KeyF(\n)<\KeyF(\prt)$ and $\neg\lft\rightarrow\KeyF(\n)>\KeyF(\prt)$, is trivially preserved by any step step, except for steps $s\in\Step(\Sys,\rom{m0})$. This is because in our model, we assumed that key values of nodes never change. We will not bother to reiterate this point for every step. We will only address this part of R0 directly in the case of $s\in\Step(\Sys,\rom{m0})$, since these steps reset the values of $\n$, $\prt$ and $\lft$, and so the claim requires a proof.

We begin by checking the steps of an arbitrary working process $p>0$.

\paragraph*{Assume $s\in\Step(p, \rom{m0})$.} In this step, the local variables of $p$ are re-initialized such that in $N$: $\nd_p=\nxt_p=\Root$, $k_p$ is set to a value $k\in\KeyT$, and $\ControlF(p)\in\{\rom{c1,d1,i1}\}$. The extension of the predicates $\RemovedP$ and $\DeletedP$, the functions $\LeftF$ and $\RightF$, and $\AddressT$ are the same in $N$ as in $M$.

Since this step does not change the values of the local variables of $\Sys$ and does not change \LeftF/\RightF, and since R0 holds in $M$, R0 must hold in $N$.

Since $\Root$ is trivially $k$-connected for every $k\in\KeyT$, R1$(p)$ holds in $N$. Since the step does not change the local variables of any other process, and since $\RemovedP$, \LeftF/\RightF, and $\AddressT$ are the same in $M$ as in $N$, R1$(q)$ holds in $N$ for every other working process $q\neq p$.

Since R2 holds in $M$ and since \LeftF/\RightF\ and $\RemovedP$ are unchanged by the step, R2 holds in $N$ as well.

Since \LeftF/\RightF, $\AddressT$ and $\DeletedP$ are unchanged by the step, $\Set(M)=\Set(N)$.

\paragraph*{Assume $s\in\Step(p, \rom{c1})$.}
There are three possibilities for the execution of this step: 

\begin{enumerate}
    \item $s\in\Step(p, \rom{c1, m0})$, in which case $\nxt^M_p=\bot$. 
    Then $\ControlF^N(p)=m_0$ and so R1$(p)$ holds trivially in $N$
    since $\nd_p^N = \nd_p^M$, $\nxt_p^N=\nxt_p^M$, and the step did not change 
    the \LeftF/\RightF\ functions or the extension of the \DeletedP\ predicate.
    For any other working process $q\neq p$,
    R1$(q)^N$ is obvious since a step by $p$ does not change the program variables
    $\nd_q$ and $\nxt_q$.
    
    \item $s\in\Step(p, \rom{c1, c2})$, in which case 
    \[\nxt^M_p \neq \bot\wedge  \nd_p^N=\nxt_p^M = \nxt_p^N \wedge k_p =\KeyF(\nd_p^N).\]
    Since $M$ is a regular state and $\nxt_p^M\neq \bot$, R1$(p)^M$ implies that $\nxt_p^M$ is potentially $k_p$-connected in $M$. 
    Since $\nd_p^N=\nxt^M_p$, address $\nd_p^N$ is potentially $k_p$-connected in $M$  
    and hence in $N$ (and evidently so is $\nxt_p^N=\nxt_p^M$). 
    The arguments for R1$(q)^N$ where $q\neq p$.
    
    \item $s\in\Step(p, \rom{c1, c1})$, in which case $\nxt^M_p \neq \bot\wedge \nd_p^N=\nxt_p^M$, as in the previous case,
    but now $k_p \neq\KeyF(\nd_p^N)$ and $\nxt_p^N$ is defined by 
    \[M\models \nxt_p^N = \LR(\nd_p^N, k_p<\KeyF(\nd_p^N).\]
    As in the previous case, R1$(p)^M$ implies that $\nxt_p^M$ is potentially $k_p$-connected and hence that $\nd_p^N$ is potentially $k_p$-connected in $M$ and consequently in $N$. Observation \ref{lemma:pt3}(2) says that if node $x$ is potentially
    $k$-connected and $x\arrow_k a$ then $a$ is potentially connected.
    As $\nd_p^N\arrow_{k_p} \nxt_p^N$, address $\nxt_p^N$ is potentially $k_p$-connected in $N$.
    Thus R1$(p)$ holds in $N$, and R1$(q)$ for any working process $q>0$ follows as before. 
\end{enumerate}
Thus, R1 holds in $N$.
    
Since $M$ and $N$ have the same addresses, the same
$\RightF$ and $\LeftF$ functions, and the same interpretations of
the program variables of process \Sys, R2 holds in $N$ since it
holds in $M$.

For the same reasons, and since $s$ does not modify the local variables of $\Sys$, R0 holds in $N$ since it holds in $M$.

The definition of $\Set(M)$ depends only on the \KeyF\ function,
the \DeletedP\ predicate, and the $k$-connection predicate $\Root\arrow^*_k x$. 
Since an execution of c1 does not change these components between $M$ and $N$, $\Set(N)=\Set(M)$ follows and $\Set$ is not changed in any execution of c1.

Instructions d1 and i1 are textually isomorphic to c1, and the
same proof given above for c1 shows that their executions preserve regularity and the
value of \Set.

\paragraph*{Assume $s\in\Step(p, \rom{d2})$.}
By the control-dependent invariants of Figure \ref{fig:protocol-elementary-invariants}:  
\begin{equation}
\label{Eq:locked}
M\models \LockedP(\nd_p,p)\wedge \KeyF(\nd_p)=k_p.
\end{equation}

There are three possibilities for the execution of this step: 
\begin{enumerate}
    \item $M\models \DeletedP(\nd_p)$, and so $\ControlF^N(p)=\rom{m0}$.
    In all other aspects, $N$ is identical to $M$. 
    In particular $\nd_p^N=\nd_p^M$ and $\nxt_p^N=\nxt_p^M$. 
    So $N$ is also regular, and $\Set(N)=\Set(M)$.
    
    \item $M\models \neg\DeletedP(\nd_p)\wedge \RemovedP(\nd_p)$, which implies that $s\in\Step(p, \rom{d2, d1})$,
    and so $\nxt_p^N = (\RightF(\nd_p))^M$, and there is no change in $\nd_p$.
    Since $M$ is regular, $\nd_p^M=\nd_p^N$ is potentially $k_p$-connected in $M$ and, so in $N$. 
    Since $\nd_p$ is removed, PT3$(\nd_p,k_p)$ holds in $M$. 
    Since $\nd_p^M\arrow \nxt_p^N$, Observation \ref{lemma:pt3} implies that $\nxt_p^N$ is potentially $k_p$-connected in $N$.
    So R1$(p)$ holds in $N$, and as before R1$(q)$ holds for every $q>0$. The arguments for R0 and R2 are as in the previous 
    case. $\Set(N)=\Set(M)$ is obtained as in the previous cases since the functions
    \LeftF/\RightF\ and the \DeletedP\ predicate stay the same in $N$ as in $M$.

    \item $M\models \neg\DeletedP(\nd_p)\wedge \neg\RemovedP(\nd_p)$, which implies that $s\in\Step(p,\rom{d2, m0})$, 
    and $\DeletedP(\nd_p)$ in $N$. 
    There are no changes in $\nd_p$ or in $\nxt_p$, and
    since $M$ is regular, it follows that $N$ is regular as in previous cases.
    
    We claim that $k_p\in\Set(M) \wedge \Set(N)=\Set(M)\setminus\{k_p\}$.
    This will follow immediately after proving that $\nd_p$ is $k_p$-connected in $M$
    (which entails that $k_p\in\Set(M)$ because $\KeyF(\nd_p)=k_p$ in $M$):
    
    Since $M$ is regular, $\nd_p^M$ is potentially $k_p$-connected. Of the three possibilities
    PT1, PT2, and PT3  we will rule-out the last two, and deduce
    that $\nd_p^S$ is $k_p$-connected in $M$. 
    \begin{description}
        \item[$\rom{PT2}$:] If $\nd_p^M$ were pre-removed in $M$, then
        Inductive invariant \ref{inv:pre-rem} would imply that $\nd_p^M=\n \wedge \ControlF(\Sys)\in \{\rom{f9, r9, v7, v8, v9}\}$,
        and then $\LockedP(\nd_p^M,\Sys)$ can be concluded in contradiction to (\ref{Eq:locked}).
        \item[$\rom{PT3}$:] We assumed that $\nd_p^M$ is not removed, contradicting the condition of this case.
    \end{description} 
    So $\nd_p^M$ is $k_p$-connected in $M$, as required.
\end{enumerate}

\paragraph*{Assume $s\in\Step(p, \rom{i2})$.}
The code of i2 is very similar to that of d2, only replacing $\nd.\del$ with $\neg\nd.\del$. 
Thus the proof of regularity of $N$ is obtained by the same arguments as those that served for d2.
As with executions of d2, there are three possibilities for executions of i2.
\begin{enumerate}
    \item $M\models \neg \DeletedP(\nd_p)$, then $N$ is regular and $\Set(N)=\Set(M)$.
    
    \item $M\models \DeletedP(\nd_p)\wedge \RemovedP(\nd_p)$, then $N$ is regular and $\Set(N)=\Set(M)$.
    
    \item $M\models \DeletedP(\nd_p)\wedge \neg\RemovedP(\nd_p)$, then the same proof as for d2 shows that $N$ is a regular.
    For handling the change to \Set, the same proof gives that $\nd_p^N=\nd_p^M$ is $k_p$-connected in $M$
    and $N$, but new the conclusion for this step is that $k_p\in\Set(M)\wedge \Set(N)=\Set(M)\cup\{k_p\}$.
\end{enumerate}

\paragraph*{Assume $s\in\Step(p, \rom{i3})$.}
The control-dependent invariants of Figure \ref{fig:protocol-elementary-invariants} imply that
\begin{equation}
\label{Eq:ndp-not-bot}
M \models  \LockedP(\nd_p,p) \wedge \nd_p\neq \bot \wedge \KeyF(\nd_p)\neq k_p \wedge \nxt_p = \bot.
\end{equation}
There are two possibilities for this step which depend on whether or not $M$ satisfies \linebreak{}
$\LR(\nd_p,k_p < \KeyF(\nd_p))\neq \bot$:

\begin{enumerate} 
    \item If $M\models \LR(\nd_p,k_p < \KeyF(\nd_p))\neq \bot$,
    then $s\in\Step(p, \rom{i3, i1})$ and $\nxt_p^N= \LR(\nd_p,k_p < \KeyF(\nd_p))\neq \bot$.
    Since $M$ is regular, $\nd_p^M=\nd_p^N$ is potentially $k_p$-connected in $M$, and 
    Observation \ref{lemma:pt3}(2) implies that $\nxt_p^N$ is potentially $k_p$-connected. 
    So R1 holds in $N$, and as in previous cases, R0 and R2 hold in $N$ as well, and $\Set(N)=\Set(M)$.

    \item $M\models \LR(\nd_p, k_p < \KeyF(\nd_p))= \bot$, and assume without loss of generality
    that $k_p<\KeyF(\nd_p^M)$, and thus, $\LeftF(\nd_p)^M=\bot$. 
    Then $s\in\Step (p, \rom{i3, m0})$, there is a new address $\new$ in $N$ such that $\new=\LeftF^N(\nd_p^M)$, 
    and both left and right children of $\new$ are $\bot$.
    Additionally, $\nd^M_p=\nd_p^N$ and $\nxt^M_p=\nxt_p^N$. 
    
    Since $M$ is regular, and with the help of Step-property \ref{Cm2}, for any node $x\neq\bot$ and any $k\in\KeyT$, 
    if $x$ is $k$-connected in $M$ and path-connected in $N$, then it is $k$-connected in $N$.
    The only arc of $M$ that was removed by the step is $(\nd_p,\bot)$, and the only descendant of $\bot$ is $\bot$. 
    Thus every node $x$ that is path-connected in $M$ is also path-connected in $N$. 
    In addition, we can conclude that every path in $M$ exists in $N$, except for paths ending with the arc $(\nd_p,\bot)$. 
    Thus for any two nodes $y,z$ such that $z\neq\bot$, if $P$ is a path from $y$ to $z$ in $M$ then $P$ is a path from $y$ to $z$ in $N$.
    From this we conclude that R0 and R1 must hold in $N$, since they hold in $M$.

    We prove that R2 holds in $N$ by showing that any path-connected node that is not tree-like in $N$ is also not tree-like in $M$.
    Assume for a contradiction that there is some address $a$ of $N$ that is path-connected but not tree-like in $N$. 
    $a\neq\new$ trivially, since both the left and right children of $\new$ are $\bot$.
    So $a$ must be an ancestor of $\new$ in $N$, which means that $nd_p$ must be path-connected in $N$. 
    This means that $\nd_p$ is also path-connected in $M$, since the arcs added by the step $s$ are $(\nd_p,\new)$ and $(\new,\bot)$, 
    neither of which can be on a new path to $\nd_p$ in $N$.
    As a result, and since $\nd_p$ is potentially $k_p$-connected in $M$, $nd_p$ and all of its path-connected ancestors are in fact $k_p$-connected in $M$.
    Thus, the insertion of $\new$ cannot cause any of these nodes to become non-tree-like, since $\KeyF(\new)=k_p$.
    We conclude that $a$ is not tree-like in $M$, and R2 holds in $N$ because it holds in $M$.
\end{enumerate} 


In the next part of our proof, 
we handle steps by the system process. 
Recall that since \rotateLeft\ and \rotateRight\ are symmetrically similar, 
we only prove the claim for \rotateLeft. Without loss of generality, we also assume that $\lft=\TRUE$.

We will also prove that for any step $(S,T)$ by process \Sys, $\Set(T)=\Set(S)$.

\paragraph*{Assume $s\in\Step(\Sys, \rom{m0})$.} In this step, the local variables of $\Sys$ are re-initialized such that in $N$: $\lft=\TRUE$ (by our assumption); $\prt\in\AddressT^N$ that is not removed and is not $\bot$; $\n=\LeftF(\prt)$ (since $\lft=\TRUE$) and $\n\neq\bot$; and $\ControlF(\Sys)\in\{\rom{f6,r6,v6}\}$. The extension of the predicates $\RemovedP$ and $\DeletedP$, the functions $\LeftF$ and $\RightF$, and $\AddressT$ are the same in $N$ as in $M$.

Since $M$ is regular and $\ControlF(\Sys)\notin\{\rom{f7,f8,r7,r8}\}$, every node is tree-like. By Corollary \ref{cor:tree-like-path}, this implies that the path from $\Root$ to $\n$ a unique $\KeyF(\n)$-path. Since $\prt$ is an ancestor of $\n$, if $\prt$ is path-connected then it must be $\KeyF(\n)$-connected in $N$. We assumed that $\prt$ is not removed, and by Invariant \ref{inv:pre-rem}, since $\prt\neq\n$, $\prt$ is not pre-removed. Thus, $\prt$ is path-connected, as required.
Since $M$ is regular, $\prt$ is tree-like, and thus, given that $\n=\LeftF(\prt)$, we have that $\KeyF(\n)<\KeyF(\prt)$ in $N$. From these facts, we conclude that R0 holds in $N$.

Since this step does not change the values of the local variables of $p$ and does not change \LeftF/\RightF, and since R1 holds in $M$, R1 must hold in $N$.

Since R2 holds in $M$ and since \LeftF/\RightF\ and $\RemovedP$ are unchanged by the step, R2 holds in $N$ as well.

Since \LeftF/\RightF, $\AddressT$ and $\DeletedP$ are unchanged by the step, $\Set(M)=\Set(N)$.

\paragraph*{Assume $s\in\Step(\Sys, \rom{f6})$.} In this step, the new node $\new$ is added, and arc $(\r, \rl)$ of $M$ is replaced by the arcs $(\r, \new)$, $(\new,\rl)$ and $(\new,\ell_0)$ of $N$. We note that $(\r, \rl)$ is the only arc removed by the step, and that there is still a path from $\r$ to $\rl$ in $N$, via the arcs $(\r, \new)$ and $(\new,\rl)$. As a result, if any two nodes $x,y$ are connected by a path in $M$, then they must be connected by a path in $N$ (since if $x,y$ were to be disconnected by the step, they would have to connect via $(\r,\rl)$ in $M$). From this we can easily draw two conclusions:
\begin{enumerate}
    \item Since any two nodes $x,y$ are connected by a path in $M$, in the case $x=\Root$ we get that any $y$ that is path-connected in $M$ is also path-connected in $N$. As a result, by Step-property \ref{Cm2}, we get that for any $k\in\KeyT$ and for every node $y\neq\bot$ in $M$, if $y$ is $k$-connected in $M$, then it is $k$-connected in $N$.
    \item Otherwise, $x$ is potentially $k$-connected (but not $k$-connected) for some $k\in\KeyT$ in $M$. Then $y$ is as in the definition of PT2 or PT3 (Definition \ref{DefPot}), i.e., the $k$-connected ``anchor'' of $x$. Thus $x$ is potentially $k$-connected in $N$ as well, by way of the same $y$, since every such pair $x,y$ are still connected in $N$, and $y$ is still $k$-connected in $N$, as we concluded above.
\end{enumerate}

These conclusions hold for any node, and in particular for any $\nd_p$ and $\nxt_p\neq\bot$ for any process $p$, and for $\prt$. Thus, since $M$ is a regular state, R0 and R1 hold in $M$, meaning that $\prt$ is $\KeyF(\n)$-connected in $M$, and that $\nd_p$ and $\nxt_p\neq\bot$ are potentially $k_p$-connected in $M$ for any process $p$. We conclude that R0 and R1 hold in $N$ as well.

In order to prove that R2 holds in $N$, we must show that no node other than $\n$ is confluent in $N$, i.e., that every node except for $\n$ is tree-like in $N$. Since $M$ is regular, there are no cycles in $M$ by lemma \ref{Lem4.14}. In addition, since $M$ is regular, R2 holds in $M$, and since $\ControlF^M(\Sys)=f6$, we have that there are no confluent nodes in $M$. Also, we have from the control-dependent invariants of Figure \ref{cdLM} that $\n\arrow\r$ in $M$. From these facts we can conclude that there is no path from $\ell_0$ to $\r$ in $M$, otherwise either $\r$ would be confluent in $M$ (if the path from $\ell_0$ to $\r$ does not go through $\n$), or there would be a cycle going through $\n$ (since $\n\arrow\r$ in $S$).

Additionally, since we know there are no cycles or confluent nodes in $M$, all of the path-connected nodes are tree-like in $M$, from which we can conclude that $\KeyF(\ell_0)<\KeyF(\n)<\KeyF(\rl)<\KeyF(\r)$. Note also that $\KeyF(\new)=\KeyF(\n)$ in $N$, and so $\KeyF(\ell_0)<\KeyF(\n)=\KeyF(\new)<\KeyF(\rl)<\KeyF(\r)$ in $N$. As a result, the addition of the arc $(\r, \new)$ such that $\LeftF(\r)=\new$ and the arc $(\new,\ell_0)$ such that $\LeftF(\new)=\ell_0$ cannot make any node from $M$ become non-tree-like in $N$, except for $\n$. Finally, since $\n$ is tree-like in $M$ and $\KeyF(\new)=\KeyF(\n)$ in $N$ and $\LeftF(\new)=\ell_0$ and $\RightF(\new)=\rl$, $\new$ is also tree-like in $N$. We conclude that $R2$ holds in $N$.

We prove that $\Set(N)=\Set(M)$: States $M$ and $N$ have the same extensions of
predicates \RemovedP\ and \DeletedP\ on the nodes of $M$, and Node \new\ has the same
key as node $\n$. We also showed that this step does not change the $k$-connectivity of any node from $M$ in $N$ for any $k$. 
Thus, it must be the case that $\Set(N)=\Set(M)$.

\paragraph*{Assume $s\in\Step(\Sys, \rom{f7})$.} In this step, arc $(\n , \ell_0)$ of $M$ is replaced by arc $(\n,\r)$ of $N$.
As in the previous case ($s\in\Step(\Sys,rom{f6})$), any two nodes $x,y$ that are connected in $M$ are connected in $N$, because the removed
arc $(\n, \ell_0)$ can be replaced by the path $\n\arrow\r\arrow\new\arrow\ell_0$ of $N$. In the same manner as before, with the help of Step-property \ref{Cm2}, we can conclude that for every node $x$ in $M$ and any $k\in\KeyT$, if $x\neq\bot$ is potentially $k$-connected in $M$ then it is potentially $k$-connected in $N$. This gives us that $\prt$ is $\KeyF(\n)$-connected in $N$ and that $\nd_p$ and $\nxt_p\neq\bot$ are potentially $k_p$-connected for any $p$ in $N$, from which we get that R0 and R1 hold in $N$.

To prove that R2 holds in $N$, let $x$ be an address of $N$ that is path-connected
but not tree-like in $N$. Since obviously $\ControlF(\Sys)^T \in \{\rom{f7, f8, r7, r8}\}$, we must show that and that $x=\n$.
Assume that $x\neq \n$. Since $x\neq\n$, $y=\LeftF^M(x)=\LeftF^N(x)$
and $z=\RightF^M(x)=\RightF^N(x)$ (because only arc $(\n,\ell_0)$ is removed by the step). And just as any other node,
$y$ (and $z$) have the same descendants in $N$ as in $M$, meaning that $x$ is not tree-like in $M$, contradicting the regularity of $M$. 
We conclude that such an $x$ does not exist, and R2 holds in $N$.

Finally, as was the case before, $\Set(N) = \Set(M)$, since $M$ and $N$ have the same addresses and same extension of predicate
$\DeletedP$, and the step does not change the $k$-connectivity of any node from $M$ in $N$ for any $k$.

\paragraph*{Assume $s\in\Step(\Sys,\rom{f8})$.} In this step, arc $(\prt, \n)$ of $M$ is replaced by arc $(\prt, \r)$ or $N$.
This case is slightly different from the previous two cases; this time we can show that the connectivity of every pair of nodes $x,y$ of $M$ is preserved by the step {\em if $y\neq\n$}: Since $\prt\neq\bot$ is not removed by the control-dependent invariants of Figure \ref{cdLM}, and by Corollary \ref{lemma:pdc-rem-or-locked}, we have that $\prt$ is path-connected in $M$. By the control-dependent invariants of Figure \ref{cdLM}, there are paths from $\prt$ to $\ell_0$ and from $\prt$ to $\r$ in $M$ (via $(\prt,\n)$,$(\n,\ell_0)$ and via $(\prt,\n)$,$(\n,\r)$, respectively). So, using the same arguments as before, the potential $k$-connectivity of any node $y$ for any $k\in\KeyT$ is preserved by the step except in a small number of cases:
\begin{enumerate}
    \item $y=\n$ in $M$: If $\n$ is still path-connected in $N$, then nothing changes, and the claim holds. Otherwise, $\n$ is not path-connected in $N$, then it is pre-removed in $N$ by the fact that it is not removed (control-dependent invariants of Figure \ref{cdLM}) and by Corollary \ref{lemma:pdc-rem-or-locked}. So we must show that if $\n$ is $k$-connected for some $k\in\KeyT$ in $M$, then PT2$(\n,k)$ holds in $N$.
    
    Since $M$ is regular and $\prt$ is path-connected, then $\prt$ is tree-like in $M$. We assumed that $\n=\LeftF(\prt)$ (we are proving the case that $\lft=\TRUE$), and so $\KeyF(\n)<\KeyF(\prt)$ and also $\KeyF(\r)<\KeyF(\prt)$ ($\r$ is a left-descendant of $\prt$). So for any $k$ for which $\n$ is $k$-connected in $M$, $k<\KeyF(\prt)$. Next, we know that $T\models\r=\LeftF(\prt)$, and so $\r$ is trivially $k$-connected in $N$. Thus PT2$(\n,k)$ must hold in $N$, since $\n\arrow\r$ in $N$.
    
    \item $y\neq\n$ and PT3$(y,k)$ holds in $M$, and $\n$ is the $k$-connected ``anchor'' of $y$ from those definitions (see Definition \ref{DefPot}): As before, if $\n$ is still path-connected in $N$, then nothing changes, and the claim holds. Otherwise, by the previous bullet, PT2$(\n,k)$ and so PT3$(y,k)$ holds in $N$ by Definition \ref{DefPot}).
    
    Note that it is impossible that PT2$(y,k)$ holds if $y\neq\n$, since by Inductive invariant \ref{inv:pre-rem}, only $\n$ can be pre-removed.
\end{enumerate}
Thus R1 holds in $N$.

Since $M$ is regular, $\prt$ is not a descendant of $\n$ in $M$. This follows from the fact that $(\prt, \n)$ is an arc of $M$, and there are no cycles in a regular state (by lemma \ref{Lem4.14}). From this and the proof that R1 holds in $N$, it follows that $\prt$ remains $\KeyF(\n)$-connected in $N$, since it is $\KeyF(\n)$-connected in $M$. Thus R0 holds in $N$.

We now prove R2 holds in $N$. Since $\ControlF(\Sys) = \rom{f9} \notin\{\rom{f7,f8,r7,r8}\}$ in $N$, we must show that all path-connected nodes are tree-like in
$N$. Since $M$ is regular and $\ControlF^M(\Sys) = \rom{f8}$, all path-connected nodes are tree-like, except for $\n$. We show that $\n$ is not path-connected in $N$, and that $\prt$ remains tree-like in $N$, from which we conclude that all path-connected nodes are tree-like in $N$ (since only $\LeftF(\prt)$ changes in the step, showing $\prt$ remains tree-like suffices). We argued above that $\prt$ is path-connected $M$, and so $\n$ is path-connected in $M$. 

We claim that $\n$ is not confluent in $M$: If there were a path-connected node $x\neq\prt$ such that $x\arrow\n$ in $s$, then $\prt$ and $x$ would have some common ancestor $z$ such that $\n\in\LeftDes(z)$ and also $\n\in\RightDes(z)$. Since there are no cycles in $M$ (by lemma \ref{Lem4.14}), $z\neq\n$, meaning $z$ is not tree-like in $M$, contradicting the regularity of $M$. So $\n$ is path-connected in $M$ only via the arc $(\prt, \n)$, which is removed by the step, and so $\n$ is not path-connected in $N$.

That $\prt$ remains tree-like in $N$ is trivial, since we already know that $\LeftF(\prt)=\r$ in $N$, but $\r\in\LeftDes(\prt)$ in $M$. Thus R2 holds in $N$.

Finally, we show that $\Set(N) = \Set(M)$: $M$ and $N$ have the same addresses and same extension of predicate
$\DeletedP$. We also showed that this step does not change the $k$-connectivity of any node from $M$ in $N$ for any $k$, except for node $\n$. 
So it remains to show that $\new$ is $\KeyF(\new)=\KeyF(\n)$-connected in $N$. 
Lemma \ref{lemma:tree-like-path} states that if every path-connected node is tree-like then every $x$ is $\KeyF(x)$-connected.
This observation holds in our case, since we showed that R2 holds in $N$, which means that every path-connected node is tree-like in $N$.
Thus, it must be the case that $\Set(N)=\Set(M)$.

\paragraph*{Assume $s\in\Step(\Sys, \rom{f9})$.} In this step, $\LeftF$, $\RightF$ and $\AddressT$ have the same interpretations in both $M$ and $N$, $\nd_p^M=\nd_p^N$ and $\nxt_p^M=\nxt_p^N$ for every working process $p>0$, and likewise the denotations of $\n$ and $\prt$ do not change. It is obvious for every address $x\neq \n$ and key value $k$ that $x$ is potentially $k$-connected in $M$ if and only if $x$ is potentially $k$-connected in $N$. 
Thus R0 holds in $N$. 

Since $\n$ is marked as removed by the step, To prove that R1 holds in $N$, 
we must show that for any $k\in\KeyT$, $\n$ is potentially $k$-connected in $N$ if and only if it is potentially $k$-connected in $M$.
Since $M$ is regular and $\ControlF(\Sys)=\rom{f9}\not\in \{\rom{f7, f8, r7, r8}\}$, every path connected node of $M$ is tree-like.
From the control-dependent invariants of Figure \ref{cdLM}, we have that $\n$ is focused, and so not tree-like, meaning it cannot be path-connected in $M$.
Since $\n$ is not removed in $M$, it is pre-removed in $M$. 
If $\n$ is potentially $k$-connected in $M$, then PT2$(\n, k)$ holds in $M$, and since $\ControlF(\Sys)=\rom{f9}$, $\r$ is
$k$-connected in $M$. Thus $\r$ is $k$-connected in $N$, and so PT3$(\n,k)$ holds in $N$.
We conclude that R1 holds in $N$.

R2 trivially holds in $N$ if and only if it holds in $M$, since $\LeftF$, $\RightF$ and $\AddressT$ have the same interpretations in $M$ and in $N$.

Finally, $\Set(M)=\Set(N)$ holds trivially, since $\LeftF$, $\RightF$ and $\AddressT$ and $\DeletedP$ have the same interpretations in $M$ and in $N$.

\paragraph*{Assume $s\in\Step(\Sys,\rom{v6})$.} In this step, arc $(\prt,\n)$ of $M$ is replaced with arc $(\prt,\child)$ of $N$. The proof for this step is nearly identical to that of $s\in\Step(\Sys,\rom{f8})$ above, with the only real difference being in the proof that $\Set(M)=\Set(N)$.

R0 and R1 hold in $N$ by the same reasoning as for $s\in\Step(\Sys,\rom{f8})$: the connectivity of every pair of nodes $x,y$ of $M$ is preserved by the step {\em if $y\neq\n$}, using the same arguments. The same two special cases for where $\n$ is involved are handled in the same way, using $\child$ instead of $\r$.

The proof that R2 holds in $N$ is also very similar, once again substituting $\r$ with $\child$.

Finally, we show that $\Set(N) = \Set(M)$: $M$ and $N$ have the same addresses and same extension of predicate
$\DeletedP$. We also showed that this step does not change the $k$-connectivity of any node from $M$ in $N$ for any $k$, except for node $\n$, 
which becomes non-path-connected in $N$. As a result, we must show that $\KeyF(\n)\notin\Set(M)$. 
By the control-dependent invariants of Figure \ref{cdLM}, we know that $\DeletedP(\n)$ in $M$. 
In addition, since $M$ is regular and $\ControlF(\Sys)=\rom{v6}$ in $M$, R2 implies that all path-connected nodes are tree-like in $M$. 
Thus, by Corollary \ref{cor:tree-like-path}, there are no two path-connected nodes with the same key in $M$, and so $\KeyF(\n)\notin\Set(M)$, since $\DeletedP(\n)$.
Thus $\Set(N) = \Set(M)$, as required.

\paragraph*{Assume $s\in\Step(\Sys, \rom{v7})$.} In this step, arc $(\n,\bot)$ of $M$ is replaced by arc $(\n, \prt)$ of $N$.

Once again, we claim that for every pair of nodes $x,y$ such that there is a path form $x$ to $y$ in $M$, there is a path from $x$ to $y$ in $N$. This is trivial (recall that the only descendant of $\bot$ is $\bot$ itself), except for the case that $y=\bot$. As before, since $M$ is regular, and with the help of stop-property \ref{Cm2}, for every node $z\neq\bot$ in $M$ and every $k\in\KeyT$, if $z$ is $k$-connected in $M$ then it is $k$-connected in $N$. Thus R0 holds in $N$.

Since the only arc that is changed by the step is an outgoing arc of $\n$, to prove R1 holds in $N$, it suffices to show that for any $k\in\KeyT$, if $\n$ is potentially $k$-connected in $M$ then it is potentially $k$-connected in $N$. So we must show that for every $k\in\KeyT$ such that $\n\arrow_k\bot$ holds in $M$, $\prt$ is $k$-connected in $N$. If we can show that $\n$ is pre-removed in $M$, then we get this conclusion ``for free'' from the definition of PT2$(\n,k)$ (see Definition \ref{DefPot}).

To show that $\n$ is pre-removed in $M$, we check two possibilities:
\begin{enumerate}
    \item If $\child\neq\bot$, then since $\prt$ is $\KeyF(\n)$-connected in $M$, then it must be path-connected in $M$, and from the control-dependent invariants of Figure \ref{cdLM} we have that $\prt\arrow\child$. So if $\n$ were path connected in $M$, then $\child$ would be confluent in $M$. However, since R2 holds in $M$ and since $\ControlF(\Sys)=\rom{v7}$, this cannot be.
    \item Assume $\child=\bot$. Since R2 holds in $M$ and since $\ControlF(\Sys)=\rom{v7}$, then all path-connected nodes in $M$ are tree-like. Then if $\n$ is path-connected, we have from Corollary \ref{cor:tree-like-path} that there is a single $\KeyF(\n)$-path from $\Root$ to $\n$ in $M$. Since $\prt$ is $\KeyF(\n)$-connected in $M$ by R0, then $\n$ must be a descendant of $\prt$ in $M$. Since $\lft=\TRUE$, then by R0, $\KeyF(\n)<\KeyF(\prt)$, and since $\prt$ is tree-like in $M$, we have that $\n\in\LeftDes(\prt)$ in $M$. However, since $\lft=\TRUE$, $\LeftF(\prt)=\bot$ in $M$, and $\LeftDes(\prt)=\emptyset$ (see Definition \ref{def:descendants}).
\end{enumerate}
In both cases, $\n$ cannot be path-connected in $M$, as required, and so R1 holds in $N$.

Since in both $M$ and $N$, $\ControlF(\Sys)\notin\{\rom{f7,f8,r7,r8}\}$, to prove that R2 holds in $N$, we have to prove that all path-connected nodes are tree-like in $N$. 
It is easy to check that $M$ and $N$ have the same path-connected nodes,
and that for any path-connected node $x$, the left (correspondingly right) descendants of $x$ in $M$ and in $N$ form the same set. 
So R2 holds in $N$ as well.

That $\Set(N)=\Set(M)$ follows immediately from the fact that $\DeletedP$ has the same extension in $M$ as in $N$, 
and from our observations that for any path-connected node $x$, the left (correspondingly right) descendants of $x$ in $M$ and in $N$ form the same set.

\paragraph*{Assume $s\in\Step(\Sys, \rom{v8})$.} In this step, arc $(\n, \child)$ of $M$ is replaced with arc $(\n, \prt)$ of $N$.

By control-dependent invariants of Figure \ref{cdLM}, we have that in $M$: $\prt\arrow\child$, $\n\arrow\prt$, $\prt$ is not removed, and $\prt\neq\nd$. With the help of \ref{inv:pre-rem}, these facts imply that $\prt$ cannot be pre-removed in $M$. Since $M$ is regular, $\prt$ is $\KeyF(\n)$-connected in $M$. It follows that $\n$ cannot be an ancestor of $\prt$ in $M$ (otherwise, the $\KeyF(\n)$-path would end at $\n$ and not continue to $\prt$). Thus, $\n$ must be pre-removed in $M$, otherwise $\prt$ would confluent in $M$, since $\n$ points to $\prt$ but is not on the path $\Root\arrow_{\KeyF(\n)}^*\prt$. 

Since the only arc removed by the step of an outgoing arc of $\n$, we conclude that all nodes that are path connected in $M$ are path connected in $N$. Since $M$ is regular, and by Step-property \ref{Cm2}, for every node $x$ in $M$ and every $k\in\KeyT$, if $x$ is $k$-connected in $M$, then it is $k$-connected in $N$. Thus R0 holds in $N$.

As in previous steps, since the only arc that changed by the step is an outgoing arc of $\n$, in order to prove that R1 holds in $N$, it suffices to show that for every $k\in\KeyT$ such that $\n$ is potentially $k$-connected in $M$, it is potentially $k$-connected in $N$. By control-dependent invariants of Figure \ref{cdLM}, we have that $\n$ is not removed in $M$, and we already argued that $\n$ is not path-connected in $M$. Thus $\n$ is pre-removed, and if $\n$ is potentially $k$-connected in $M$, then PT2$(\n,k)$ must hold in $M$. By the definition of PT2 (see Definition \ref{DefPot}), we have that $\prt$ is $k$-connected in $M$ and we already concluded it must also be $k$-connected in $N$. Thus $\n$ is potentially $k$-connected in $N$ as well (since $\LeftF(\n)=\RightF(\n)=\prt$ in $N$), and R1 holds in $N$.

For R2, we prove that all path-connected nodes $x\notin \{\Root,\bot\}$ in $N$ are tree-like. If $x$ is path-connected in $N$ then
it is path-connected in $M$, and hence it is a tree-like node in $M$,
and this implies that it is tree-like in $N$. Similar to the proof of R1, all nodes that were path connected in $M$ are path connected in $N$, and so R2 holds in $N$ as well.

Once again, $\Set(N)=\Set(M)$ follows immediately from the fact that $\DeletedP$ has the same extension in $M$ as in $N$, and from our observations that for any path-connected node $x$, the left (correspondingly right) descendants of $x$ in $M$ and in $N$ form the same set.

\paragraph*{Assume $s\in\Step(\Sys,\rom{v9})$.} The proof for this case if identical to the proof for the case of $s\in\Step(\Sys,\rom{f9})$, for all practical purposes.

\underline{This ends the proof that regularity is an inductive invariant.}
\end{proof}

\begin{SP}
\label{sp:rem-children}
Let $(M, N)$ be a step of the algorithm. If $x\in\AddressT^M$ and $\RemovedP(x)^M$, then $\LeftF^M(x)=\LeftF^N(x)$ and $\RightF^M(x)=\RightF^N(x)$.
\end{SP}

\begin{proof}
This is exactly what Corollary \ref{cor:constantly-focused} states.
\end{proof}

\begin{SP}[uses: \ref{Cl4.9}, \ref{cor:tree-like-path}]
\label{sp:pkc-after-disconnect}
Let $(M, N)$ be a step such that $M$ is regular, let $x\neq\bot$ be an address of $M$, and let $k\in\omega$ be a key value 
such that $M\models\Root\arrow_k^* x$ but $N\models\Root\notarrow_k^* x$. Then $x$ is potentially $k$-connected in $N$.

In addition, if $k=\KeyF(x)$, then both $\Root\arrow_k^*\LeftF(x)$ and $\Root\arrow_k^*\RightF(x)$ in $N$.
\end{SP}

\begin{proof}
Since $N\models\Root\notarrow_k^* x$, Step-Property \ref{Cm2} implies that $N\models \Root\notarrow^* x$. 
Thus, by lemma \ref{Cl4.9}, $(M,N)$ executes either f8, r8 or v6, and so $x=\n$.

If the step executes f8, then $\n=\LR(\prt,\lft)$ in $M$, and $\r=\LR(\prt,\lft)$ in $N$, and $\RightF(\n)=\LeftF(\n)=\r$ in both $M$ and $N$. 
By the regularity of $M$: 
\begin{enumerate*}[label=(\arabic*)]
    \item $\prt$ is $k$-connected in $M$, and so also in $N$ (arc $(\prt,\n)$ is not an arc of path $\Root\arrow^*\prt$);
    \item $\n$ is not confluent in $M$, and so path-connected in $N$; and 
    \item since $\prt\arrow_k\n$ in $M$, then $\prt\arrow_k\r$ in $N$, and so PT2$(\n,k)$ holds in $N$.
\end{enumerate*}
We conclude that the claim holds.

If the step executes v6, then $\n=\LR(\prt,\lft)$ in $M$, and $\child=\LR(\prt,\lft)$ in $N$, and $\n\arrow\child\wedge \n\arrow\bot$ in both $M$ and $N$.
By the regularity of $M$:
\begin{enumerate*}[label=(\arabic*)]
    \item $\prt$ is $k$-connected in $M$, and so also in $N$ (arc $(\prt,\n)$ is not an arc of path $\Root\arrow^*\prt$);
    \item $\n$ is not confluent in $M$, and so path-connected in $N$; and 
    \item since $\prt\arrow_k\n$ in $M$, then $\prt\arrow_k\child$ in $N$.
\end{enumerate*}
We know that $\Set(M)=\Set(N)$. If $k\in\Set(N)=$ then $\n\notarrow_k\bot$ in $M$, and so PT2$(\n, k)$ holds in $N$ via $\n\arrow\child$. 
Otherwise, $k\notin\Set(N)$, and so $\child\arrow^*_k\bot$, implying that $\bot$ is $k$-connected, and so PT2$(\n, k)$ holds in $N$ via $\n\arrow\bot$.
If $k=\KeyF(\n)$, then $k\notin\Set(N)$, and we already showed that both $\bot$ and $\child$ are $k$-connected in $N$.
\end{proof}

\begin{SP}
\label{sp:pdc-remains-pkc}
Let $s=(M, N)$ be a step such that $M$ is regular, let $k\in\omega$ be a key value, and let $x\in\AddressT$ such that $\Root\notarrow^* x$ and $x$ is potentially $k$-connected in $M$. Then $x$ is potentially $k$-connected in $N$.
\end{SP}

\begin{proof}
Since key-values are immutable, it suffices to prove the claim for steps that change $\LeftF$ or $\RightF$.

Note that $x\neq\bot$, since $\bot$ is always path-connected.

\begin{description}
    \item[$s\in\Step(p,i3)$:] Since this step only removes arc $(y,\bot)$, it is impossible that the step changes the potential-connectivity of any node.
    \item[$s\in\Step(\Sys,f6)$:] Clearly, any node $x$ that was potentially connected in a way dependent on the arc $(\r,\rl)$ remains potentially connected via the arcs $(\r,\new)$ and $(\new,\rl)$.
    \item[$s\in\Step(\Sys,f7)$:] Clearly, any node $x$ that was potentially connected in a way dependent on the arc $(\n,\ell_0)$ remains potentially connected via the arcs $(\n,\r)$, $(\r,\new)$, and $(\new,\ell_0)$.
    \item[$s\in\Step(\Sys,f7)$:] Clearly, any node $x$ that was potentially connected in a way dependent on the arc $(\prt,\n)$ remains potentially connected via the arcs $(\prt,\r)$.
    \item[$s\in\Step(\Sys,v6)$:] Clearly, any node $x$ that was potentially connected in a way dependent on the arcs $(\prt,\n)(\n,\child)$ remains potentially connected via the arcs $(\prt,\child)$.
    \item[$s\in\Step(\Sys,v7)$:] Since this step only disconnects $\bot$ from $\n$, the claim holds trivially.
    \item[$s\in\Step(\Sys,v8)$:] Clearly, any node $x$ that was potentially connected in a way dependent on the arc $(\n,\child)$ remains potentially connected via the arcs $(\n,\prt)$ and $(\prt,\child)$.
\end{description}
\end{proof}

\begin{SP}
\label{sp:dc-chd-leads-to-chd}
Let $s=(M, N)$ be a step such that $M$ is regular, and let $x\in\AddressT$ such that $M\models\Root\notarrow^* x$, and either $\LeftF^M(x)\neq\LeftF^N(x)$ or $\RightF^M(x)\neq\RightF^N(x)$. Let $y$ be the new child of $x$ in $N$, and let $z$ be the child that did not change. Then for every $k\in\omega$ such that $N\models\Root\arrow_k^*z$, also $N\models\Root\arrow_k^*y$.
\end{SP}

\begin{proof}
There are only two types of steps that match this scenario: 
\begin{description}
    \item[$s\in\Step(\Sys,v7)$:] In this step, arc $(\n,\bot)$ is replaced with arc $(\n,\prt)$. Note that $x=\n$, $y=\prt$ and $z=\child$. Since $\Set(M)=\Set(N)$ and $\child=\LR(\prt,\lft)$ in both $M$ and $N$, the claim holds trivially.
    \item[$s\in\Step(\Sys,v8)$:] In this step, arc $(\n,\child)$ is replaced with arc $(\n,\prt)$. Note that $x=\n$, $y=\prt$ and $z=\child$. Since $\Set(M)=\Set(T)$ and $\child=\LR(\prt,\lft)$ in both $M$ and $N$, the claim holds trivially.
\end{description}
\end{proof}

\begin{SP}[uses: \ref{inv:rem-chld-not-bot}]
\label{sp:rem-not-pc}
Let $s=(M,N)$ be a step such that $M$ is regular. Then for any $x\in\AddressT^M$, if $\neg\RemovedP(x)$ in $M$ and $\RemovedP(x)$ in $N$, then $\Root\notarrow^* x$ in $M$.
\end{SP}

\begin{proof}
Let $x$ be a removed node of $N$. From Inductive invariant \ref{inv:rem-chld-not-bot} we have that $\RightF(x)=\LeftF(x)$. Thus $x$ is not tree-like. If $x$ is path-connected, since $M$ is regular and regularity is an invariant, then $N$ is regular. Thus, R2 holds in $N$, and so $x=\n$. Since the step $s$ changes the extension of the $\RemovedP$ predicate, by observation of the algorithm, $s\in\Step(\Sys,\rom{f9})\cup\Step(\Sys,\rom{r9})\sup\Step(\Sys,\rom{v9})$.

If $s\in\Step(\Sys,\rom{f9})$, then the claim follows from the fact that $\prt$ is path-connected in $M$ by R0, and that both $\prt$ and $\n$ point to $\r$ (by the control-dependent invariants of Figure \ref{cdLM}), implying that $\n$ cannot be path-connected (otherwise $\r$ would be confluent, contradicting R2 in $M$).

If $s\in\Step(\Sys,\rom{v9})$, then the claim follows from the fact that $\prt$ is $\KeyF(\n)$-connected in $M$ by R0, and that $\n$ points to $\prt$ (by the control-dependent invariants of Figure \ref{cdLM}). In this case, if $\n$ were path-connected, then $\n$ would be a parent of $\prt$ in $M$ (contradicting R0).
\end{proof}

\begin{corollary}
\label{cor:rem-not-pc}
The combination step-properties \ref{InvRem}, \ref{Eq5a} and \ref{sp:rem-not-pc} implies that in any regular state $M$, for any address $x$, if $\RemovedP(x)$ then $\Root\notarrow^* x$ in $M$.
\end{corollary}


\section{Histories and their properties}
\label{SecHistory}
At this stage we know that regularity is an inductive invariant (Theorem \ref{InvReg}), 
and hence we may {\em assume} that all states are regular.

\begin{definition}
A {\em history} is an infinite sequence of {\em regular} states $\oM =(M_i\mid i\in\omega)$ such that $M_0$ is an initial address structure, and for every index $i$ the 
pair $(M_i,M_{i+1})$ is a step by one of the processes $p\in\{0,\ldots,N\}$. 
\end{definition}


A history sequence models an execution of the CF algorithm.

We clarify the distinction between local and history statements:
A local statement $\varphi$ is a statement in the language $\mathcal{L}_{AS}$ of address structures, 
and for any structure $M$, either $M\models \varphi$ or else $M\models \neg\varphi$.
The truth of a {\em history statement} is evaluated at any given history sequence \oM\ to be
true or false for that history. 

A history statement may include quantification over history indexes, 
which would be meaningless in local statements (which use language $\mathcal{L}_{AS}$).
A history statement involves local state statements of the form $(\varphi)^{M_i}$ where
$M_i$ is a reference to the $i$th state of a history sequence and $\varphi$ is a local statement.
If $\psi$ is a history statement, and $\oM$ a history,
then either $\psi$ holds in $\oM$ (denoted $\oM\models \psi$),
or its negation holds (denoted $\oM\models \neg \psi$).

The following is an example of a history statement which turns out to be useful.

\begin{lemma}[uses: \ref{Cm2}, \ref{Eq5a}]
\label{T4.4}
Let $\oM$ be a history, $x\neq \bot$ an address, and $i$ an index such that $M_i\models \Root\arrow_k ^* x$. 
Let $j>i$ be an index such that $M_j\models \Root \arrow^* x$. Then $\Root\arrow_k ^* x$ in $M_j$ as well.
\end{lemma} 
\begin{proof}
Let $j_0$ be the maximal index such that $j_0\leq j$ and $M_{j_0}\models  \Root\arrow^*_k x$. 
Then $i\leq j_0$. We want to prove that $j_0=j$. Assume for a contradiction that $j_0< j$. 
So $j_0+1\leq j$. Then Step-property \ref{Cm2} implies directly that $x$ is not path-connected in $M_{j_0+1}$ 
(since by the maximality of $j_0$, $x$ is not $k$-connected  in $M_{j_0+1}$). 
By Invariant \ref{Eq5a}, a node that is not path connected, cannot later become path-connected, 
contradiction to the assumption that $M_j\models \Root\arrow^* x$.
\end{proof}



\begin{lemma}[uses: \ref{lemma:pdc-rem-or-locked}, \ref{Eq5a}]
\label{Lem4.4}
Let $\oM$ be a history.
Let $\ell$ be a history index, and let $y$ be an address such that $y$ is $k$-connected in $M_\ell$,
but not path-connected in $M_{\ell+1}$.
Then for every index $j\geq \ell$, $\DeletedP^{M_{\ell}}(y)\, \iff\, \DeletedP^{M_{j}}(y)$.
\end{lemma}

\begin{proof}
Since $y$ is not path-connected in $M_{\ell+1}$ and $y\neq \bot$, $y$ remains path-disconnected for all $M_i$ such that $i\geq\ell+1$, by Step-property \ref{Eq5a}. 
It follows from Corollary \ref{lemma:pdc-rem-or-locked} that either $\RemovedP(y)$ or $\LockedP(y, \Sys)$ in $M_i$. 
In either case, no step may change the truth of $\DeletedP(y)$.
\end{proof}

\begin{definition}[Abstract $k$-scanning]
\label{Ekpath}
Let $\oM= (M_i\mid i\in\omega)$ be a history.
For any key value $k\in \omega$, an {\em abstract $k$-scanning} in $\oM$ is a finite sequence of triples  
\[T= (\langle \ell_0, x_0,y_0\rangle, \ldots, \langle\ell_i, x_i,y_{i}\rangle, \ldots, \langle\ell_n,x_n,y_{n}\rangle)\]
such that the following hold. 

\begin{enumerate}
    \item Each $\ell_i\in \omega$ is an index, and 
    the indexes are increasing: $\ell_0< \ell_1< \cdots< \ell_n$. 
     
    \item $x_i$ and $y_i$ are addresses in $M_{\ell_i}$, and $x_i\neq \bot$ and is potentially $k$-connected there.
     
    \item For every $0\leq i<n$,
    \begin{enumerate}
        \item $y_i \neq\bot  \Rightarrow x_{i+1}=y_i$ (a handshake), or\label{item:handshake}
        \item $(y_i=\bot \vee \KeyF(x_i)= k) \Rightarrow x_{i+1}= x_i$ (a traversal stutter).
    \end{enumerate}
    
    \item  For every $0\leq i\leq n$,
    one of the following possibilities holds.
    \begin{enumerate}
        \item $M_{\ell_i} \models x_i\arrow_k y_{i}$ (a $k$-search triple);
        \item $M_{\ell_i}\models \KeyF(x_i)=k  \wedge \RemovedP(x_i)\wedge y_i =\RightF(x_i)$ (a backtracking triple);\label{item:backtrack}
        \item $M_{\ell_i}\models y_i=x_i$ (a delaying triple).
    \end{enumerate}
    
    \end{enumerate}
\end{definition}

The abstract $k$-scanning definition is meant to represent the abstract notion of a $k$-searching traversal, using the standard ``hand over hand'' approach. 
In the case of the CF algorithm, the two ``hands'' of a working process $p>0$ are represented by the local variables $\nd_p$ and $\nxt_p$, 
that match $x_i$ and $y_i$ in the $k$-scanning triples, respectively. 
This is seen most clearly by analysing the diagrams of Figure~\ref{fig:flowcharts}.
Each triple $(\ell,x,y)$, represents a step
$(M_{\ell-1},M_\ell)$ where $x=\nd_p^{M_\ell}$,
and $y= \nxt^{M_\ell}$ is the candidate for the next value of $\nd_p$.
This abstract notion is introduced in order to formulate the minimal assumptions that are nevertheless
sufficiently strong to enable a proof of Theorem \ref{T4.2} which is the main tool in the linearizability proof of the
CF algorithm.


\begin{theorem}[The Scanning Theorem]
\label{T4.2}
Let $\oM= (M_i\mid i\in\omega)$ be a history,
and let $T= (\langle \ell_0,x_0,y_0\rangle,\ldots,\langle \ell_n, x_n, y_n\rangle)$ 
be an abstract $k$-scanning in $\oM$. 
Suppose that $y_0\neq \bot$ is $k$-connected in $M_{\ell_0}$. 
Then, for some index $j$ such that $\ell_0\leq j \leq \ell_n$, $y_{n}$ is $k$-connected in $M_j$,
and \linebreak{}$\DeletedP(y_{n})^{M_j} \iff \DeletedP(y_{n})^{M_{\ell_n}}$.
\end{theorem}

\begin{proof}
The proof is by induction on $n\geq 1$, and for any fixed $n>1$ the proof is by induction on
$\ell_n-\ell_0$. We start with the case $n=1$. 
So $T=(\langle \ell_0,x_0,y_0\rangle, \langle \ell_1,x_1,y_1\rangle)$, $\ell_0<\ell_1$ and 
$y_0\neq \bot$ is assumed to be $k$-connected in $M_{\ell_0}$. 
$y_0\neq\bot$ implies (by item \ref{item:handshake} of Definition \ref{Ekpath}) that $x_1=y_0$ and so:

\begin{enumerate}[label=\Alph*.]
    \item $M_{\ell_0} \models \Root\arrow^*_k\, y_0 = x_1$.
    \item $M_{\ell_1} \models  y_1=x_1 \vee x_1\arrow_k y_1 \vee (\KeyF(x_1)=k\wedge \RemovedP(x_1)\wedge y_1=\RightF(x_1))$.\label{item-b}
\end{enumerate} 

We have to prove that for some index $j$ such that $\ell_0\leq j \leq \ell_1$, $y_{1}$ is $k$-connected in $M_j$, 
and $y_1$ is deleted in $M_{j}$ if and only if $y_1$ is deleted in $M_{\ell_1}$. 

If $M_{\ell_1}\models\Root\arrow_k^* y_1$, the claim holds for $j=\ell_1$. So assume that 
\begin{equation}
    \label{eqn:y1-not-kc}
    M_{\ell_1}\models\Root\notarrow_k^* y_1.
\end{equation}
It suffices to find an index $\ell$ such that 
\begin{equation}
    \label{eq:sufficient}
    (\ell_0\leq \ell \leq \ell_1 \wedge \Root\arrow_k  y_1)^{M_\ell},
\end{equation} 
because in that case we pick such $\ell$ that is maximal, which entails that $\ell<\ell_1$, $y_1$ is $k$-connected in $M_\ell$ but is not 
$k$-connected in $M_{\ell+1}$, and hence $y_1$ is not path-connected in $M_{\ell+1}$ (by Step-Property \ref{Cm2}). 
Then  by Lemma \ref{Lem4.4}, for every $m\geq \ell$, $\DeletedP(y_1)^{M_\ell} \iff \DeletedP(y_1)^{M_m}$, and $\ell$ is the required index.

We now check each of the possibilities of item \ref{item-b} above, and find in each case an index
$\ell$ as in (\ref{eq:sufficient}):
\begin{description}
    \item[${M_{\ell_1}} \models y_1=x_1$:] In this case, $y_1$ is $k$-connected in 
	$M_{\ell_0}$, since $y_1=x_1=y_0$, and $y_0$ is $k$-connected in 
	$M_{\ell_0}$. Then the condition at (\ref{eq:sufficient}) holds and
	the required index is $\ell=\ell_0$.
    
    \item[$M_{\ell_1}\models x_1\arrow_k y_1$:] In this case, $x_1\neq y_1$, and $\KeyF(x_1)\neq k$ (by the definition of $\arrow_k$). This entails that 
	$M_{\ell_1}\models\Root\notarrow_k^* x_1$ (otherwise  $M_{\ell_1}\models\Root\arrow_k^* x_1\arrow_k y_1$ would contradict (\ref{eqn:y1-not-kc})). Taking into account Item A above, let $r\geq \ell_0$ be the maximal
	index such that $r\leq \ell_1$ and $M_r\models\Root\arrow_k^* x_1$. Then $r<\ell_1$ (since $M_{\ell_1}
	\models \Root\notarrow_k^* x_1$) and   
	\begin{equation}
		\label{Eq24}
		\mbox{$\Root\notarrow_k^* x_1$ at $M_{r+1}$ and all
		subsequent indexes until $\ell_1$.} 
    \end{equation}
    If $M_r\models x_1\arrow_k y_1$, then  $M_r\models\Root\arrow_k^* x_1\arrow_k y_1$ shows that (\ref{eq:sufficient}) holds for $\ell=r$ entailing that the required index can be found.
    
    Hence we may assume that $M_r\models x_1\notarrow_k y_1$. Let
	$s\in\{r+1,\ldots,\ell_1\}$ be the last index such that $M_s\models x_1\notarrow_k y_1$. 
	Then $s<\ell_1$ and $M_{s+1}\models x_1\arrow_k y_1$. 
	Thus step $(M_s,M_{s+1})$ changes the left or the right child of $x_1$ and this indicates that
	\begin{equation}
    	\label{Eq25}
    	\mbox{$x_1$ is not removed at $M_s$ and at $M_{s+1}$,}
	\end{equation}
	(by Step-property \ref{sp:rem-children}).
	By applying Step-property \ref{sp:pkc-after-disconnect} to the step $(M_r, M_{r+1})$
	we deduce that $x_1$ is potentially $k$-connected at $M_{r+1}$. Then,
	applying Step-property \ref{sp:pdc-remains-pkc} to all the steps from $M_{r+1}$ to 
	$M_{s+1}$, 
	we conclude that $x_1$ is potentially $k$-connected in $M_{s+1}$.
	Equations (\ref{Eq24}) and (\ref{Eq25}), imply that PT2($x_1,k)$ at $M_{s+1}$. 
	By Step-property \ref{sp:dc-chd-leads-to-chd} we have that $M_{s+1}\models\Root\arrow^*_k y_1$ must hold, 
	and so $y_1$ is $k$-connected at $M_{s+1}$.

    \item[$M_{\ell_1}\models\KeyF(x_1)=k\wedge \RemovedP(x_1)\wedge y_1=\RightF(x_1))$:] In this case, 
	$x_1\neq \Root$ (as $k\in \omega$), and $x_1\neq y_1$ (as $y_1=\RightF(x_1)$; see Invariant \ref{inv:address-points-to-root}).
	Since $M_{\ell_1}\models\RemovedP(x_1)$, by Corollary \ref{cor:rem-not-pc}, we have that $M_{\ell_1}\models\Root\notarrow^* x_1$.
	Thus, as $x_1$ is $k$-connected at $M_{\ell_0}$, there is a maximal index $r\in\{\ell_0,\ldots,\ell_1\}$ 
	such that $M_r\models\Root\arrow_k^* x_1$. Then $r<\ell_1$, and $x_1$ is not $k$-connected (and hence
	not connected) at $M_{r+1}$
	and at every subsequent structure until $M_{\ell_1}$ (see Lemma \ref{T4.4}).

	Step-property \ref{sp:pkc-after-disconnect} can be applied to step $(M_r,M_{r+1})$ to conclude that
	\begin{equation}
	    \mbox{$x_1$ is potentially $k$-connected at $M_{r+1}$.}
	\end{equation}
	Moreover, since $\KeyF(x_1)=k$, 
	if $x_1\arrow y_1$ at $M_r$, then $\Root\arrow_k^* y$ at $M_{r+1}$.     
	In this case an index satisfying (\ref{eq:sufficient}) is found.
	
	Thus we may assume that $M_r\models x_1\notarrow y_1$. Let then  $s\in\{r+1,\ldots,\ell_1\}$ be the maximal index such that $M_s\models x_1\notarrow y_1$. Then $s<\ell_1$ (since 
	$M_{\ell_1}\models y_1=\RightF(x_1)$), 
	but $M_{s+1}\models x_1\arrow y_1$. 
	
	Applying Step-property \ref{sp:pdc-remains-pkc} to all the steps from step $(M_{r+1},M_{r+2})$ to step 
	$(M_s,M_{s+1})$, we have that $x_1$ is potentially $k$-connected at states $M_{r+1}$ to $M_{s+1}$. 
	We know that $x_1$ cannot be removed in $M_s$ or in $M_{s+1}$, since one of its children changed in step 
	$(M_s, M_{s+1})$.
    Since $x_1$ is potentially $k$-connected, not $k$-connected and not removed in $M_{s}$, it must be that PT2$(x_1,k)$ holds in $M_{s}$. 
    By Step-property \ref{sp:dc-chd-leads-to-chd} we have that $M_{s+1}\models\Root\arrow^*_k y_1$ must hold. Since $M_{\ell_1}\models\Root\notarrow_k^* y_1$, there must be an index $t\in\{s+2,\ldots,\ell_1-1\}$ such that $M_t\models\Root\arrow_k^* y_1$ but $M_{t+1}\models\Root\notarrow_k^* y_1$, and again we are in a situation in which an index $\ell=t$ that satisfies (\ref{eq:sufficient}) is found.
\end{description}

This concludes the proof of the base case of the induction.

Now suppose that $n>1$. In case $y_1=\bot$, we have $x_2=x_1=y_0$, and then, by removing the triple
$(\ell_1,x_1,y_1)$ from the abstract scan $T$, a shorter abstract scan is obtained
to which the inductive hypothesis applies. So we may assume now that $y_1\neq \bot$.
Recall that $y_0$ is assumed to be $k$-connected in $M_{\ell_0}$.

The case $n=1$ of the theorem applies to the $k$-scan 
$\left(\langle \ell_0,x_0,y_0\rangle, \langle \ell_1,x_1,y_1\rangle\right)$, and so there is an
 index $\ell\in\{\ell_0,\ldots,\ell_1\}$ such that $y_1$ is $k$-connected in $M_\ell$.
 Let $p$ be the node on the $k$-path before $y_1$ in $M_{\ell}$ (exists since
$y_1\neq\Root$). Then
apply the inductive assumption to the shorter path 
$T' = (\langle \ell, p,y_1\rangle, \langle  \ell_2, x_2,y_2\rangle,\ldots, \langle \ell_n x_n,y_{n}\rangle)$ 
and get an index $\ell'$ such that $\ell\leq \ell'\leq \ell_n$ and $y_{n}$ is $k$-connected in $M_{\ell'}$, and
$\DeletedP(y_{n})^{M_j} \iff \DeletedP(y_{n})^{M_{\ell_n}}$.
\end{proof}


\section{Linearizability of the Contention-Friendly Algorithm}
\label{Sec6}
We are ready to prove the linearizability of the CF algorithm. 
Let $\oM$ be an arbitrary history sequence of the algorithm.
For any index $i\in\omega$, let $\Set(M_i)\subset \KeyT$ be the set of key values
represented by state $M_i$ of the history (the key values $\infty$ and $-\infty$ are not
members of $Set(M_i)$; see Definition \ref{DefSet}). 

We say that step $(M,N)$ is {\em set-preserving} if $\Set(M)=\Set(N)$, i.e. the step did not change
the \Set\ value. An operation is considered to be set-preserving if all its steps are set-preserving.

Let $E$ be a terminating data operation execution by process
$p>0$ in history \oM\ ($E$ is one of $\bfcontains(k)$, $\bfdelete(k)$ and $\bfinsert(k)$).
$E$ has a unique returning boolean value $\returnVal(E)$. 
We say that $E$ is {\em successful} if and only if $\returnVal(E)=\TRUE$.
Let $i_0=\mathit{inv}(E)$ and $r=\mathit{res}(E)$ be the history
indexes such that $s_{i_0+1} =(M_{i_0}, M_{i_0+1}) = \Begin(E)\in\Step(p,\rom{m0},f)$ (where $f$ points to the first instruction of $E$) is the invocation of $E$, 
and $s_{r+1}=(M_r,M_{r+1})=\End(E)\in \Step(p,rt,\rom{m0})$ (where $rt$ is a \return\ instruction) is the response of $E$\footnote{We denote step $(M_{i-1},M_i)$ by $s_i$ (rather than by $s_{i-1}$) so that $S_i$ is the consequence of step $s_i$.}.
The sequence of states $(M_{i_0},\ldots,M_r)$ is said to be the {\em extension interval} of $E$, and the steps $(M_i,M_{i+1})$ for $i_0\leq i\leq r$ are the steps
of that extension. So the invocation $\Begin(E)$ and the response $\End(E)$ are the first and last steps of
this extension of $E$; these are steps by $p$ but other steps 
in this extension can be by other processes. We identify
$E$ with the set of all steps by $p$ that are in the extension of $E$, so
that $\Begin(E)$ and $\End(E)$ are the first and last steps of $E$. If we remove the last (returning) step from the terminating
operation execution $E$, then the resulting set of steps form the {\em search part} of
$E$, denoted $\search(E)$. The search part is thus, the set of steps $s_i$ by $p$ such that $i_0<i\leq r$. 
If we enumerate the set of the steps in $\search(E)$ in increasing order $\ell_0,\ldots,\ell_n$, then
the scanning steps of $E$ are the steps
$s_{\ell_0}=(M_{\ell_0-1}, M_{\ell_0}),\ldots, s_{\ell_n}=(M_{\ell_n-1},M_{\ell_n})$, and
$(M_{r},M_{r+1})$ is the return step of $E$.

Our aim is to define for every terminating operation execution $E$
 by process $p>0$ an index $\mathit{inv}(E)\leq \ell(E)\leq \mathit{res}(E)$
which has the following `return' properties Ret1--Ret4.
$\ell$ is called the {\em linearization index} of $E$, and $(M_{\ell}, M_{\ell+1})$ its {\em linearization point}.

\begin{enumerate}
\item[Ret1:] If $E$ is a $\bfcontains(k_p)$ execution, then $E$ is set-preserving, and $\returnVal(E)=\TRUE$ if and only if $(k_p\in \Set(M_{\ell(E)}))$.

\item[Ret2:] If $E$ is a $\bfdelete(k_p)$ execution then one 
of the following possibilities holds.
\begin{enumerate}
\item $\returnVal(E)=\TRUE$, and $k_p\in\Set(M_{\ell(E)})$ but $\Set(M_{{\ell(E)}+1}) = \Set(M_{\ell(E)})\setminus \{k_p\}$.
Any step in $E$ other than $(M_{\ell(E)},M_{{\ell(E)}+1})$ is set-preserving.
\item $\returnVal(E)=\FALSE$, $E$ is set-preserving and $k_p\notin \Set(M_{\ell(E)})$.
\end{enumerate}

\item[Ret3:] If $E$ is a $\bfinsert(k)$ execution then one of the 
following possibilities holds.
\begin{enumerate}
\item $\returnVal(E)=\TRUE$, and $k_p\not\in \Set(M_{\ell(E)})$ but $\Set(M_{{\ell(E)}+1}) = \Set(M_{\ell(E)}) \cup \{k_p\}$. 
Any step in $E$ other than $(M_{\ell(E)},M_{{\ell(E)}+1})$ is set-preserving.
\item $\returnVal(E)=\FALSE$, $E$ is set-preserving and $k_p\in \Set(M_{\ell(E)})$.
\end{enumerate}

\item[Ret4:] If $(M_i,M_{i+1})$ is any step by the \Sys\ process, then 
$\Set(M_i)=\Set(M_{i+1})$.
That is, \Sys\ steps are set-preserving.
\end{enumerate}

Before we prove these four return properties, we prove Theorem~\ref{thm:operations-are-kscans} below, 
which establishes that any terminating operation execution by a working process $p>0$
induces an abstract scanning sequence (Definition \ref{Ekpath}).
Then, the Scanning Theorem \ref{T4.2} applies, producing the linearization point $\ell(E)$.

\begin{theorem}
\label{thm:operations-are-kscans}
Let $\oM$ be a history of the CF algorithm, let $p>0$ be a working process, let $E$ be a terminating data operation executed by $p$ in $\oM$, and let $k$ be the value of $k_p$ during the execution of $E$. Then the scanning  steps $s_i$ of $p$ in the interval $\left[\Begin(E),\End(E)\right)$ induce an abstract $k$-scan.
\end{theorem}

Before we delve into the details of the proof, we present an intuitive overview of the intent of this theorem.

Recall the an abstract $k$-scanning (Definition \ref{Ekpath}) is meant to be an abstract representation of the traversal of a working process $p>0$ through the virtual graph. 
In most cases, the steps of the data operations of the CF algorithm naturally induce an abstract $k$-scanning, even the backtracking steps. 
For these steps, the proof that a valid abstract $k$-scanning is induced is a simple case-analysis process.
There is, however, one case that does not match the abstract $k$-scanning definition. This case arises when two concurrent insertion operations contend for the same insertion location:
Consider the case in which process $p$ is attempting to insert value $k$, and does not find a logically-deleted node with value $k$ to ``un-delete''. This means that a new node must be inserted as a leaf of the $k$-connected node referenced by $\nd_p$. This case is indicated by $\nxt_p$ referencing the $\bot$ node when $\ControlF(p)=\rom{i1}$. Before $p$ has a chance to ``discover'' that this is the case, and to ``decide'' to execute instruction i3, some other working process $q$ preempts $p$, and inserts a new node of its own in the same spot where $p$ wanted to insert $k$ (as a $k$-child of $\nd_p$). This creates a situation in which $\KeyF(\nd_p)\neq k_p$ but $\nd_p\notarrow_k\nxt_p=\bot$ (since $q$ inserted a node in between $\nd_p$ and $\bot$). The abstract $k$-scanning definition does not cover this situation, which may arise in some steps of the form $\Step(p,\rom{i1},\rom{i3})$.

The solution to this discrepancy arises from the fact that $\nd_p\arrow_k\nxt_p=\bot$ was true in a previous step of $p$ (the step in which $\bot$ was assigned to $\nxt_p$, before $q$ preempted $p$), yielding a valid $k$-scanning triple. Also, due to the preemption of $p$ by $q$, when executing i3, $p$ will go back to i1 to continue trying to find a spot to insert $k$. In this case, $\nxt_p$ will be re-defined to the current $k$-child of $\nd_p$, thus, inducing a valid $k$-scanning triple once again. Thus, in the case of a $k$-scan induced by \bfinsert, steps $\Step(p,\rom{i1},\rom{i3})$ are not included.

We now turn the proof itself. Recall that in the case of the CF algorithm, the $x_i$ component of a $k$-scan triple represents $\nd_p$, and $y_i$ represents $\nxt_p$. We consider each of the forms of non-returning steps of $p$ in $E$, and show that each step $s_{\ell_i}=(M_{{\ell_i}-1}, M_{\ell_i})$ induces a valid $k$-scan triple $(\ell_i,x_i,y_i)$ as defined in item 4 of Definition \ref{Ekpath}. Additionally, we show that every pair of these triples has a valid relationship, as defined in item 3 of Definition \ref{Ekpath}.

We prove the claim individually for each of the data operations. We recommend that the reader commences with the diagrams of Figure~\ref{fig:flowcharts} in hand.

\begin{lemma}
\label{lemma:contains-are-kscans}
Theorem \ref{thm:operations-are-kscans} holds if $E$ is a \bfcontains\ operation.
\end{lemma}

\begin{proof}
\begin{description}
    \item[$s_i\in\Step(p,\rom{m0},\rom{c1})$:] In this case, $i=\ell_0$ by the definition of an operation. As illustrated in Figure \ref{FIGflowchart1}, $x_0=y_0 =\Root$ in $M_{\ell_0}$, and so $s_i$ is a delaying triple.
    \item[$s_i\in\Step(p,\rom{c1},\rom{c1})$:] As illustrated in Figure \ref{FIGflowchart1}, $y_{i-1}\neq\bot$, and so $x_i=y_{i-1}$ (a handshake), and $\KeyF(x_i)\neq k$, and so $(x_i\arrow_k y_i)^{M_{\ell_i}}$ (k-search triple).
    \item[$s_i\in\Step(p,\rom{c1},\rom{c2})$:] As illustrated in Figure \ref{FIGflowchart1}, $y_{i-1}\neq\bot$, and so $x_i=y_{i-1}$ (a handshake), and $\KeyF(x_i)=k$, and so $y_i=y_{i-1}$ (a delaying triple).
\end{description}

Thus the sequence of triples induced by the steps of $p$ in $E$ is an abstract $k$-scan as defined in \ref{Ekpath}.
\end{proof}

\begin{lemma}
\label{lemma:deletes-are-kscans}
Theorem \ref{thm:operations-are-kscans} holds if $E$ is a \bfdelete\ operation.
\end{lemma}

\begin{proof}
\begin{description}
    \item[$s_i\in\Step(p,\rom{m0},\rom{d1})$:] In this case, $i=\ell_0$ by the definition of an operation. As illustrated in Figure \ref{FIGflowchart2}, $x_0=y_0 =\Root$ and so $s_i$ is a delaying triple.
    \item[$s_i\in\Step(p,\rom{d1},\rom{d1})$:] As illustrated in Figure \ref{FIGflowchart2}, $y_{i-1}\neq\bot$, and so $x_i=y_{i-1}$ (a handshake), and $\KeyF(x_i)\neq k$, and so $(x_i\arrow_k y_i)^{M_{\ell_i}}$ (a k-search triple).
    \item[$s_i\in\Step(p,\rom{d1},\rom{d2})$:] As illustrated in Figure \ref{FIGflowchart2}, $y_{i-1}\neq\bot$ and so $x_i=y_{i-1}$ (a handshake), and $\KeyF(x_i)=k$, and so $y_i=y_{i-1}$ (a delaying triple).
    \item[$s_i\in\Step(p,\rom{d2},\rom{d1})$:] As illustrated in Figure \ref{FIGflowchart2}, $x_i=x_{i-1}$ (a traversal stutter), $\RemovedP(x_i)^{M{\ell_i}}$ and $y_i=\Right(x_i)$. Note that in any execution of \bfdelete, $s_i$ must be preceded by a step in $\Step(p,\rom{d1},\rom{d2})$, which means that $x_{i-1}=y_{i-1}$ and $\KeyF(x_{i-1})=k$. Since $x_i=x_{i-1}$, we have that $M_{\ell_i}\models\RemovedP(x_i)\wedge\KeyF(x_i)=k\wedge y_i=\RightF(x_i)$ (a backtracking triple).
\end{description}

Thus the sequence of triples induced by the steps of $p$ in $E$ is an abstract $k$-scan as defined in \ref{Ekpath}.
\end{proof}

\begin{lemma}
\label{lemma:inserts-are-kscans}
Theorem \ref{thm:operations-are-kscans} holds if $E$ is an \bfinsert\ operation.
\end{lemma}

\begin{proof}
\begin{description}
    \item[$s_i\in\Step(p,\rom{m0},\rom{i1})$:] In this case, $i=\ell_0$ by the definition of an operation. As illustrated in Figure \ref{FIGflowchart3},  $x_0=y_0 =\Root$ and so $s_i$ is a delaying triple.
    \item[$s_i\in\Step(p,\rom{i1},\rom{i1})$:] As illustrated in Figure \ref{FIGflowchart3}, $y_{i-1}\neq\bot$, and so $x_i=y_{i-1}$ (a handshake), and $\KeyF(x_i)\neq k$, and so $(x_i\arrow_k y_i)^{M_{\ell_i}}$ (a k-search triple).
    \item[$s_i\in\Step(p,\rom{i1},\rom{i2})$:] As illustrated in Figure \ref{FIGflowchart3}, $y_{i-1}\neq\bot$, and so $x_i=y_{i-1}$ (a handshake), and $\KeyF(x_i)=k$, and so $y_i=y_{i-1}$ (a delaying triple).
    \item[$s_i\in\Step(p,\rom{i2},\rom{i1})$:] As illustrated in Figure \ref{FIGflowchart3}, $x_i=x_{i-1}$ (a traversal stutter), $\RemovedP(x_i)^{M{\ell_i}}$ and $y_i=\Right(x_i)$. Note that in any execution of \bfinsert, $s_i$ must be preceded by a step in $\Step(p,\rom{i1},\rom{i2})$, which means that $x_{i-1}=y_{i-1}$ and $\KeyF(x_{i-1})=k$. Since $x_i=x_{i-1}$, we have that $M_{\ell_i}\models\RemovedP(x_i)\wedge\KeyF(x_i)=k\wedge y_i=\RightF(x_i)$ (a backtracking triple).
    \item[$s_i\in\Step(p,\rom{i3},\rom{i1})$:] As illustrated in Figure \ref{FIGflowchart3}, $x_{i-1}\notarrow_k\bot$ and $y_{i-1}=\bot$, and so $x_i=x_{i-1}$ (a traversal stutter) and $(x_i\arrow_k y_i)^{M_{\ell_i}}$ (a k-search triple).
    \item[$s_i\in\Step(p,\rom{i1},\rom{i3})$:] This step does not induce a valid triple. However, note that step $s_i$ must be followed by either a returning step (which is not part of any $k$-scan), or a step $s_{i+1}\in\Step(p,\rom{i3},\rom{i1})$. To see that $s_i$ does not interrupt the $k$-scan induced by steps $s_{i-1}$, $s_{i+1}$, observe that both $\nd_p$ and $\nxt_p$ are unchanged by step $s_i$. As a result, the analysis of the step $s_{i+1}\in\Step(p,\rom{i3},\rom{i1})$ shown above holds with $x_{i-1}$ and $y_{i-1}$ carrying over from step $s_{i-1}$.
\end{description}

Thus the sequence of triples induced by the steps of $p$ in $E$ is an abstract $k$-scan as defined in \ref{Ekpath}.
\end{proof}

\begin{corollary}
The data operations of the CF algorithm induce valid abstract scans.
\end{corollary}

We now turn to the proof the return properties Ret1 -- Ret4. We recommend that the reader commences with either the formal step definitions (Appendix \ref{Sec4.1}) or the code (Figure \ref{RLred}) in hand. As was the case before, we use $y$, $y_0$, $y_n$, etc., to denote the variable $\nxt_p$ in the various states, while $x$, $x_0$, $x_n$, etc., denotes the variable $\nd_p$ in the various states.

\begin{proof}[Proof of Ret1]
We must prove that the terminating $\bfcontains(k_p)$ operation $E$ is set-preserving, and we must find a linearization point $\ell(E)$ of operation $E$ such that $\returnVal_p(E)=\TRUE$ if and only if $k_p\in\Set(M_{\ell(E)})$.

In Theorem \ref{InvReg} we proved that the steps in $\Step(\rom{c1})\cup\Step(\rom{c2})$ are set-preserving, which means that $E$ is set-preserving.

Next, we find the linearization point of $E$. By Lemma \ref{lemma:contains-are-kscans}, $E$ induces a valid $k_p$-scan $T=(\langle \ell_0,x_0,y_0\rangle,\ldots,\langle \ell_n, x_n, y_n\rangle)$ such that $M_{\ell_0}\models y_0=\Root$ and so $y_0$ is $k_p$-connected in $M_{\ell_0}$.

Since $E$ is a terminating \bfcontains\ operation, the final step of $p$ in $E$, $s=(M_r, M_{r+1})$ (where $r\geq\ell_n+1$), is such that one of the following two possibilities holds:
\begin{enumerate}
    \item $s\in\Step(p,\rom{c1},\rom{m0})$ is a \FALSE\ returning step. In this case, $\nxt_p^{M_r}=\bot$ (a precondition of the step), and so $y_n=\bot$. We apply the Scanning Theorem \ref{T4.2} to the $k_p$-scan $T$, and find that there is some index $j\in\{\ell_0,\ldots,\ell_n\}$ such that $y_n=\bot$ is $k_p$-connected in $M_j$. Take $\ell(E)=j$, and conclude that $k_p\notin \Set(M_j)$.
    
    \item $s\in\Step(p,\rom{c2},\rom{m0})$. Step
    $s=(M_{\ell_n-1}, M_{\ell_n})$, the last step before
    the return step $(M_r,M_{r+1})$, is such that
    $x_n=\nd_p^{M_{\ell_n}}=\nxt_p^{M_{\ell_n}}$, and we get that 
    $\KeyF(x_n)=k_p$ (so that $x_n\neq\bot$). The return value is
    $\neg\DeletedP(\nd_p^{M_r})$. The abstract scanning
    induced by $E$ has $(\ell_n, x_n, x_n)$
    as its last triple. We add to that abstract scanning
    another triple: $t= (\ell_{n+1}, x_{n+1},y_{n+1}) = 
    (r, x_n, x_n)$ and get a longer (by 1)
    abstract scanning $T'=T\cdot t$.  $T'$ satisfies
    the requirements of abstract scanning \ref{Ekpath} as $x_n=y_n\neq\bot$, $x_{n+1}=x_n$,
    and since $M_r$ is a regular state, we get that $\nd_p^{M_r}=x_n$ is potentially $k_p$-connected.
    We now apply the Scanning theorem  to $T'$\footnote{Note that applying the Scanning theorem directly to $T$ would give us an index $\ell_0\leq j \leq \ell_n$ such that $x_n$ is $k_p$ connected in $M_j$, but we would not know that the \DeletedP\ predicate for $x_n$ is agreed upon between $M_j$ and $M_r$.}, and find that there is an index $j\in\{\ell_0,\ldots,r\}$ such that $x_n=y_{n+1}=x_n$ is $k_p$-connected in $M_j$ and $\DeletedP(x_n)^{M_j}\iff\DeletedP(x_n)^{M_{r+1}}$. Take $\ell(E)=j$, and
    the answer provided by the return step is the correct one.
     
\end{enumerate}
\end{proof}

\begin{proof}[Proof of Ret2]
We must find a linearization point $\ell(E)$ of a \bfdelete\ operation $E$, and prove
that it is correctly related with the value returned by the operation.

By Lemma \ref{lemma:deletes-are-kscans}, $E$ induces a valid $k_p$-scan $T=(\langle \ell_0,x_0,y_0\rangle,\ldots,\langle \ell_n, x_n, y_n\rangle)$ such that $M_{\ell_0}\models y_0=\Root$ and so $y_0$ is $k_p$-connected in $M_{\ell_0}$.

Since $E$ is a terminating \bfdelete\ operation, the returning step of $p$, $s=(M_r, M_{r+1})$ (where $r\geq\ell_n+1$), is such that one of the following holds:
\begin{enumerate}
   \item $s\in\Step(p,\rom{d1},\rom{m0})$ and the returned value is
   $\FALSE$. In this case, $\nxt_p^{M_r}=\bot$ (a precondition of the step), and so $y_n=\bot$. We apply Theorem \ref{T4.2} to the $k$-scan $T$, and find that there is some index $j\in\{\ell_0,\ldots,\ell_n\}$ such that $y_n=\bot$ is $k_p$-connected in $M_j$. So
   $k_p\notin\Set(M_j)$, and the returned value corresponds correctly to
   the choice of $\ell(E)=j$ as the linearization point. All steps of $E$ are set-preserving
   in this case by Theorem \ref{InvReg}.
    
    \item $s\in\Step(p,\rom{d2},\rom{m0})$.  In Figure \ref{fig:protocol-elementary-invariants} we have the control-dependent invariant $\ControlF(p)=d2\rightarrow\LockedP(\nd_p)\wedge
    \KeyF(\nd_p)=k_p$. There are two possibilities:
    \begin{enumerate}
        \item $\returnVal_p(E)=\TRUE$, which implies that
        \begin{equation}
        \label{Eq:Mr-is-search-connected}
        M_r\models \nd_p\neq\bot\wedge \neg\DeletedP(\nd_p)\wedge \neg\RemovedP(\nd_p) \wedge 
        \LockedP(\nd_p,p),
        \end{equation}
        and $M_{r+1}\models\DeletedP(\nd_p)$. As discussed in the proof of \ref{InvReg}, in this case, $\Set(M_{r+1})=\Set(M_r)\setminus\{k_p\}$.
        Since $M_r$ is a regular state, $\nd_p$ is potentially $k_p$-connected there.
        Equation (\ref{Eq:Mr-is-search-connected}) together with Corollary \ref{lemma:pdc-rem-or-locked}
        implies that $\nd_p$ is path-connected. Hence necessarily $\nd_p$ is $k_p$-connected.
        We can take $r+1=\ell(E)$ as the linearization point of $E$. 
        
        \item $\returnVal_p(E)=\FALSE$, and $M_r\models\DeletedP(\nd_p)$.
        The last triple of $T$ is $(\ell_n,x_n,y_n)$, and thus, the last step
        $s=(M_{\ell_n-1},M_\ell)$ before the return step $(M_r, M_{r+1})$ is
        an execution of instruction in $\Step(p,\rom{d1, d2})$.
        As illustrated in Figure \ref{FIGflowchart2}, it follows that $\nd_p=\nxt_p$ has key value $k_p$.
        We use these facts, and define a $k_p$-scan triple for step $s$ as follows: $t=\langle r, x_n, x_n\rangle$. We use this triple to define $T'=T\cdot t$, which is a valid $k_p$-scan, by Definition \ref{Ekpath}. We now apply Theorem \ref{T4.2} to $T'$, and find that there is an index $j\in\{\ell_0,\ldots,r\}$ such that $x_n$ is $k_p$-connected in $M_j$ and $\DeletedP(x_n)^{M_j}\iff\DeletedP(x_n)^{M_{r+1}}$. But we know that $\DeletedP(x_n)^{M_{r+1}}$, and so $\DeletedP(x_n)^{M_j}$. Take $\ell(E)=j$, and conclude that the returned value \FALSE\ is correct
        since $k_p\notin\Set(M_j)$.
    \end{enumerate}
\end{enumerate}
\end{proof}

\begin{proof}[Proof of Ret3]
We must find a linearization point $\ell(E)$ of an \bfinsert\ operation $E$, and prove
that it is correctly related with the value returned by the operation.

By lemma \ref{lemma:inserts-are-kscans}, $E$ induces a valid $k_p$-scan $T=(\langle \ell_0,x_0,y_0\rangle,\ldots,\langle \ell_n, x_n, y_n\rangle)$ such that $M_{\ell_0}\models y_0=\Root$ and so $y_0$ is $k_p$-connected in $M_{\ell_0}$.

Since $E$ is a terminating \bfinsert\ operation, the returning
step of $p$ (the one that follows $T$), is $s=(M_r, M_{r+1})$ (where $r\geq\ell_n+1$), is such that one of the following holds:
\begin{enumerate}
    \item $s\in\Step(p,\rom{i2},\rom{m0})$. In Figure \ref{fig:protocol-elementary-invariants} we have the control-dependent invariant $\ControlF(p)=i2\rightarrow\LockedP(\nd_p)\wedge\KeyF(\nd_p)=k_p$. There are two possibilities:
    \begin{enumerate}
        \item $\returnVal_p(E)=\TRUE$.
        This implies that 
        \begin{equation}
        M_r\models \LockedP(\nd_p,p) \wedge \DeletedP(\nd_p) \wedge 
        \neg\RemovedP(\nd_p),
        \end{equation}
        and $M_{r+1}\models\neg\DeletedP(\nd_p)$. As discussed in the proof of \ref{InvReg}, in this case,\linebreak{}
        $\Set(M_{r+1})=\Set(M_r)\cup\{k_p\}$, but we have to prove that $k_p\notin \Set(M_r)$ in order to
        justify the return value. 
        The displayed equation together with \ref{lemma:pdc-rem-or-locked}
        implies that $\nd_p$ is $k_p$-connected at $M_r$ and $M_{r+1}$,
        so that $k_p\notin \Set(M_{r+1})$,
        and we can take $r+1=\ell(E)$ as the linearization point of $E$.
        \item $\returnVal_p(E)=\FALSE$. Then $M_r\models\neg\DeletedP(\nd_p)$ and so $M_{r+1}\models\neg\DeletedP(\nd_p)$ as well 
        since the return step does not change the \DeletedP\ predicate.
        As illustrated in Figure \ref{FIGflowchart3}, the last
        triple of $T$ $(\ell_n, x_n, y_n)$ is such that
        $x_n=y_n$, that is $\nd_p=\nxt_p$ in $M_{\ell_n}$ and hence
        in $M_r$. We use these facts, and define a $k_p$-scan triple for step $s$ as follows: $t=\langle r, x_n, x_n\rangle$. We use this triple to define $T'=T\cdot t$, which is a valid $k_p$-scan, as
        defined in \ref{Ekpath}. We now apply Theorem \ref{T4.2} to $T'$, and find that there is an index $j\in\{\ell_0,\ldots,r\}$ such that $x_n$ is $k_p$-connected in $M_j$ and $\DeletedP(x_n)^{M_j}\iff\DeletedP(x_n)^{M_{r+1}}$. But we know that $\neg\DeletedP(x_n)^{M_{r+1}}$, and so $\neg\DeletedP(x_n)^{M_j}$.
        Thus $k_p\in \Set(S_j)$ and the return value \FALSE\ is justified
        by the choice of $\ell(E)=j$ as the linearization point of $E$.
    \end{enumerate}
    
    \item $s\in\Step(p,\rom{i3},\rom{m0})$. If the returning step
    is an execution of i3, then $\returnVal_p(E)=\TRUE$, 
    and as part of the proof of Theorem \ref{InvReg}, we showed that $\Set(M_{r+1}) = \Set(M_r)\cup\{k_p\}$. In order to prove that the returned
    value is appropriate, we have to show that $k_p\notin\Set(M_r)$,
    and for this we will show that $\bot$ is $k_p$-connected at $M_r$.
    Since $\ControlF(p)=\rom{i3}$ in $M_r$, $\KeyF(\nd_p)\neq k_p$.
    Since $\ControlF(p)=\rom{m0}$ in $M_{r+1}$, $\LR(\nd_p,k_p< \KeyF(\nd_p))=\bot$.
    By the control-dependent invariant of Figure \ref{fig:protocol-elementary-invariants}, we know that $\LockedP(\nd_p,p)$
    Since $M_r$ is regular, $\nd_p$ is potentially $k_p$-connected. 
    Since $\nd_p\arrow_{k_p} \bot$, and since $\nd_p$ is not removed at $M_r$ (see \ref{inv:rem-chld-not-bot}),
    it is not the case that PT3$(\nd_p,k_p)$.
    By the control-dependent invariant of Figure \ref{fig:protocol-elementary-invariants}, we know that $\LockedP(\nd_p,p)$, 
    which excludes the possibility that PT2$(\nd_p,k_p)$.
    Thus PT1$(\nd_p,k_p)$, and so $\nd_p\arrow_{k_p}\bot$ implies immediately that $\bot$ is $k_p$-connected.
\end{enumerate}
\end{proof}

\begin{proof}[Proof of Ret4]
As part of the proof of Theorem \ref{InvReg}, we proved that the steps performed by the system process are set-preserving. Thus, Property 4 holds.
\end{proof}

This concludes the correctness proof of the CF algorithm.

\section{Related Work}
\label{subsec:related-work}
O'Hearn et al.~\cite{ohearn2010hindsight} described a proof framework for linked-list-based concurrent set algorithms. They used it to prove the correctness of the Lazy Set algorithm of to Heller et al.~\cite{heller2007lazy}. Similar to our approach, theirs is two-tiered: they first formulate invariants and step-properties specific to the algorithm in question (some of which match some of our invariants and step-properties). They then formulate and prove the Hindsight Lemma in terms of these invariants and properties, which they use to prove the linearizability of the Lazy Set. As in our case, that lemma is formulated in a way that is abstracted away from the technical details of the algorithm they analyze.

In this chapter, we discussed our development of a framework for analyzing the behavior of BST-backed sets, which requires an approach different from that for the analysis of linked-list-backed sets, due to the differing constraints and more complex behavior of BSTs.

Feldman et al.~\cite{Feldman20} presented a general framework for proving the correctness of concurrent tree- and list-based implementations of the set data structure. Their framework is based on temporal predicates on instructions of the operations of the implementation under inspection. A temporal predicate $\phi$ is said to hold at instruction $i$ of operation $o$ if $\phi$ holds \emph{at some moment} between the time $t$ at which $o$ was invoked and the time $t'$ when $i$ is executed. They use the convention $\DDiamond^{t'}_t(\phi)$ to represent such temporal predicates.

Feldman et al.~\cite{Feldman20} used of their framework to prove the correctness of multiple concurrent tree-based implementations of the set data structure. Among these implementations is a variant of the CF algorithm of Crain et al.~\cite{Crain13,Crain16}. While the variant they proved is very similar to the original algorithm, there is one major behavioral difference in the \bfinsert\ and \bfdelete\ operations: it is possible that when the traversing process reaches a node $x$ where a logical deletion or insertion should take place, $x$ is found to be physically removed (by way of the \rem\ flag). In this case, the original algorithm continues the traversal process from $x$, relying on the clever backtracking mechanism. On the other hand, the variant that Feldman et al. analyzed restarts the traversal process from the \Root\ node all over again.

This difference stems from a condition of Feldman et al.'s proof framework, which they call the \emph{forepassed condition with respect to field $f$}. Intuitively, this condition requires that writes that may interfere with a concurrent traversal do not change the field $f$ of any node $x$ after $x$ has been disconnected, if $x$ was disconnected during said traversal.

Core to the proof of Feldman et al. is the past temporal logic predicate $\DDiamond(\Root\arrow^*_k y \wedge y.\key = k \wedge \neg y.\rem)$, which can be found in the assertion annotations of the \bfdelete\ and \bfinsert\ operations in the Appendix of~\cite{Feldman20}. To prove the correctness of this predicate, it is required for the \emph{forepassed condition with respect to \rem} to hold. However, in the original algorithm, the \bfremove, \rotateLeft\, and \rotateRight\ operations first disconnect the node \n, and only then do they modify \n.\rem, thus violating this condition.

In this chapter, we proved the correctness of the original form of the algorithm, including its full backtracking mechanism.

\section{Conclusion}
In this chapter, we formally proved the correctness of the contention-friendly algorithm of Crain et al.~\cite{Crain13,Crain16}.
To our knowledge, this is the first time this has been done for the original algorithm of Crain et al., which includes its full backtracking mechanism.

To facilitate the proof, we presented the abstract notions of ``potential-connectivity'', ``regularity'', and ``abstract scanning''.
We believe that these notions constitute a general framework for proving the correctness of concurrent BST algorithms.
We intend to explore this belief in the future by attempting to apply the tools developed here to other BST algorithms, 
such as the Logically-Ordered tree algorithm~\cite{drachsler2014lotree} and the Citrus tree algorithm~\cite{arbel2014citrus}.

We supplemented our proof with a bounded model of the algorithm, encoded in TLA+,
and verified the various invariants and properties in Section~\ref{sec:invariants-and-props} against that model~\cite{thesis-code}.
While not a full verification of our proofs (due to the bounded nature of the model), 
this model-checking process does act to validate the correctness of our proofs.

Our methodology of using model-checking to validate our manual work proved useful, helping us find and correct multiple minor issues. In addition, as detailed in footnote~\ref{footnote:pt2} of definition~\ref{DefPot}, TLA+ flagged a serious problem in this definition, which we corrected with the help of the problematic scenario provided by the model-checker as a counter example.

\bibliographystyle{abbrv}
\bibliography{main}

\newpage
\appendix
\section{Steps of the Contention Friendly Algorithm}
\label{Sec4.1}

In this appendix we precisely define each of the steps of the contention friendly protocol presented in figures \ref{LM} and \ref{RLred}. 
For the sake of brevity, in the definition of a step $(M,N)$, interpretations of the constituents of $N$ are assumed to be identical to those of $M$, 
unless explicitly stated otherwise, and we list those variables whose denotations may change in {\em `` Possible changes''}.

We begin by describing the steps of process $p>0$ invoking one of its three operations:

Invocations of $\bfcontains(k)$ (respectively $\bfdelete(k)$ and $\bfinsert(k)$), where $k\in\omega$,
is the set of all steps $\Step(p,\rom{m0},\rom{c1})$ such that the following hold.
$\ControlF^M(p)=\rom{m0}$, $\ControlF^N(p)=\rom{c1}$ (respectively
$\ControlF^N(p)=\rom{d1}$ and $\ControlF^N(p)=\rom{i1}$), $k_p^N=k$, $\nd_p^N=\Root$,
and $\nxt_p^N=\Root$.

\paragraph*{Steps of process $p>0$ executing its $\bfcontains(k_p)$ for $k_p\in\omega$.}

\begin{enumerate}
    \item 
    $(M,N) \in \Step(p,\rom{c1})$ if and only if $\ControlF^M(p)=\rom{c1}$ and
    one of the following three possibilities occurs:
    \begin{enumerate}
        \item
        $\nxt_p^M=\bot \wedge  \returnVal_p^N=\FALSE\wedge \ControlF(p)=\rom{m0}$.
        
        \item $\nxt_p^M\neq\bot \wedge \nd_p^N=\nxt_p^M\wedge k_p=\KeyF(\nd_p^N)
        \wedge \ControlF^N(p)=\rom{c2}$.
        
        \item 
        $\nxt_p^M\neq\bot \wedge \nd_p^N=\nxt_p^M\wedge k_p\neq \KeyF(\nd_p^N)
        \wedge \ControlF^N(p)=\rom{c1}\wedge(\nd_p\arrow_{k_p} \nxt_p)^N$.
    \end{enumerate}
    \possible\ $\ControlF(p), \nd_p, \nxt_p$.
    
    \item
    $(M,N)\in \Step(p,\rom{c2})$ if and only if:
    \begin{equation*}
        \ControlF^M(p)=\rom{c2} \wedge \ControlF^N(p)=\rom{m0}\wedge\returnVal_p^N \equiv (\neg\DeletedP(\nd_p))^M
    \end{equation*}
    \possible\ $\ControlF(p), \returnVal_p$.
\end{enumerate}

\paragraph*{Steps of process $p>0$ executing its $\bfdelete(k_p)$  for $k_p\in\omega$}.

\begin{enumerate}
    \item 
    $(M,N) \in \Step(p,\rom{d1})$ if and only if $\ControlF^M(p)=\rom{d1}$ and
    one of the following three possibilities occurs:
    \begin{enumerate}
        \item
        $\nxt_p^M=\bot \wedge \ControlF^N=\rom{m0}\wedge \returnVal_p^N=\FALSE$.
        
        \item $\nxt_p^M\neq\bot \wedge \nd_p^N=\nxt_p^M\wedge k=\KeyF(\nd_p^N)
        \wedge \ControlF^N(p)=\rom{d2}$.
        
        \item 
        $\nxt_p^M\neq\bot \wedge \nd_p^N=\nxt_p^M\wedge k\neq \KeyF(\nd_p^N)
        \wedge \ControlF^N(p)=\rom{d1}\wedge(\nd_p\arrow_{k_p} \nxt_p)^N$.
    \end{enumerate}
    \possible\ $\ControlF(p), \nd_p, \nxt_p$.
    
    \item $(M,N) \in \Step(p,\rom{d2})$ if and only if $\ControlF^M(p)=\rom{d2}$ and one of the following three possibilities occurs:
    \begin{enumerate}
        \item $\DeletedP(\nd)^M\rightarrow \ControlF^N(p)=\rom{m0}\wedge \returnVal_p^N=\FALSE$.
        \item $\neg\DeletedP(\nd_p)^M\wedge\RemovedP(\nd_p)^M\rightarrow \nxt_p^N=\RightF(\nd_p^M)\wedge \ControlF^N(p)=\rom{d1}$.
        \item $\neg\DeletedP(\nd_p)^M\wedge\neg\RemovedP(\nd_p)^M\rightarrow \DeletedP^N(\nd_p^M)\wedge \ControlF^N(p)=\rom{m0}\wedge\\ \returnVal_p^N=\TRUE$.
    \end{enumerate}
    \possible\ $\ControlF(p), \nxt_p, \DeletedP$.
\end{enumerate}

\paragraph*{Steps of process $p>0$ executing its $\bfinsert(k_p)$  for $k_p\in\omega$.}

\begin{enumerate}
    \item $(M,N) \in \Step(p,\rom{i1})$ if and only if $\ControlF^M(p)=\rom{i1}$ and
    one of the following three possibilities occurs:
    \begin{enumerate}
        \item
        $\nxt_p^M=\bot \wedge \ControlF^N(p)=\rom{i3}$.
        
        \item $\nxt_p^M\neq\bot \wedge \nd_p^N=\nxt_p^M\wedge k_p=\KeyF(\nd_p^N)
        \wedge \ControlF^N(p)=\rom{i2}$.
        
        \item 
        $\nxt_p^M\neq\bot \wedge \nd_p^N=\nxt_p^M\wedge k_p\neq \KeyF(\nd_p^N)
        \wedge \ControlF^N(p)=\rom{i1}\wedge(\nd_p\arrow_{k_p} \nxt_p)^N$.
    \end{enumerate}
    \possible\ $\ControlF(p), \nd_p, \nxt_p$.
    
    \item $(M,N) \in \Step(p,\rom{i2})$ if and only if $\ControlF^M(p)=\rom{i2}$ and one of the following three possibilities occurs:
    \begin{enumerate}
        \item $\neg\DeletedP(\nd)^M\rightarrow \ControlF^N(p)=\rom{m0}\wedge \returnVal_p^N=\FALSE$.
        \item $\DeletedP(\nd_p)^M\wedge\RemovedP(\nd_p)^M\rightarrow \nxt_p^N=\RightF(\nd_p)^M\wedge \ControlF^N(p)=\rom{i1}$.
        \item $\DeletedP(\nd_p)^M\wedge\neg\RemovedP(\nd_p)^M\rightarrow \neg\DeletedP^N(\nd_p^M)\wedge \ControlF^N(p)=\rom{m0}\wedge\\ \returnVal_p^N=\TRUE$.
    \end{enumerate}
    \possible\ $\ControlF(p), \nxt_p, \DeletedP$.
    
    \item $(M,N) \in \Step(p,\rom{i3})$ if and only if $\ControlF^M(p)=\rom{i3}$ and one of the following three possibilities occurs:
    \begin{enumerate}
        \item $(\nd_p\notarrow_{k_p}\bot)^M\rightarrow(\nd_p\arrow_{k_p}\nxt_p)^N\wedge\ControlF^N(p)=\rom{i1}$.
        \item $(\nd_p\arrow_{k_p}\bot)^M\rightarrow\exists\new\in\AddressT^N:$
        \begin{equation*}
            \begin{split}
            & \AddressT^N\setminus\AddressT^M=\{\new\}\wedge(\nd_p\arrow_{k_p}\new)^N\wedge \\
            & \KeyF(\new)=k_p\wedge\LeftF^N(\new)=\RightF^N(\new)=\bot\wedge\\
            & \neg\DeletedP^N(\new)\wedge\neg\RemovedP^N(\new)
            \end{split}
        \end{equation*}
    \end{enumerate}
    \possible\ $\ControlF(p), \nd_p, \nxt_p, \LeftF, \RightF, \KeyF, \AddressT$.
\end{enumerate}

Next, we tackle the steps of process $\Sys$ (i.e. $p=0$) invoking one of its three operations:

\begin{enumerate}
    \item Invocations of $\rotateLeft(\prt,\lft)$ is the set $\Step(\Sys,\rom{m0},\rom{f6})$ of steps
    $(M,N)$ such that $\ControlF^M(\Sys)=\rom{m0}\wedge \ControlF^N(\Sys)=\rom{f5}$, and all the prerequisites
    pr1--pr5 hold at state $N$. 
    
    \item Invocations of $\rotateRight(\prt,\lft)$ is the set defined symmetrically with r6 replacing f6. 
    
    \item Invocations of $\bfremove(\prt,\lft)$ is defined similarly.
\end{enumerate}

\paragraph*{Steps of process $\Sys$ executing its $\rotateLeft(\prt,\lft)$.}

\begin{enumerate}
    \item $(M,N)\in \Step(\Sys,\rom{f6})$ if and only if:
    \begin{equation*}
        \begin{split}
        & \ControlF^M(\Sys)=\rom{f6}\wedge \ControlF^N(\Sys)=\rom{f7} \wedge\\
        & \exists a \in \AddressT^N:\AddressT^N\setminus\AddressT^M=\{a\} \wedge \\
        & \LeftF^N(a) = \ell_0^M \wedge \RightF^N(a) = \rl^M \wedge \KeyF^N(a) = \KeyF^M(\n) \wedge \\
        &\DeletedP^N(a) = \DeletedP^M(\n) \wedge \LeftF^N(\r^N) = a \\
        \end{split}
    \end{equation*}
    \possible\ $\ControlF(\Sys), \LeftF, \RightF, \KeyF, \AddressT$.

    \item $(M,N)\in \Step(\Sys,\rom{f7})$ if and only if:
        \begin{equation*}
        \begin{split}
        & \ControlF^M(\Sys)=\rom{f7}\wedge \ControlF^N(\Sys)=\rom{f8} \wedge\\
        & \LeftF^N(\n^N) = \r^M
        \end{split}
    \end{equation*}
    \possible\ $\ControlF(\Sys), \LeftF$.

    \item $(M,N)\in \Step(\Sys,\rom{f8})$ if and only if:
        \begin{equation*}
        \begin{split}
        & \ControlF^M(\Sys)=\rom{f8}\wedge \ControlF^N(\Sys)=\rom{f9} \wedge\\
        & \lft^M = \TRUE \rightarrow \LeftF^N(\prt^N) = \r^M \wedge \\
        & \lft^M = \FALSE \rightarrow \RightF^N(\prt^N) = \r^M
        \end{split}
    \end{equation*}
    \possible\ $\ControlF(\Sys), \LeftF, \RightF$.

    \item $(M,N)\in \Step(\Sys,\rom{f9})$ if and only if:
        \begin{equation*}
        \begin{split}
        & \ControlF^M(\Sys)=\rom{f9}\wedge \ControlF^N(\Sys)=\rom{m0} \wedge\\
        & \RemovedP^N(\n^N) \wedge \neg\LockedP^N(\n^N, \Sys)
        \end{split}
    \end{equation*}
    \possible\ $\ControlF(\Sys), \RemovedP$.
\end{enumerate}

\paragraph*{Steps of process $\Sys$ executing its $\rotateRight(\prt,\lft)$.}

\begin{enumerate}
    \item $(M,N)\in \Step(\Sys,\rom{r6})$ if and only if:
    \begin{equation*}
        \begin{split}
        & \ControlF^M(\Sys)=\rom{r6}\wedge \ControlF^N(\Sys)=\rom{r7} \wedge\\
        & \exists a \in \AddressT^N:\AddressT^N\setminus\AddressT^M=\{a\} \wedge \\
        & \LeftF^N(a) = \lr^M \wedge \RightF^N(a) = \r^M \wedge \KeyF^N(a) = \KeyF^M(\n) \wedge \\
        &\DeletedP^N(a) = \DeletedP^M(\n) \wedge \RightF^N(\ell_0^N) = a \\
        \end{split}
    \end{equation*}
    \possible\ $\ControlF(\Sys), \LeftF, \RightF, \KeyF, \AddressT$.

    \item $(M,N)\in \Step(\Sys,\rom{r7})$ if and only if:
    \begin{equation*}
        \begin{split}
        & \ControlF^M(\Sys)=\rom{r7}\wedge \ControlF^N(\Sys)=\rom{r8} \wedge\\
        & \RightF^N(\n^N) = \ell_0^M
        \end{split}
    \end{equation*}
    \possible\ $\ControlF(\Sys), \RightF$.

    \item $(M,N)\in \Step(\Sys,\rom{r8})$ if and only if:
    \begin{equation*}
        \begin{split}
        & \ControlF^M(\Sys)=\rom{r8}\wedge \ControlF^N(\Sys)=\rom{r9} \wedge\\
        & \lft^M = \TRUE \rightarrow \LeftF^N(\prt^N) = \ell_0^M \wedge \\
        & \lft^M = \FALSE \rightarrow \RightF^N(\prt^N) = \ell_0^M
        \end{split}
    \end{equation*}
    \possible\ $\ControlF(\Sys), \LeftF, \RightF$.

    \item $(M,N)\in \Step(\Sys,\rom{r9})$ if and only if:
    \begin{equation*}
        \begin{split}
        & \ControlF^M(\Sys)=\rom{r9}\wedge \ControlF^N(\Sys)=\rom{m0} \wedge\\
        & \RemovedP^N(\n^N) \wedge \neg\LockedP^N(\n^N, \Sys)
        \end{split}
    \end{equation*}
    \possible\ $\ControlF(\Sys), \RemovedP$.
\end{enumerate}

\paragraph*{Steps of process $\Sys$ executing its $\bfremove(\prt,\lft)$.}

\begin{enumerate}

    \item $(M,N)\in \Step(\Sys,\rom{v5})$ if and only if:
        \begin{equation*}
        \begin{split}
        & \ControlF^M(\Sys)=\rom{v5}\wedge \ControlF^N(\Sys)=\rom{v6} \wedge \\
        & \lft^M = \TRUE \rightarrow \LeftF^N(\prt^N) = \child^M \wedge \\
        & \lft^M = \FALSE \rightarrow \RightF^N(\prt^N) = \child^M
        \end{split}
    \end{equation*}
    \possible\ $\ControlF(\Sys), \LeftF, \RightF$.

    \item $(M,N)\in \Step(\Sys,\rom{v6})$ if and only if:
    \begin{equation*}
        \begin{split}
        & \ControlF^M(\Sys)=\rom{v6}\wedge \ControlF^N(\Sys)=\rom{v7} \wedge \\
        & \LeftF^M(\n^M) = \bot \rightarrow \LeftF^N(\n^N) = \prt^M \wedge \\
        & \LeftF^M(\n^M) \neq \bot \rightarrow \RightF^N(\n^N) = \prt^M
        \end{split}
    \end{equation*}
    \possible\ $\ControlF(\Sys), \LeftF, \RightF$.

    \item $(M,N)\in \Step(\Sys,\rom{v7})$ if and only if:
    \begin{equation*}
        \begin{split}
        & \ControlF^M(\Sys)=\rom{v7}\wedge \ControlF^N(\Sys)=\rom{v8} \wedge \\
        & \LeftF^M(\n^M) = \prt^M \rightarrow \RightF^N(\n^N) = \prt^M \wedge \\
        & \LeftF^M(\n^M) \neq \prt^M \rightarrow \LeftF^N(\n^N) = \prt^M
        \end{split}
    \end{equation*}
    \possible\ $\ControlF(\Sys), \LeftF, \RightF$.

    \item $(M,N)\in \Step(\Sys,\rom{v8})$ if and only if:
        \begin{equation*}
        \begin{split}
        & \ControlF^M(\Sys)=\rom{v8}\wedge \ControlF^N(\Sys)=\rom{m0} \wedge\\
        & \RemovedP^N(\n^N) \wedge \neg\LockedP^N(\n^N, \Sys)
        \end{split}
    \end{equation*}
    \possible\ $\ControlF(\Sys), \RemovedP$.
\end{enumerate}

\end{document}